\documentclass{article}
\usepackage{amsfonts,color,xcolor,graphicx,framed,amsthm,amsmath,hyperref,mathtools,wrapfig}
\usepackage{amssymb,latexsym,amsmath,amsfonts,graphicx, ebezier,authblk,float}
\numberwithin{equation}{section}

\numberwithin{figure}{section}
\def\iint{\int\!\!\!\!\int}
\def\d{{\rm d}}
\def\C{\ensuremath{\mathbb C}}

\newcommand{\be}{\begin{equation}}
\newcommand{\ee}{\end{equation}}
\newcommand{\bea}{\begin{eqnarray}}
\newcommand{\eea}{\end{eqnarray}}

\newcommand{\supp}{\operatorname{supp}}

\newtheorem{theorem}{Theorem}[section]
\newtheorem{definition}[theorem]{Definition}
\newtheorem{proposition}[theorem]{Proposition}
\newtheorem{lemma}[theorem]{Lemma}

\newtheorem{remark}[theorem]{Remark}
\renewenvironment{proof}[1][Proof]{\noindent {\bf #1.}\ }{\hfill {\bf Q.E.D.}}

\newcommand{\resi[1]}{\operatornamewithlimits{Res}_{#1}}
\newcommand{\order}[1]{\ensuremath{{\mathcal O}\left(#1\right)}}
\renewcommand{\c}[1]{\ensuremath{\overline{#1}}}
\renewcommand{\Re}{\operatorname{Re}}
\renewcommand{\Im}{\operatorname{Im}}
\newcommand{\U}{{\mathcal U}}
\newcommand{\inte}{\operatornamewithlimits{Int}}
\newcommand{\exte}{\operatornamewithlimits{Ext}}

\newcommand{\poly}{\pi}
\newcommand{\lb}{\lambda}

\def\XXint#1#2#3{{\setbox0=\hbox{$#1{#2#3}{\int}$}
     \vcenter{\hbox{$#2#3$}}\kern-.5\wd0}}

\def\Zc{{\mathcal Z}}
\def\Nc{{\mathcal N}}
\def\Tr{\operatorname{Tr}}
\definecolor{shadecolor}{rgb}{0.95, 0.95, 0.86}
\newcommand{\semmi}[1]{}

\newcommand{\tfigure}[4]{
\begin{figure}[H]
\centering
\includegraphics[#4]{#1}
\caption{#2}
\label{#3}
\end{figure}}

\mathtoolsset{showonlyrefs=true}

\title{Orthogonal polynomials for a class of measures with discrete rotational symmetries in the complex plane}
\author[1,2]{F. Balogh}
\author[1,3]{T. Grava}
\author[1]{D. Merzi}
\affil[1]{Scuola Internazionale Superiore di Studi Avanzati, Trieste}
\affil[2]{Concordia University, Montr\'eal}
\affil[3]{School of Mathematics, University of Bristol}
\date{ Correspondence to be sent to grava@sissa.it}
\begin{document}
\maketitle

\begin{abstract} We obtain  the strong asymptotics of polynomials $p_n(\lb)$, $\lb\in\mathbb{C}$, orthogonal 
  with respect to measures in the complex plane of the form
$$
e^{-N(|\lb|^{2s}-t\lb^s-\c{t}\c{\lb}^s)}dA(\lb),
$$
 where $s$ is a positive integer, $t$ is a complex parameter and $dA$ stands for the area measure in the plane.
Such problem has its origin from normal matrix models.  We  study the asymptotic  behaviour of $p_n(\lb)$ in the limit  $n,N\to\infty$ in such a way that $n/N\to T$ constant.
Such asymptotic behaviour has two distinguished regimes according to the  topology of the limiting support of the  eigenvalue distribution of the normal matrix model. If  $0<|t|^2<T/s$,  the eigenvalue distribution
 support is a simply connected compact set of the complex plane, while  for $|t|^2>T/s$ the eigenvalue distribution support consists of $s$ connected components. Correspondingly the support of the limiting zero distribution of the orthogonal polynomials consists of a closed contour contained in each connected component.
 
Our asymptotic analysis is obtained by reducing the planar orthogonality conditions  of the polynomials to  an equivalent system of contour integral orthogonality conditions. The  strong asymptotics for the orthogonal polynomials is obtained from the corresponding Riemann--Hilbert problem by  the Deift--Zhou nonlinear steepest descent method.

\end{abstract}
 \tableofcontents
 
 \section{Introduction}
 We study the asymptotics of orthogonal polynomials with respect to a family of measures
 supported on the whole complex plane. To set up the notation for the general case, let $p_n(\lb)$ denote the monic  orthogonal polynomials of degree $n$  such that \\
\be
 \label{eq:op_original}
 \int_{\C}p_{n}(\lb)\overline{p_m(\lb)} e^{-NW(\lb)}dA(\lb) = h_{n,N}\delta_{nm},\quad n,m=0,1,2,\dots
 \ee
 where $W:\mathbb{C}\to \mathbb{R}$ is called the external potential, $N$ is a positive parameter\footnote{For simplicity we use the simpler notation $p_n(\lb)$ instead of $p_{n,N}(\lb)$, even though the orthogonal polynomials depend on the value of the scaling parameter $N$.}, $dA(\lb)$ is the area measure in the complex plane and $h_{n,N}$ is the norming constant.
 The external potential is assumed to have sufficient growth at infinity so that the integrals in \eqref{eq:op_original} are bounded.

 Planar orthogonal polynomials satisfying \eqref{eq:op_original} appear naturally in the context of normal matrix models \cite{ChauZaboronsky} where one studies  probability distributions of the form
 \be
 \label{random_matrix_density}
 M \mapsto \frac{1}{\Zc_{n,N}}e^{-N\Tr(W(M))}dM, \qquad \Zc_{n,N}=\int_{\Nc_{n}}e^{-N\Tr(W(M))}dM,
 \ee
 where $\Nc_N$ is the algebraic variety of $n\times n$ normal matrices
 \be
 \Nc_n=\left\{M\ \colon\ [M, M^{\star}] =0 \right\}\subset \operatorname{Mat}_{n\times n}(\C),
 \ee
 and $dM$ is the volume form induced on ${\mathcal N}_n$ which is invariant under conjugation by unitary matrices. Since normal matrices are diagonalizable by unitary transformations,
 the probability density \eqref{random_matrix_density} can be reduced to the form \cite{Mehta}
 \[
 \dfrac{1}{Z_{n,N}} \prod_{i<j}|\lambda_i-\lambda_j|^2 e^{-N\sum_{j=1}^{n}\Tr{W(\lambda_i)}}dA(\lambda_1)\cdots dA(\lambda_n), 
\]
where $\lb_j$ are the complex eigenvalues of the normal matrix $M$ 
 and the normalizing factor $Z_{n,N}$, called partition function, is given by 
 \[
 Z_{n,N}=\int_{\mathbb{C}^n}\prod_{i<j}|\lambda_i-\lambda_j|^2 e^{-N\sum_{j=1}^{n}W(\lambda_i)}dA(\lambda_1)\cdots dA(\lambda_n). 
 \]
 The statistical quantities related to eigenvalues can be expressed in terms of  the orthogonal polynomials $p_n(\lb)$ defined in \eqref{eq:op_original}. In particular, the average density of eigenvalues is 
 \be
 \rho_{n,N}(\lb) = \frac{1}{n}e^{-NW(\lb)}\sum_{j=0}^{n-1}\frac{1}{h_{j}}\left|p_j(\lb)\right|^2.
 \ee
 The partition function $Z_{n,N}$ can be written as a product  of the normalizing constants
 \be
 Z_{n,N}= \prod_{j=0}^{n}h_{j,N}.
 \ee
While the asymptotic density of eigenvalues can be studied using an approach from potential theory \cite{SaffTotik}, the zero distribution of orthogonal polynomials remains an open issue for general potential weights despite  general results in \cite{StahlTotik}.
 The density of eigenvalues $\rho_{n,N}(\lambda)$ converges (in the sense of measures) in the limit
  \be
 \label{limit}
 n \to \infty\ ,\quad N\to \infty\ ,\quad \frac{N}{n} \to \frac{1}{T}\ ,
 \ee 
to the unique probability measure $\mu^*$ in the plane which minimizes the functional \cite{ElbauFelder, HedMak}
 \begin{equation}
 \label{def_I}
  I(\mu)=\int\!\!\!\!\int\log|\lambda-\eta|^{-1}d\mu(\lb)d\mu(\eta)+\frac{1}{T}\int W(\lb)d\mu(\lb).
 \end{equation}
 The functional $I(\mu)$ in \eqref{def_I} is the Coulomb energy functional in two dimensions and  the existence of a unique minimizer is a well-established fact under mild assumptions on the potential $W(\lb)$ \cite{SaffTotik}. If $W$ is twice continuously differentiable and its Laplacian $\Delta W$ is non-negative, the equilibrium measure is given by 
  \[
 d\mu^*(\lambda)=\Delta W(\lambda)\chi_{D}(\lambda)dA(\lambda),
 \]
 where $\chi_{D}$ is the characteristic function of the compact support set $D = \supp(\mu^{*})$. 
   Sub-leading order corrections to the behaviour  of the eigenvalues distribution  $ \rho_{n,N}(\lb)$ as $n,N\to\infty$  and fluctuations  in the bulk and at the boundary of the support $D$ have been considered in \cite{Berman, AmHeMa,LeeRiser,Mak1,Mak2}.

%
 The measure $\mu^*$ can be also uniquely characterized by the Euler--Lagrange conditions
 \be
 \label{E-L}
\dfrac{W(\lb)}{T}-2\int\log|\lambda-\eta|d\mu^*(\eta)+\ell_{D}\geq 0
\ee
for all values $\lb\in\mathbb{C}$ with equality in \eqref{E-L} on the support\footnote{To be precise, the equality on the support is valid only up to a set of capacity zero \cite{SaffTotik}.} of $\mu^*$.
The Lagrange multiplier $\ell_{D}$ is called the \emph{(generalized) Robin constant}.
 It is a  non-trivial problem to  determine the  shape of the support set $D$. In some cases this problem  is called {\it Laplacian growth}. When the potential $W(\lb)$ is real analytic the boundary $\partial D$ is a finite union of analytic arcs with at most a finite number of singularities \cite{Sakai}, see also \cite{HedMak}. 

 There is only a handful of potentials $W(\lb)$ for which the polynomials $p_n(\lb)$ can be explicitly computed.
 The simplest example is $W(\lb)=|\lb|^2$ for which the orthogonal polynomials $p_n(\lb)$ are monomials of degree $n$, the constants $h_{n,N}$ and the average density of eigenvalues 
 $\rho_{n,N}$ can be computed explicitly in terms of the Gamma function. The matrix model associated to this potential is known as the \emph{Ginibre ensemble} \cite{Ginibre,Girko} and the density of eigenvalues 
 $\rho_{n,N}$ converges 
  to the normalized area measure on the disk of radius $\sqrt{T}$ centered at the origin.
  In general, for radially symmetric potentials $W=W(|\lb|)$, the orthogonal polynomials are always monomials and in the limit \eqref{limit}, the eigenvalue distribution is supported  either on a disk or an annulus by  the \emph{Single-ring Theorem} of Feinberg and Zee \cite{FeinbergZee}, whose rigorous proof can be found in \cite{Guionnet}. The  correlation functions in this case have been studied in \cite{ChauZaboronsky}.  
  
The harmonic deformation of the Gaussian case
$W(\lb)=|\lb|^2+(\mbox{harmonic})$ has been intensively studied. In particular the  potential  $W(\lb)=  |\lb|^2- t(\lb^2+\bar{\lb^2})$ is associated to the Hermite polynomials for $|t| < 1/2$. In the limit \eqref{limit} the distribution of eigenvalues is the normalized area measure on an ellipse, while the distribution of the zeros of the orthogonal polynomials is given by the (rescaled) Wigner semicircle law with support between the two foci of the ellipse \cite{Difrancesco}.

The normal matrix model  with a general deformation $W(\lb)=|\lb|^2+\Re(P(\lb))$ where $P(\lb)$ is a polynomial of a fixed degree has first been considered in the seminal papers 
  \cite{MWZ} and \cite{WiegZab1} (see  also the review article \cite{Zab}), where the connection with the Hele--Shaw problem and integrable structure in conformal dynamic  has been pointed out. More general potentials  have been considered later in \cite{Zab13}.
   For such potentials, however, the matrix integrals have convergence issues in the complex plane and therefore a natural cut-off has been introduced in the work of Elbau and Felder \cite{ElbauFelder}. In \cite{Elbau} the polynomials associated to such deformed Gaussian potentials have been studied and it was argued that the Cauchy transform of the limiting zero distribution  of the orthogonal polynomials coincides with the  Cauchy transform of the   limiting eigenvalue distribution of the matrix model outside the support of the eigenvalues.
  Moreover, it was also conjectured in \cite{Elbau} that the zero distribution of the polynomials $p_n(\lb)$  is supported on tree-like segments (the mother body, see definition below) inside the compact set (the droplet) that attracts the eigenvalues of the normal matrix model.
  
  For the  external potential $W(\lb)=|\lb|^2+\Re(t\lb^3)$,  Bleher and Kuijlaars \cite{BleherKuij} defined polynomials orthogonal with respect to  a system of unbounded contours on the complex plane, without any cut-off and  which satisfy the same recurrence relation that is asymptotically valid for the orthogonal polynomials of Elbau and Felder. They then study the asymptotic distribution of the zeros of such polynomials confirming the predictions of \cite{Elbau}. Similar results were obtained for the more general external potential \cite{Kuij2} $ W(\lb)=|\lb|^2-\Re(t|\lb|^k)$, $k\geq 2$ and $|t|$ sufficiently small so that  the eigenvalue distribution of the matrix model has  an analytic simply connected support.  Cases in which the eigenvalue support has singularities were analyzed in \cite{KuiTov} and \cite{BBLM}. In particular in the work \cite{BBLM}
  the external potential $W(\lb)=|\lb|^2-2c\log|\lb-a|$ with $c$ and $a$ positive constants, has been studied and the strong asymptotics of the corresponding orthogonal polynomials has been derived both in the case in which the support of the eigenvalues distribution is simply connected (pre-critical case) or multiply connected (post-critical case) and critical transition was observed (see also \cite{Teodorescu, BH}). We remark that in the work \cite{BBLM},  differently form the previous works, the zeros of the orthogonal polynomials do not  always accumulate on a curve that corresponds to the mother-body  of the  domain where the eigenvalues of the normal matrix models are distributed. 
    
 In this work we study the strong asymptotic of the polynomials $p_n(\lb)$ orthogonal with respect to a density $e^{-NW(\lb)}$ where the external potential is of the form
 $
W(\lb)=|\lambda|^{2s}-t\lb^s-\bar t\bar{\lb}^{s} $ with $ \lb \in \C,$ $s$ a positive integer and $ t\in\mathbb{C}\setminus\{0\}.$
By a simple rotation of the variable $\lb$ the analysis can be reduced to the case of real and positive $t$.
Therefore, without loss of generality, we may and do assume that $t \in \mathbb{R}_{+}$, that is,
 \be
 \label{def_W}
W(\lb)=|\lambda|^{2s}-t(\lb^s+\bar{\lb}^{s}) \qquad \lb \in \C,\;s\in \mathbb{N},\quad t>0,
 \ee
 and the associated orthogonality measure is
 \be
 \label{eq:ortho_measure}
e^{-N\left(|\lambda|^{2s}-t(\lb^s+\bar{\lb}^{s})\right)}dA(\lambda) \qquad \lb \in \C,\;s\in \mathbb{N},\quad t>0.
 \ee
Note that the potential $W(\lb)$ has a discrete rotational $\mathbb{Z}_s$-symmetry. It was observed in  \cite{BM} (see also \cite{Etingof_Ma}) that  if a potential $W(\lambda)$ can be written in the form 
\[
W(\lambda)=\dfrac{1}{s}Q(\lambda^s),
\]
the equilibrium measure for $W$ can be obtained from the equilibrium measure of $Q$ by an unfolding procedure. In our particular case 
\[
Q(u)=s|u|^2-st (u+\bar{u})=s|u-t|^2-st^2
\]
corresponds to the Ginibre ensemble 
so that the equilibrium measure for the potential $Q$ is the normalized area measure
of  the disk
\[
|u-t|=t_c,\quad t_c=\sqrt{\dfrac{T}{s}},
\]
where $T$ has been defined in \eqref{limit}.
The equilibrium measure for $W$ turns out the be equal to 
 \begin{equation}
 \label{dmustar}
 d\mu^*(\lb)=\dfrac{s}{\pi t_c^2}|\lambda|^{2(s-1)}\chi_{D}(\lb)dA(\lb),
 \end{equation}
where $\chi_D$ is the characteristic function of the support set
 \begin{equation}
 \label{domainmu}
D:=\{\lb\in\mathbb{C},\,\,|\lambda^s-t|\leq  t_c\}.
\end{equation}

We observe that for $t<t_c$ the equation \eqref{domainmu} describes a 
simply connected domain in the complex plane with uniformizing  map from the exterior of the unit disk in the $\xi$-plane to the exterior of $D$  given by
\[
f(\xi)=t_c^{\frac{1}{s}}\xi\left(1+\frac{t}{t_{c}}\frac{1}{\xi^{s}}\right)^{\frac{1}{s}}
\]
with inverse
\begin{equation}
\label{def_F}
F(\lambda)=f^{-1}(\lb)= t_c^{-\frac{1}{s}}\lambda \left(1-\dfrac{t }{\lambda^s}\right)^{\frac{1}{s}}.
\end{equation}
For $t>t_c$  the domain defined by the equation \eqref{domainmu} consists of  $s$ connected components which have a discrete rotational symmetry. For $s=2$ the domain $D$ is called Cassini oval.
The boundary of $D$, namely $\partial D$ can be identified with the real ovals of the Riemann surface  ${\cal S}$  defined by ${\cal S}:=\{(\lb,\eta)\in\mathbb{C}^2,\;\; (\lambda^s-t)(\eta^s-t)=t_c^2\}$ which has genus
 $(s-1)^2$. Such Riemann surface   does not coincide with the Schottky double of $D$ for $s>2$.

The boundary of the domain $D$ can be also expressed by the equation
\begin{equation}
\label{Schwartz} 
\bar{\lb}=S(\lb),\quad S(\lb)=\left(t+\dfrac{t_c^2}{\lb^s-t}\right)^{\frac{1}{s}}.
\end{equation}
The function   $S(\lb)$  is analytic in a neighbourhood of $\partial D$ and it is called the \emph{Schwarz function} associated to $\partial D$ (see e.g. \cite{Davis}).

\begin{remark}
The domain $D$ is a quadrature domain \cite{Gustafsson} with respect to the measure $\d\mu^*$. Indeed for any function $h(\lambda)$ analytic in a neighbourhood of $D$ one has, by applying Stokes' theorem and the Residue Theorem,
\[
\int_{D}h(\lb)d\mu^*(\lb)=\dfrac{1}{2\pi i t_{c}^{2}}\int_{\partial D}h(\lb)\lb^{s-1}(S(\lb))^sd\lb=\sum_{k=0}^{s-1} c_k h(\lb_k)
\]
where $S(\lb)$ is the Schwarz function \eqref{Schwartz}, $c_k=\dfrac{1}{s}$  and $\lb_k=t^{\frac{1}{s}}e^{\frac{2 \pi i k}{s}}$.
\end{remark}

\subsection{Statement  of  results}
The goal of this manuscript is to determine pointwise asymptotics
of  the polynomials $p_n(\lambda)$ defined in \eqref{eq:op_original} orthogonal with respect to the weight \eqref{eq:ortho_measure}  in the two cases
\begin{itemize}
\item pre-critical: $t<t_c$;
\item post-critical: $t>t_c$.
\end{itemize}
The $\mathbb{Z}_s$-symmetry of the orthogonality measure \eqref{eq:ortho_measure} is inherited by the corresponding orthogonal polynomials.
 Indeed the non-trivial orthogonality relations are 
  \[
  \int_{\mathbb{C}}p_n(\lb)\bar{\lb}^{js+l}e^{-N W(\lb)}dA(\lambda)=0,\quad j=0,\dots,k-1,
  \]
  where  $k$ and $l$ are such that
   \be
 \label{eq_nd}
  n= ks+l, \qquad 0\leq l \leq s-1,
  \ee
  i.e., the $n$-th monic orthogonal polynomial satisfies the relation
 \be
 p_n(e^{\frac{2\pi i}{s}}\lb)=e^{\frac{2\pi i n}{s}}p_n(\lb).
 \ee
It follows that there exists a monic polynomial $q^{(l)}_k$ of degree $k$ such that
 \begin{equation}
 \label{qk}
 p_n(\lb) = \lambda^lq^{(l)}_k(\lb^s).\quad
 \end{equation}

 Therefore the sequence of orthogonal polynomials $\{p_{n}(\lb)\}_{n=0}^{\infty}$ can be split into $s$ subsequences labelled by the remainder $l \equiv n\mod s$, and the asymptotics along the different subsequences can be studied via the sequences of reduced polynomials
 \be
 \nonumber
\left\{q^{(l)}_k(u)\right\}_{k=0}^{\infty}\ ,\qquad l=0,1,\dots, s-1.
 \ee
 By a simple change of coordinates it is easy to see that the monic polynomials in the sequence $\{q^{(l)}_k\}_{k=0}^{\infty}$  are orthogonal with respect to the measure
 \be
 \label{eq_reduced}
 |u|^{-2\gamma} e^{-N(|u|^2 -tu-t\bar{u})}dA(u),\quad \gamma := \frac{s-l-1}{s} \in [0,1),
 \ee
 namely they satisfy the orthogonality  relations
 \be
 \label{eq:ort_reduced}
 \int_{\C}q_k^{(l)}(u)\bar{u}^{j}|u|^{-2\gamma}e^{-N\left(|u|^2-tu-t\bar{u}\right)}dA(u)=0, \qquad j=0,\dots, k-1.
 \ee
 As a result of this symmetry reduction, starting from the class of measures \eqref{eq:ortho_measure}, it is sufficient to consider the orthogonal polynomials with respect to the family of measures
 \eqref{eq_reduced}. It is clear from the above relation that for $l=s-1$ one has $\gamma=0$ and  the polynomials $q_k^{(s-1)}(u)$ are monomials in the variable $(u-t)$, that is,
 \[
 q_k^{(s-1)}(u)=(u-t)^k.
 \]
 It follows that the monic  polynomials $p_{ks+s-1}(\lb)$
have the form
\[
p_{ks+s-1}(\lb)=\lb^{s-1}(\lb^s-t)^k.
\] 
\begin{remark}
Observe that the weight in the  orthogonality relation \eqref{eq:ort_reduced} can be written in the form
\[
 |u|^{-2\gamma} e^{-N(|u|^2 -tu-t\bar{u})}=e^{-N(|u-t|^2-t^2+2\frac{\gamma}{N}\log|u|)},
\]
and it is  similar to  the weight $e^{-NW(u)}$ with $W(u)=|u|^2-c\log|u-a|$ with $c>0$  studied in \cite{BBLM}. However  in our case it turns out that  $c=-2\gamma/N<0$, so the point interaction near   $a=0$ is repulsive and the asymptotic distribution of the zeros of the  polynomials  \eqref{eq:ort_reduced} turns out to be  substantially different  from the one in \cite{BBLM}.\end{remark}

Define
\be
z_0=\dfrac{t^2_c}{t^2},
\ee
and for $r>0$ the function
\begin{equation}
\label{hat_phi_r}
\hat{\phi}_r(\lb)=\log(t-\lb^s)+\dfrac{\lb^s}{tz_0}-\log rt+\dfrac{r-1}{z_0}.
\end{equation}
Let us consider the level curve $\hat{\mathcal{C}}_r$
\begin{equation}
\label{Gamma_r}
\hat{\mathcal{C}}_{r}:=\left\{\lb\in\mathbb{C},\ \Re\,\hat{\phi}_r(\lb)=0,\;\;|\lb^s-t|\leq rt\right\}.
\end{equation}
The level curves $\hat{\mathcal{C}}_r$ consist of  $s$ closed contours contained in the set $ D$, where $D$ has been defined in \eqref{domainmu}. For these curves we consider the usual counter-clockwise orientation. Define the measure $\hat{\nu}$ associated with this family of curves given by
\begin{equation}
\label{nur}
d\hat{\nu}=\dfrac{1}{2\pi i s}d \hat{\phi}_r(\lambda),
\end{equation}
and supported on $\hat{\mathcal{C}}_r$.
\begin{lemma}
\label{lemma1}
The a-priori complex measure $d\hat{\nu}$ in \eqref{nur}  is a probability measure on the contour $\hat{\mathcal{C}}_r$ defined in \eqref{Gamma_r} for $0<r\leq \frac{t}{t_c}$.
\end{lemma}
Let us denote by $\nu(p_n)$ the zero counting measure associated with the polynomials $p_n$, namely
\[
\nu(p_{n})=\dfrac{1}{n}\sum_{p_{n}(\lb)=0}\delta_{\lb}, 
\]
where $\delta_{\lb} $ is point distribution with total mass one at the point $\lambda$.
\begin{theorem}
\label{theorem1}
The zeros of the polynomials $p_n(\lb)$ defined in \eqref{eq:op_original}  behave as follows
\begin{itemize}
\item for $n=s k+s-1$,  let $\omega=e^{\frac{2\pi i}{s}}$. 
Then $t^{\frac{1}{s}},\omega t^{\frac{1}{s}},\dots,\omega^{k-1} t^{\frac{1}{s}}$ are zeros of the polynomials $p_{ks+s-1}$ with multipicity $k$ and $\lb=0$ is a zero with multiplicity $s-1$.
\item for $n=ks+l$, $l=0,\dots,s-2$ the polynomial $p_n(\lb)$ has a zero in $\lb=0$ with multiplicity $l$ and the remaining zeros in the limit $n,N\to\infty$ such that
\begin{equation}
\label{def_N}
N=\dfrac{n-l}{T}
\end{equation}    accumulates on the level  curves  $\hat{\mathcal{C}}_r$ as in \eqref{Gamma_r} with $r=1$ for $t<t_c$ and $r=\frac{t_c^2}{t^2}$ for $t>t_c$. Namely  the curve $\hat{\mathcal{C}}$ on which the zeros accumulate is given by 
\begin{equation}
\label{Gamma}
\hat{\mathcal{C}}:\quad \left| (t-\lb^s)\exp\left(\frac{\lb^st}{t_{c}^2}\right)\right|=
\left\{\begin{array}{ll}
t&\quad\mbox{pre-critical case $0<t<t_{c}$},\\&\\
\dfrac{t_{c}^2}{t}e^{\frac{t^2}{t_c^2}-1}&\quad  \mbox{post-critical case $t>t_{c}$},
\end{array}\right.
\end{equation}
with $|\lb^s-t|\leq z_0t$.
The measure $\hat{\nu}$ in \eqref{nur} is the weak-star limit of the normalized zero counting measure  $\nu_n$ of the polynomials $p_n$  for $n=sk+l$, $l=0,\dots, s-2$.
\end{itemize}
\end{theorem}

\begin{figure}[H]
\centering
\includegraphics[scale=0.4]{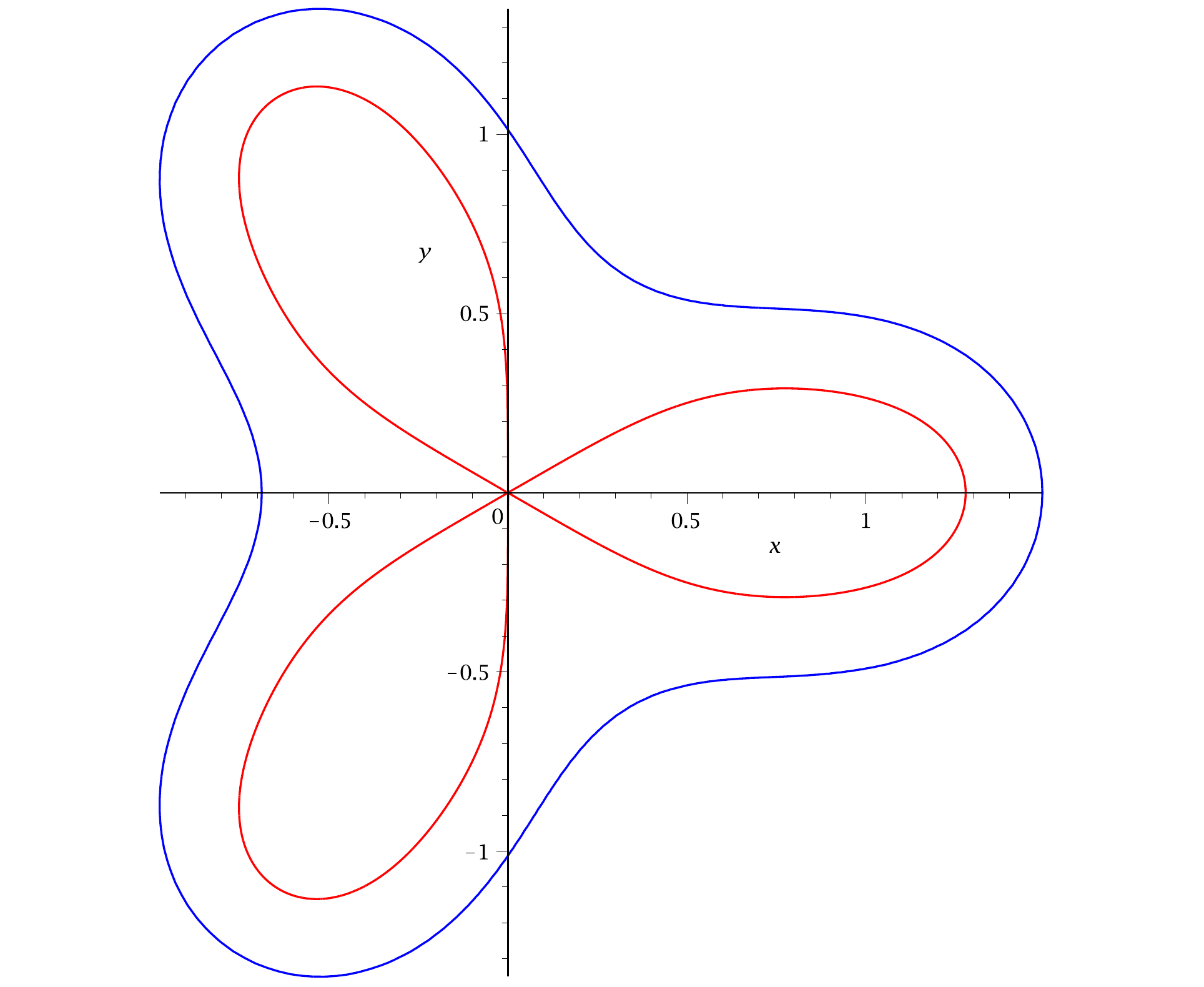}\includegraphics[scale=0.4]{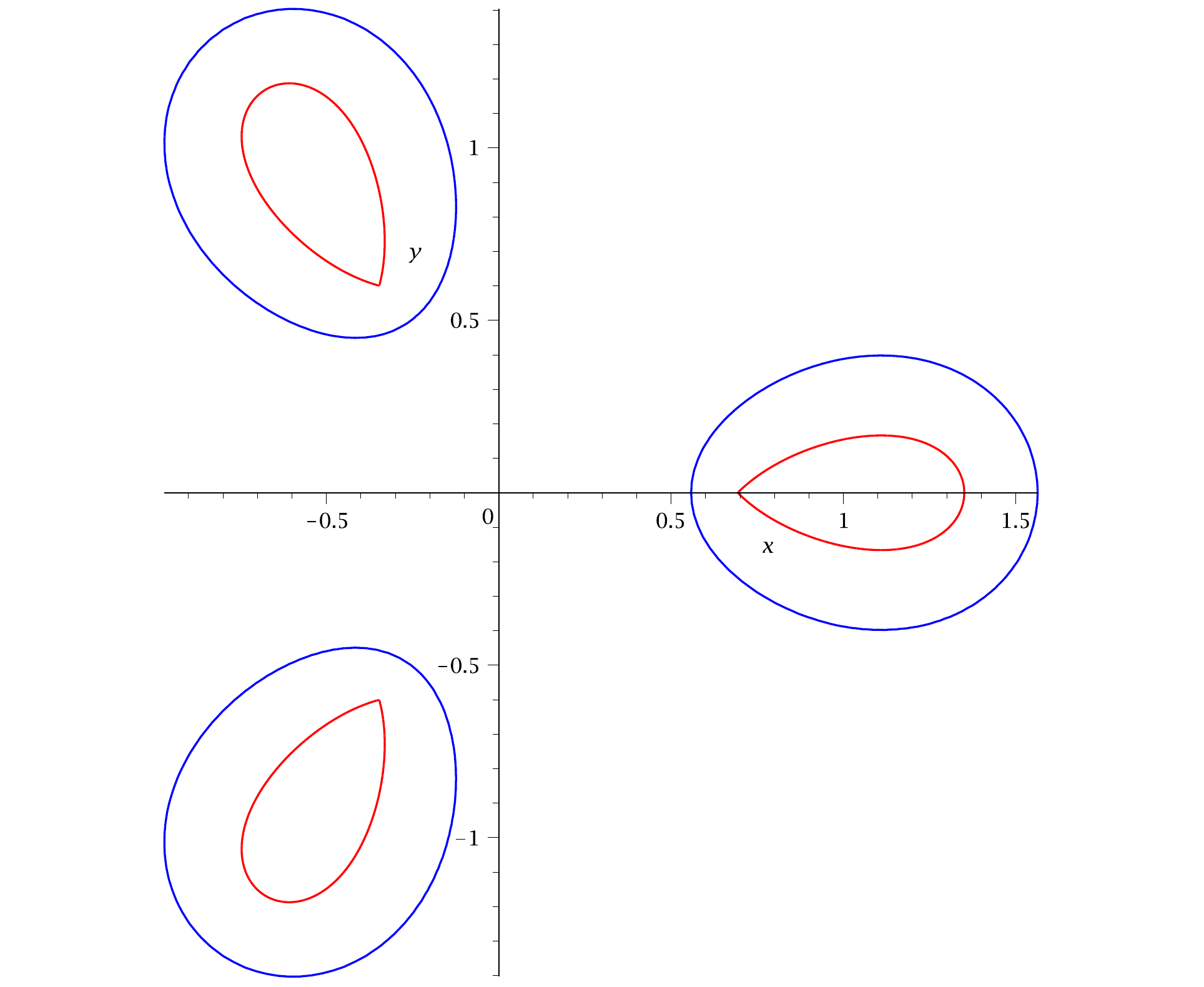}\\
\caption{The blue contour is the boundary of the eigenvalue distribution support $D$ defined in \protect\eqref{domainmu} relative to the normal matrix model and  the red  contour is $\hat{\mathcal{C}}$ defined in \protect\eqref{Gamma} which  describes the support of the limiting zero distribution of the polynomials $p_n(\lb)$. Here $s=3$ and $t<t_c$ (left figure) and $t>t_c$ (right figure).
}
\label{domain_D}
\end{figure}
\begin{remark}
We observe that  the curve \eqref{Gamma} in the rescaled variable $z=1-\lambda^s/t$  takes the form
\begin{equation}
\label{Gamma0}
\mathcal{C}:\quad \left| z e^{\frac{1-z}{z_0}}\right|=
\left\{\begin{array}{ll}
1&\quad\mbox{pre-critical case $z_0>1$},\\&\\
z_0e^{\frac{1}{z_0}-1}&\quad  \mbox{post-critical case $0<z_0<1$}.
\end{array}\right.
\end{equation}
with $z_0=\dfrac{t_c^2}{t^2}$ and $|z|\leq z_0$. The curve $\mathcal{C}$ 
 is similar to the  Szeg\H o curve  $\{z\in\mathbb{C}\ \colon\ |ze^{1-z}|=1,\;\;|z|\leq 1\}$ that was first observed in relation to the zeros of the Taylor polynomials of the exponential function \cite{Szego}, 
 and  coincides exactly with such  curve in the critical case $z_0=1$. The  Szeg\"o curve 
  also appeared in the asymptotic analysis of the generalized Laguerre polynomials, see e.g.~\cite{MMO},\cite{Chen},\cite{KuijMc}. 
  The curve  \eqref{Gamma0} is the limiting curve for the  zeros of the polynomials
\be
\label{def:pi_k0}
\poly_k(z) := \frac{(-1)^k}{t^k}q^{(l)}_k(-t(z-1))
\ee
where $q^{(l)}_k(u)$ have been defined in \eqref{qk}.
\end{remark}
From   Theorem~\ref{theorem1}  and \eqref{E-L}, the following identity follows immediately.
\begin{lemma}
\label{lemma_balayage}
Given the potential $W(\lb)$  in \eqref{def_W},  the measure $ \hat{\nu}$ in \eqref{nur} with $r=1$ for $t<t_c$ and $r=\frac{t_c^2}{t^2}$ for $t>t_c$, and the contour $\hat{\mathcal{C}}$ defined in \eqref{Gamma},  the equation
\begin{equation}
\label{balayage00}
\dfrac{1}{T}\dfrac{\partial}{\partial \lb}W(\lb)=\int_{\hat{\mathcal{C}}}\dfrac{d\hat{\nu}(\eta)}{\lambda-\eta},\quad \lb\notin\hat{\mathcal{C}},
\end{equation}
defines the boundary of a domain which coincides with $D$ defined in \eqref{domainmu}.
The measure $\mu^*$ in \eqref{dmustar} of the eigenvalue distribution of the normal matrix model and the measure $\hat{\nu}$  of  the zero distribution of the orthogonal polynomials are related by
\begin{equation}
\label{balayage1}
\int_{D}\dfrac{d\mu^*(\eta)}{\lambda-\eta}=\int_{\hat{\mathcal{C}}}\dfrac{d\hat{\nu}(\eta)}{\lambda-\eta},\quad \lb\in\mathbb{C}\backslash D.
\end{equation}

\end{lemma}

\begin{remark}
The identities  \eqref{balayage00}  and \eqref{balayage1} in Lemma~\ref{lemma_balayage} are expected to hold in general for a large class of  normal matrix models. It  has been verified for several other potentials (see, for example \cite{Akemann,Elbau,BBLM,BleherKuij,Kuij2}).

  
We also observe that for the orthogonal polynomials appearing in random matrices, in some cases, the asymptotic distribution of the zeros is supported on the so-called  mother body or potential theoretic skeleton  of the support $D$ of the eigenvalue distribution.
We recall that a measure $\nu$ is a strong mother body for  a domain $D$ with respect to a measure $\mu$ if \cite{GustafssonVasiliev}
\begin{itemize}
\item[1)]  $\int\log|\lb-\eta|d\mu(\eta)\leq \int\log|\lb-\eta|d\nu(\eta)$, for $\lb \in\mathbb{C}$ with equality for $\lb$ outside $D$;
\item[2)]  $\nu\geq 0$ and $\mbox{supp }\nu\subset \mbox{supp }\mu$;
\item[3)] the support of $\nu$ has zero area measure;
\item[4)] the support of $\nu$ does not disconnect any part of $D$ from $\bar{D}^c$.
\end{itemize}
If the measure $\nu$ has only the property 1), 2) and 3)  it is called weak mother body.
The problem of constructing mother bodies is not always solvable and the solution is not always unique \cite{SSS}.

Concerning the explicit examples appearing in the random matrix literature, for the exponential weight $W(\lb)=|\lb|^2+\Re (P(\lb))$ where $P(\lb)$ is polynomial, the support of the zero  distribution of the orthogonal polynomials  is  a strong mother body of the domain that corresponds to the eigenvalue distribution of the matrix 
model (see, e.g., \cite{BleherKuij, Kuij2, Ginibre, ElbauFelder, WiegZab1}).
In contrast, in the model studied in \cite{BBLM} and also in the present case, the support of the zero distribution of the orthogonal polynomials does not have property 4) and therefore it is a weak mother body of the 
 set $D$.
 \end{remark}
  The proof of Theorem~\ref{theorem1} is obtained from  the strong and uniform asymptotics of the polynomials $p_n(\lb)$ in the whole complex plane which is obtained  by characterising the orthogonal polynomials $p_n(\lb)$ via a Riemann-Hilbert method.\\
{\bf Uniform Asymptotics.}
In the next theorem we describe the strong and  uniform  asymptotic of the polynomials $p_n(\lb)$ in the complex plane. We distinguish the pre-critical case $t<t_c $ and post-critical case $t>t_c$.  We first  define the function
\begin{equation}
\label{hat_phi}
\hat{\phi}(\lb)=
\left\{
\begin{array}{lll}
\hat{\phi}_{r=z_0}(\lb)&0<t<t_c & \text{(pre-critical)}\\
&&\\
\hat{\phi}_{r=1}(\lb)&t_c < t & \text{(post-critical)},
\end{array}
\right.
\end{equation}
where $\hat{\phi}_r(\lb)$ is given by \eqref{hat_phi_r}.
\begin{theorem}[Pre-critical case]
\label{theorem2}
For $0<t<t_c $  the  polynomial $p_n(\lb)$ with $n=ks+l$, $l=0,\dots, s-2$, $\gamma=\frac{s-l-1}{s}\in(0,1)$,   have the following asymptotic behaviour for when $n,N\to\infty$ in such a way that $NT=n-l$:
\begin{itemize}
\item[(1)] for $\lambda$ in  compact subsets of the exterior of $\hat{\mathcal{C}}$ one has for any integer $M\geq 2$
\begin{equation}
\label{pn1_theo2}
p_n(\lb)=\lb^{s-1}(\lb^s-t)^{k-\gamma}
\left(1+\order{\frac{1}{k^{M+\gamma}}}\right);
\end{equation}
\item[(2)]  for $\lambda$ near $\hat{\mathcal{C}}$ and away from $\lb=0$,
\begin{equation}
\label{pnzero}
p_n(\lb)=\lb^{s-1}(\lb^s-t)^{k-\gamma}\left[1+\dfrac{e^{-k\hat{\phi}(\lb)}}{k^{1+\gamma}}\left(\dfrac{1}{\Gamma(-\gamma)}\left(1-\frac{1}{z_0}\right)^{-1-\gamma} \dfrac{t}{\lb^{s}}\left(1-\frac{t}{\lb^s}\right)^{\gamma}+\order{\frac{1}{k}}\right)\right]\,,
\end{equation}
where $\hat{\phi}(\lb)$ has been defined in \eqref{hat_phi} and $\Gamma(z)=(z-1)!$;
\item[(3)] for $\lb$ in compact subsets of the  interior of $\hat{\mathcal{C}}$ and away from $\lb=0$, 
\begin{equation}
p_{n}(\lb)=\lb^l\dfrac{e^{-\frac{k\lb^s}{tz_0} }}{k^{1+\gamma}}\left(  \dfrac{(-t)^{k+1}}{ \Gamma(-\gamma)}\dfrac{1}{\lb^s}\left(1-\frac{1}{z_0}\right)^{-1-\gamma}+\order{\frac{1}{k}}\right);
\end{equation}
\item[(4)] for $\lambda$ in a neighbourhood of $\lambda=0$,  we introduce the function $\hat{w}(\lb)=\hat{\phi}(\lb)+2\pi i$ if $\lb^s\in\mathbb{C}_-$ and  $\hat{w}(\lb)=\hat{\phi}(\lb)$ if $\lb^s\in\mathbb{C}_+$.
Then 
\begin{equation}
\label{pn4_theo2}
p_n(\lb)=\lb^{l}(\lb^s-t)^{k-\gamma}\left(\dfrac{\lb^{s}}{\hat{w}(\lb)}\right)^{\gamma}\left[(\hat{w}(\lb))^{\gamma}-\dfrac{e^{-k\hat{\phi}(\lb)}}{k^{\gamma}}\left(\tilde{\Psi}_{12}(k\hat{w}(\lb))+\order{\frac{1}{k}}\right)\right],
\end{equation}
\end{itemize}
where the $(1,2)$-entry of the matrix $\tilde{\Psi}$ is defined in \eqref{tildePsi}.\end{theorem} 

We observe that  in compact subsets of  the exterior  of  $\hat{\mathcal{C}}$ there are no  zeros of the polynomials $p_n(\lb)$. The only possible zeros are located in $\lb=0$ and   in the region where the  second 
term in parenthesis in the expression \eqref{pnzero} is of order one.  Since $\Re \hat{\phi}(\lb)$ is negative inside $\hat{\mathcal{C}}$ and positive outside $\hat{\mathcal{C}}$ it follows that the possible zeros of $p_n(\lb)$ lie
inside $\hat{\mathcal{C}}$ and are determined by the condition
\begin{equation}
\label{phi_modified}
\Re\hat{\phi}(\lb)=-(1+\gamma)\dfrac{\log k}{k}+\dfrac{1}{k}\log \left(\dfrac{1}{|\Gamma(-\gamma)|} \dfrac{t}{|\lb|^{s}}\left|1-\frac{1}{z_0}\right|^{-1-\gamma}\left|1-\frac{t}{\lb^s}\right|^{\gamma}\right),\quad |\lb^s-t|\leq t.
\end{equation}
The expansion \eqref{pnzero}  shows that the zeros of the polynomials $p_n(\lb)$  are within a  distance $\order{1/k^2}$ from the level curve \eqref{phi_modified}. This level curve approaches $\hat{\mathcal{C}}$ defined in \eqref{Gamma}
 at a rate $\order{\log k/k}$.

 \begin{figure}[H]
\centering
\includegraphics[scale=0.35]{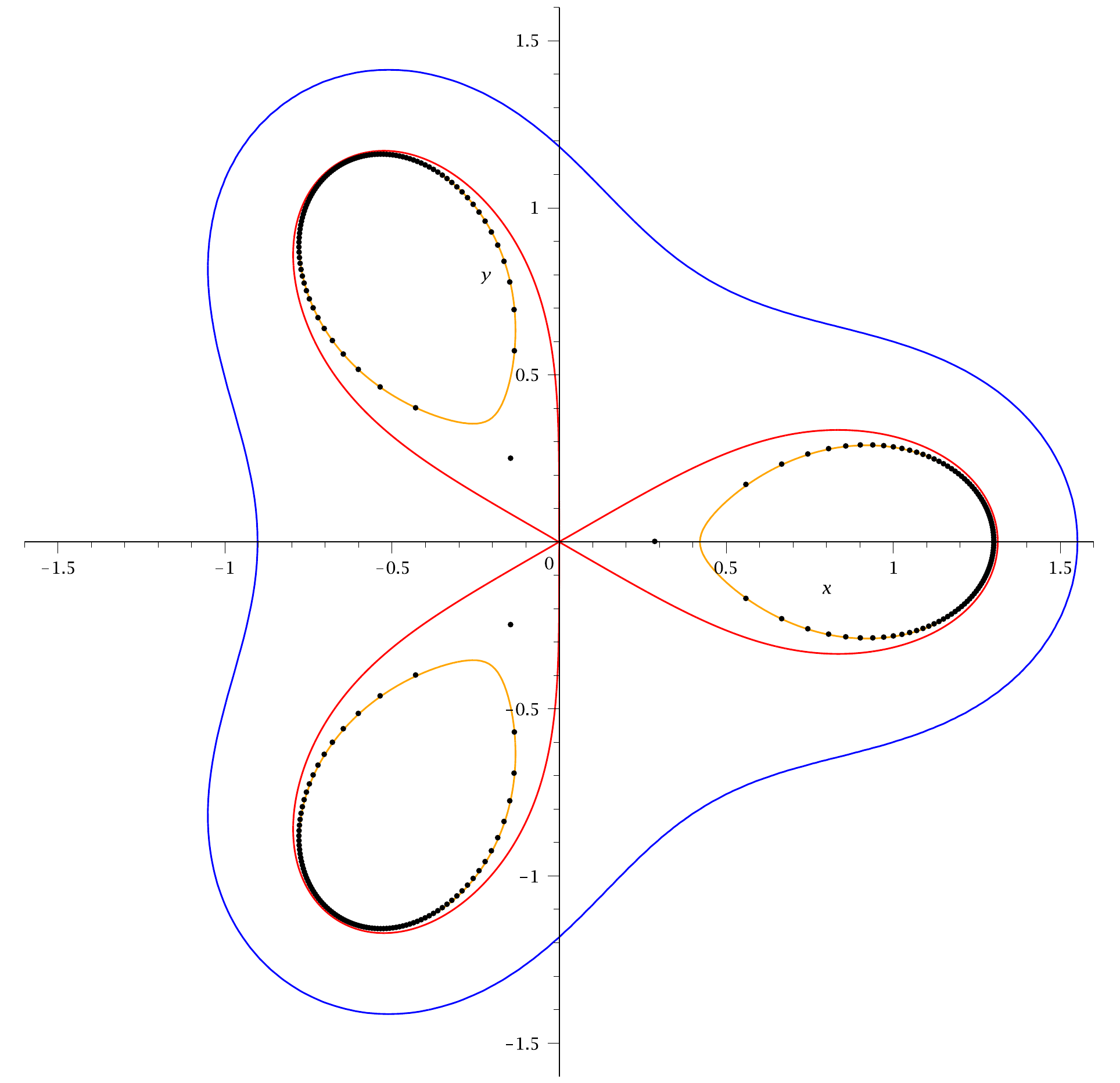}\includegraphics[scale=0.36, trim=0cm -.2cm 0cm 0cm]{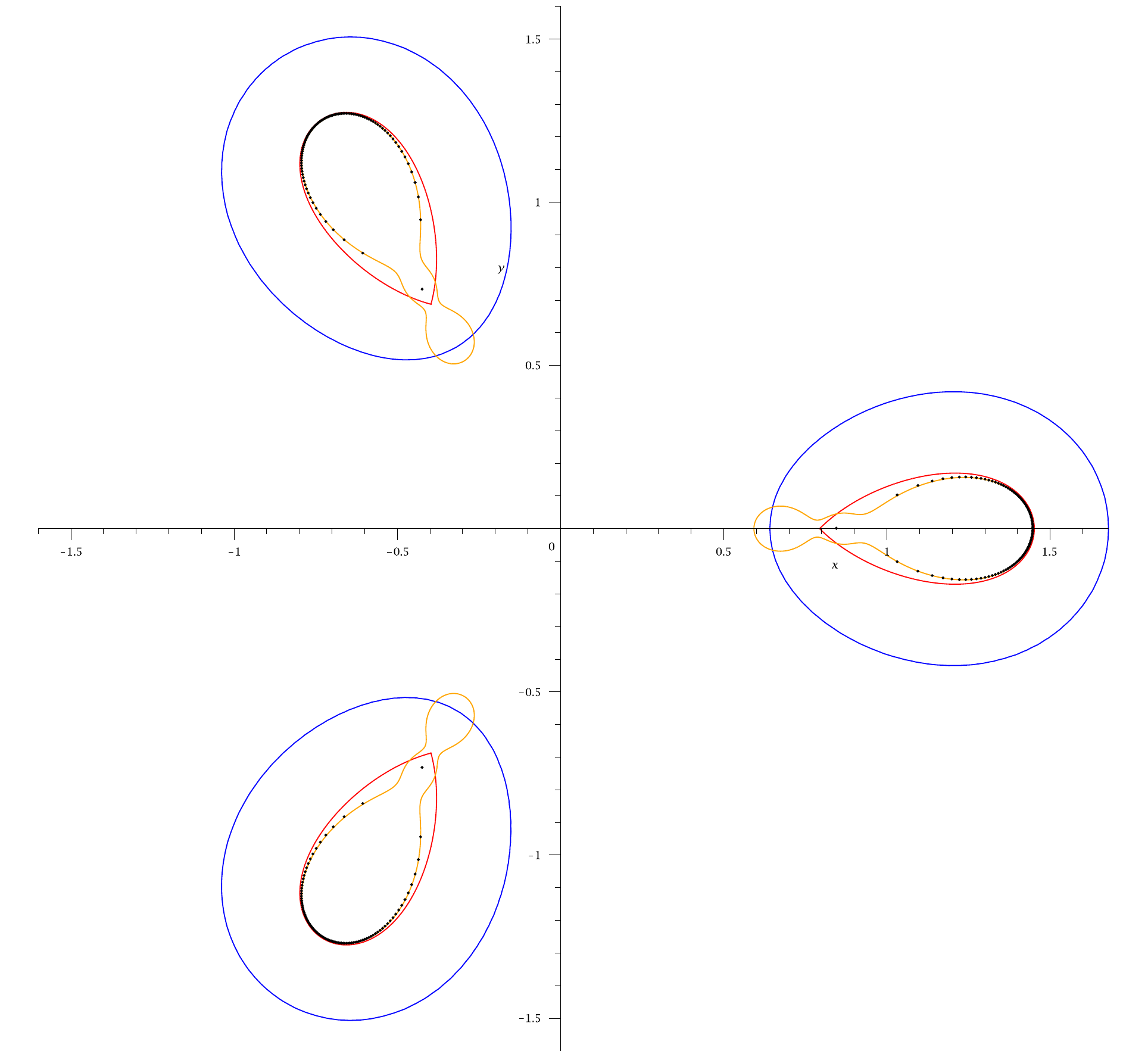}\\
\caption{The blue contour is the boundary of the  support of the eigenvalue distribution defined in \protect\eqref{domainmu} and  the red  contour  $\hat{\mathcal{C}}$ defined in \protect\eqref{Gamma}, is the support of the limiting zero distribution of the orthogonal polynomials $p_n(\lb)$.  The yellow contour is given by \protect\eqref{phi_modified} in the pre-critical case and \protect\eqref{phi_modified_post} in the post-critical case and lies within a distance of order $\order{\frac{\log n}{n}}$ from the contour $\hat{\mathcal{C}}$. The dots are the zeros of the polynomial $p_n(\lb)$ for $s=3$, $n=285$,
 $l=0$ and $t<t_c$ on the right and $t>t_c$ on the left, respectively.}
\label{zero_pn}
\end{figure}

 For $t>t_c $   the  polynomials $p_n(\lb)$ with $n=ks+l$, $l=0,\dots, s-2$,   have the following  asymptotic behaviour.

\begin{theorem}[Post-critical case]
\label{theorem3}
For $t>t_c $    the  polynomials $p_n(\lb)$ with $n=ks+l$, $l=0,\dots, s-2$, and $\gamma=\frac{s-l-1}{s}\in(0,1)$ have the following behaviour when $n,N\to\infty$ in such a way that $NT=n-l$
\begin{itemize}
\item[(1)] for $\lb$  in compact subsets of the exterior of $\hat{\mathcal{C}}$   one has 
\begin{equation}
\label{pnpost1}
p_n(\lb)=\lb^{l}(\lb^s-t)^{k-\gamma}(\lb^s+t(z_0-1))^{\gamma}\left(1+\order{\frac{1}{k}}\right),
\end{equation}
with $z_0=\dfrac{t_c^2}{t^2}$;
\item[(2)] for $\lb$ in  the region near  $\hat{\mathcal{C}}$   and away from  the points $\lb^s=t(1-z_0)$ one has 
\begin{equation}
\label{pn_post}
p_n(\lb)=\lb^{l}(\lb^s-t)^{k-\gamma}(\lb^s+t(z_0-1))^{\gamma}\left(1-\dfrac{e^{-k\hat{\phi}(\lb)}}{k^{\frac{1}{2}+\gamma}}\dfrac{\gamma z_0^2}{c}\dfrac{t((\lb^s-t)\lb^s)^{\gamma}}{(\lb^s+t(z_0-1))^{2\gamma+1}}+\order{\frac{1}{k}}\right),
\end{equation}
with $\hat{\phi}(\lb)$ defined in \eqref{hat_phi} and the constant $c$ defined in \eqref{def_c};
\item[(3)] for $\lb$ in compact subsets of the  interior region of $\hat{\mathcal{C}}$  one has
\begin{equation}
p_n(\lb)=\lb^l \dfrac{e^{k\frac{t-\lb^s}{tz_0}}}{k^{\frac{1}{2}+\gamma}}\dfrac{t\gamma z_0^{2}}{c}\left(\dfrac{tz_{0}}{e}\right)^{k}\left(\dfrac{\lb^{s\gamma}}{(\lb^s+t(z_0-1))^{\gamma+1}}+\order{\frac{1}{k}}\right);
\end{equation}
\item[(4)] in the  neighbourhood of   each of the points that solve the equation $\lb^s=t(1-z_0)$ one has 
\begin{equation}
p_n(\lb)=\lb^{l}(\lb^s-t)^{k-\gamma} \left(\dfrac{\lb^s+t(z_0-1)}{\sqrt{k}\hat{w}(\lb)}\right)^{\gamma}e^{-k\hat{\phi}(\lb)}\left({\cal U}(-\gamma-\frac{1}{2}; \sqrt{2k}\hat{w}(\lb))+\order{\frac{1}{k}}\right);
\end{equation}
where ${\cal U}(a;\xi)$ is the parabolic cylinder function (Chap.~13, \cite{AbramowitzStegun}) satisfying  the equation $\dfrac{d^2}{d\xi^2}{\cal U}=\left(\frac{1}{4}\xi^2+a\right){\cal U}$ and $\hat{w}^2(\lb)=-\hat{\phi}(\lb)-2\pi i $ for $\lb^s\in\mathbb{C}_-$ and $\hat{w}^2(\lb)=-\hat{\phi}(\lb)$ for $\lb^s\in\mathbb{C}_+$.
\end{itemize}
\end{theorem}

We observe that  in compact subsets of  the exterior  of  $\hat{\mathcal{C}}$  the polynomials $p_n(\lb)$ have zero at $\lb=0$ with multiplicity $l$.
 The other  possible zeros are located   in the region where the  second 
term in parenthesis in the expression \eqref{pn_post} is of order one. This happens in the region where  $\Re \hat{\phi}(\lb)<0 $ and $\Re\hat{\phi}(\lb)=\order{\log k/k}$, namely in a region inside the contour  $\hat{\mathcal{C}}$ and defined by the equation
\begin{equation}
\label{phi_modified_post}
\Re\hat{\phi}(\lb)=-(1+\gamma)\dfrac{\log k}{k}+\dfrac{1}{k}\log \left(\dfrac{t \gamma z_0^2}{|c|}\dfrac{|(\lb^s-t)\lb^s|^{\gamma}}{|\lb^{s}+t(z_{0}-1)|^{2\gamma+1}}
\right).
\end{equation}
The expansion \eqref{pn_post}  shows that the zeros of the polynomials $p_n(\lb)$  are within a  distance $\order{1/k^2}$from the level curve \eqref{phi_modified_post}. This level curve approaches $\hat{\mathcal{C}}$ defined in \eqref{Gamma}
 at a rate $\order{\log k/k}$.
 \medskip
 
 The proofs of Theorem~\ref{theorem2} and Theorem~\ref{theorem3} are obtained by reducing the planar orthogonality relations of the polynomials $p_n(\lb)$ to orthogonality relations with respect to a complex density on a contour. More precisely, the sequence of polynomials $p_n(\lb)$ can be reduced to $s$ families of polynomials $q^{(l)}_k(\lb^s)$, $l=0,\dots,s-1$ as in \eqref{qk} with $n=ks+l$. The orthogonality relations of the polynomials $q^{(l)}(u)$ are reduced to orthogonality relations on a contour.  We then reformulate such orthogonality relation as a Riemann-Hilbert problem. We perform the asymptotic analysis of the polynomials $q^{(l)}_k(\lb)$ 
  using  the nonlinear steepest descent/stationary phase method introduced by Deift and Zhou \cite{DeiftZhou} and successfully applied in Hermitian random matrices and orthogonal polynomials on the line
   in the seminal papers, \cite{DKMVZ},\cite{DKMVZ2}. 
  See \cite{Deift} for an introduction to this method, with a special emphasis on the applications to orthogonal polynomials and random matrices. It would be also interesting to explore the asymptotic 
  for the polynomials $p_n(\lb)$ using the  $\bar{\partial}$-problem introduced in \cite{ItsTak}.
  
  The zeros of $p_n(\lb)$  accumulate along an open contour as shown in Figure~\ref{domain_D}.  
   The determination of this contour is a first step in the analysis. Riemann--Hilbert analysis in which a contour selection method was required, also appeared in the  papers \cite{BBLM},\cite{KuijMc}.

  The Riemann-Hilbert method gives strong and uniform asymptotics of the polynomials $p_n(\lb)$  in the whole complex plane. The asymptotic behaviour of the polynomials $p_n(\lb)$ in the leading and sub-leading order can be expressed  in terms of 
  elementary functions in the pre-critical case $t<t_c$, while  in the post-critical case $t>t_c$ we have used, in some regions of the complex plane,    parabolic cylinder functions as in \cite{DeiftZhou},\cite{JM},\cite{M}.
The proof of Theorem~\ref{theorem1} is  then deduced from the strong asymptotic of the orthogonal polynomials.
 
 The paper is organized as follows:
 \begin{itemize}
 \item In Sect.~\ref{sec:RH}, by using the symmetry properties of the external potential, we reduce the orthogonality relations of the polynomials $p_n(\lb)$ on a contour and recall how to formulate a Riemann-Hilbert problem for such orthogonal polynomials \cite{FIK}.
 \item In Sect.~\ref{sec:pre-crit} we obtain the strong and uniform asymptotics for the polynomials $p_n(\lb)$ in the pre-critical case $t<t_c$.
 \item In Sect.~4 we find the strong and uniform asymptotics for the polynomials $p_n(\lb)$ in the post-critical case $t>t_c$ and we prove Theorem~\ref{theorem1}.
\end{itemize}

In the  critical case  $t=t_c$,  the set $D$ in \eqref{domainmu} that supports the eigenvalue distribution of the normal matrix models has a singularity in $\lb=0$. The corresponding asymptotics of the orthogonal polynomials  seems to be described in terms of Painlev\'e IV  equation as in \cite{Dai} and this  situation is quite different from the generic singularity that  have been described in terms of the Painlev\'e I equation 
 \cite{LeeTeoWieg2},\cite{KuiTov}. Such problem will be investigated in a subsequent publication.
 
 Furthermore, from the norming constants of the orthogonal polynomials, one may consider the problem of determining  the asymptotic expansion of the partition function $Z_{n,N}$ in the spirit of \cite{CGM}.
 Also this problem will be investigated in a subsequent publication.

\medskip
{\bf Acknowledgements.} The authors thank M.~Bertola,  G.~Silva,  A.~Tovbis for useful discussions and  K.~McLaughin for suggesting the problem.
 The first part of this  work  was carried out while F.B.~was a postdoctoral fellow at SISSA, later he was supported in part by the \emph{Fonds de recherche du Qu\'ebec - Nature et technologies} (FRQNT) via the research grant \emph{Matrices Al\'eatoires, Processus Stochastiques et Syst\`emes Int\'egrables}. F.~B.~acknowledges the excellent working conditions at Concordia University and at the Centre des Recherches Math\'ematiques (CRM).
 T.~G.~acknowledges the support by  the Leverhulme Trust Research Fellowship RF-2015-442 and the Miur Research project Geometric and analytic theory of Hamiltonian systems in finite and infinite dimensions of Italian Ministry of Universities and Research.

 \section{The associated Riemann--Hilbert problem}
 \label{sec:RH}
In this section we set up the Riemann--Hilbert problem to study the asymptotic behaviour of the orthogonal polynomials $p_n(\lb)$. As seen above, the analysis of $p_n(\lb)$ can be reduced to that of the polynomials $q_k^{(l)}(\lb^d)$ introduced in \eqref{qk}. The polynomials $q_k^{(l)}(u)$ are characterized by the symmetry-reduced orthogonality relations
 \be
 \label{eq:ort_reduced1}
 \int_{\C}q_k^{(l)}(u)\bar{u}^{j}|u|^{-2\gamma}e^{-N\left(|u|^2-tu-t\bar{u}\right)}dA(u)=0, \qquad j=0,\dots, k-1.
 \ee
 \subsection{Reduction to contour integral orthogonality}
 The crucial step in the present analysis is to replace the two-dimensional integral conditions \eqref{eq:ort_reduced1} by an equivalent set of linear constraints in terms of contour integrals.
In what follows it will be advantageous to perform the change of coordinate
\be
u = -t(z-1),\quad z\in\mathbb{C},
\ee
and characterize $q_k^{(l)}(u)$ in terms of the transformed polynomial
\be
\label{def:pi_k}
\poly_k(z) := \frac{(-1)^k}{t^k}q^{(l)}_k(-t(z-1))
\ee
in the new variable $z$. The polynomial $\poly_k(z)$ is also a monic polynomial of degree $k$ and it can be characterized as follows.

\begin{theorem}
\label{theodbar}

The polynomial $\poly_k(z)$ is characterized by the non-hermitian orthogonality relations
\be
\label{pi_k}
\oint_{\Sigma} \poly_k(z)z^j \frac{e^{-Nt^2 z}}{z^k}\left(\frac{z}{z-1}\right)^{\gamma}dz  = 0 \qquad j=0,1,\dots ,k-1\ ,\;\;\gamma\in(0,1),
\ee
where $\Sigma$ is a simple positively oriented contour encircling $z=0$ and $z=1$ and the function $\left(\frac{z}{z-1}\right)^{\gamma}$ is analytic in $\mathbb{C}\backslash [0,1]$ and tends to one for $|z|\to\infty$.
\end{theorem}

\begin{wrapfigure}{ltbh!}{0.3\textwidth}
\includegraphics{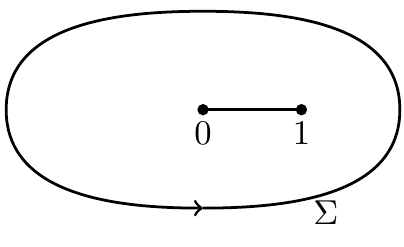}
\caption{The contour $\Sigma$}
\label{fig:cont_sigma}
\end{wrapfigure}
\begin{proof}
 In order to prove the theorem we first show that the orthogonality relation for the polynomials $q^l_k(u)$ on the  plane can be reduced to an orthogonality relation on a contour. To this end we seek  a function $\chi_j(u,\bar{u})$ that solves the $\bar{\partial }$-problem
 \be
 \label{dbarchi}
\partial_{\c{u}}\chi_j(u,\bar{u}) = \c{u}^j|u|^{-2\gamma}e^{-N\left(|u|^2-tu-t\c{u}\right)}.
\ee
Having such a function, for any polynomial $q(u)$ one has 
\begin{align}
& &d\left[q(u) \chi_j(u,\c{u}) du\right] &= q(u)\partial_{\c{u}}\chi_j(u,\c{u})d\c{u}\wedge du\\
& & &= q(u)\c{u}^j|u|^{-2\gamma}e^{-N\left(|u|^2-tu-t\c{u}\right)}d\c{u}\wedge du,
\end{align}
where $d$ denotes the operation of exterior differentiation. If such $\chi_j(u,\c{u})$ exists, one can use Stokes' theorem and reduce the planar orthogonality relation to an orthogonality relation on a suitable contour.
The equation \eqref{dbarchi} has a contour integral solution
\begin{align*}
\chi_j(u,\bar{u}) &= u^{-\gamma}e^{Ntu}\int_{0}^{\c{u}}a^{j-\gamma}e^{-Nua+Nt a}da\\
&= \frac{1}{N^{j-\gamma+1}}\left(1-\frac{t}{u}\right)^{\gamma}\frac{e^{Ntu}}{ (u-t)^{j+1}}\int_{0}^{N\c{u}(u-t)}r^{j-\gamma}e^{-r}dr\\
&= \frac{1}{N^{j-\gamma+1}}\left(1-\frac{t}{u}\right)^{\gamma}\frac{e^{Ntu}}{(u-t)^{j+1}}\left[\Gamma(j-\gamma+1)-\int_{N\c{u}(u-t)}^{\infty}r^{j-\gamma}e^{-r}dr\right]\\
&= \frac{\Gamma(j-\gamma+1)}{N^{j-\gamma+1}}\left(1-\frac{t}{u}\right)^{\gamma}\frac{e^{Ntu}}{(u-t)^{j+1}}\left[1-\order{e^{-N\c{u}(u-t)}}\right]\quad |u|\to \infty\\
\end{align*}

It follows that  for any polynomial $q(u)$ the following integral identity holds:
\begin{align*}
\int_{\C}q(u)\bar{u}^{j}|u|^{-2\gamma}e^{-N(|u|^2-tu-t\c{u})}dA(u)&=\frac{1}{2i}\lim_{R \to \infty}\int_{|u|\leq R}q(u)\bar{u}^{j}|u|^{-2\gamma}e^{-N(|u|^2-tu-t\c{u})}d\c{u}\wedge du\\
&=\frac{1}{2i}\lim_{R \to \infty}\oint_{|u|= R}q(u)\chi_j(u,\c{u}) du\\
&=\frac{1}{2i}\lim_{R \to \infty}\oint_{|u|= R}q(u)\left[G_j(u) - \order{e^{-\c{u}(u-t)}}\right]du\\
&=\frac{1}{2i}\oint_{|z|= R_0}q(u)G_j(u)du\\
\end{align*}
where  $R$ and  $R_0$ are  sufficiently large and 
$ G_j(u) = \frac{\Gamma(j-\gamma+1)}{N^{j-\gamma+1}}\left(1-\frac{t}{u}\right)^{\gamma}\frac{e^{Ntu}}{(u-t)^{j+1}}$.
So  it follows that  for any polynomial $q(u)$ the following identity is satisfied:
\be
\int_{\C}q(u)\bar{u}^j |u|^{-2\gamma}e^{-N\left(|u|^2-tu-t\bar{u}\right)}dA(u) =\frac{\pi \Gamma(j-\gamma+1)}{N^{j-\gamma+1}}\frac{1}{2\pi i}\oint_{\tilde\Sigma}q(u)\frac{e^{Nt u}}{(u-t)^{j+1}}\left(1-\frac{t}{u}\right)^{\gamma}du\ ,
\ee
where $\gamma \in (0,1)$, $j$ is an arbitrary non-negative integer, and $\tilde\Sigma$ is a positively oriented simple closed loop enclosing $u=0$ and $u=t$.
Making the change of coordinate $u=-t(z-1)$   one arrives to the statement of the theorem.
\end{proof}

\subsection{The Riemann--Hilbert problem}
Our aim is to study the behaviour of the polynomials $\poly_k(z)$ in the limit $k\to\infty$ and $N\to\infty $ in such a way that for $n=ks+l$ one has
\begin{equation}
\label{defN}
 N=\frac{n-l}{T},\quad T>0.
\end{equation}
Let
\be
t_c^2=\dfrac{T}{s} \quad \text{and} \quad z_0=\dfrac{t_c^2}{t^2}
\ee
and introduce the function
\begin{equation}
\label{def_V}
V(z)=\log z+\dfrac{z}{z_0}.
\end{equation}
In terms of the weight function
\be
w_k(z) := e^{-kV(z)}\left(\frac{z}{z-1}\right)^{\gamma},
\ee
the orthogonality relations \eqref{pi_k} can be written in the form
\[
\oint_{\Sigma} \poly_k(z)z^j w_k(z)\,dz = 0 \qquad j=0,1,\dots ,k-1.
\]

In the limit $k \to \infty$, we distinguish two different cases:
\begin{itemize}
\item \emph{pre-critical case} $z_0 >1$, corresponding to $0<t < t_c$\ ,
\item \emph{post-critical case} $z_0<1$, corresponding to $t_c<t$\ .
\end{itemize}

Our goal now is to characterize the polynomial $\pi_k(z)$ as a particular entry of the unique solution of a matrix-valued Riemann-Hilbert problem.
Let us first define the complex moments
\begin{equation}
\nu_j:= \oint_{\Sigma} z^j w_{k}(z)\d z
\end{equation}
where,  for simplicity, the dependence on $k$ is suppressed in the notation. Introduce the auxiliary polynomial  
\begin{equation}
\label{Pidet}
\Pi_{k-1}(z):=\frac 1{\det\left[\nu_{i+j}\right]_{0\leq i,j\leq k-1}} \det \left[
\begin{array}{ccccc}
\nu_0 & \nu_1 & \dots & \nu_{k-1}\\
\nu_1 & \nu_2 & \dots & \nu_{k}\\
\vdots &&&\vdots \\
\nu_{k-2}&\dots&&\nu_{2k-3}\\
1 & z & \dots & z^{k-1}
\end{array}
\right]
\end{equation}
Note that $\Pi_{k-1}$ is not necessarily monic and its degree may be less than $k-1$: its existence is guaranteed just by requiring that the determinant in the denominator does not vanish.
\begin{proposition}
The determinant $\det [\nu_{i+j}]_{0\leq i,j\leq k-1}$ does not vanish and therefore $\Pi_{k-1}(z)$ is well-defined.
\end{proposition}
\begin{proof}
We have
\begin{equation}
\begin{split}
\det [\nu_{i+j}]_{0\leq i,j\leq k-1}& =  \det \left[\oint_{\Sigma} z^{i+j} \frac{e^{-Nt^2 z}}{z^k}\left(\frac{z}{z-1}\right)^{\gamma}dz\right]_{0\leq i,j\leq k-1}\\
&=(-1)^{k(k-1)/2}  \det \left[\oint_{\Sigma} z^{i-j} \frac{e^{-Nt^2 z}}{z}\left(\frac{z}{z-1}\right)^{\gamma}dz\right]_{0\leq i,j\leq k-1},
\end{split}
\end{equation}
where the last identity has been obtained by the reflection of the column index $j\to k-1-j$.
Due to Theorem \ref{theodbar} we have
\be
\int_{\C}\pi(z)(\bar{z}-1)^j |z-1|^{-2\gamma}e^{-Nt^2|z|^2}dA(z) =t^{2-2j-2\gamma}\frac{\pi \Gamma(j-\gamma+1)}{N^{j-\gamma+1}}\frac{1}{2\pi i}\oint_{\Sigma}\pi(z)\frac{e^{-Nt^2 z}}{z^{j+1}}\left(\frac{z}{z-1}\right)^{\gamma}dz\ ,
\ee
an hence the second determinant is given by 
\begin{equation}
\det \left[\oint_{\Sigma} z^{i-j} \frac{e^{-Nt^2 z}}{z}\left(\frac{z}{z-1}\right)^{\gamma}dz\right] = \det \left[\iint_\C z^i (\bar{z}-1)^j |z-1|^{-2\gamma}e^{-Nt^2|z|^2} \d A(z)\right] \prod_{j=0}^{k-1} \frac {2it^{2j+2\gamma-2}N^{j-\gamma+1}}{\Gamma(j-\gamma+1)}.\label{detcont}
\end{equation}
Finally, the determinant on the right-hand side is strictly  positive because 
\begin{equation}
\det \left[\iint_\C z^i (\bar{z}-1)^j |z-1|^{-2\gamma}e^{-Nt^2|z|^2} \d A(z)\right]  =  \det \left[\iint_\C z^i \bar{z}^j |z-1|^{-2\gamma}e^{-Nt^2|z|^2} \d A(z) \right] >0,
\end{equation}
where the equality follows from the fact that the columns of the two matrices are related by a unimodular triangular matrix, while the inequality follows from the positivity of the measure. Finally, since $\Gamma(z)$ has no zeros (and no poles since $j-\gamma+1>0$), the non-vanishing follows from \eqref{detcont}. 
\end{proof}

 Define the matrix 
\begin{equation}
\label{Ymatrix}
Y(z) = \begin{bmatrix}
\pi_{k}(z) & \dfrac{\strut1}{\strut2 \pi i}  \displaystyle{\int_{\Sigma}} \dfrac{\strut \pi_{k}(z')}{\strut z' - z} w_{k}(z') dz'\\
-{2 \pi i } \Pi_{k-1}(z) &- \displaystyle{\int_{\Sigma}} \dfrac{\strut\Pi_{k-1}(z')}{\strut z' - z} w_{k}(z') dz' 
\end{bmatrix}\ .
\end{equation}
It is easy to verify that the matrix $Y(z)$ is the \emph{unique} solution of the following Riemann-Hilbert problem  of  Fokas--Its--Kitaev-type \cite{FIK}:
\begin{enumerate}
\item Piecewise analyticity:
\be
\label{RH_Y0}
Y(z)  \text{ is analytic in } \C\setminus \Sigma \text{ and both limits } Y_{\pm}(z) \text{ exist along } \Sigma,
\ee
\item Jump on $\Sigma$:
\begin{equation}
\label{RH_Y11}
Y_{+}(z) = Y_{-}(z)\begin{pmatrix}1& w_k(z) \\ 0 & 1 \end{pmatrix}\ , \quad z \in \Sigma,
\end{equation}
\item Behaviour at $z=\infty$:
\begin{equation}
\label{RH_Y2}
Y(z) = \left(I+{\mathcal O}\left(\frac{1}{z}\right)\right)z^{k\sigma_3}\ ,\quad z \to \infty.
\end{equation}
\end{enumerate}
%

The last relation is obtained by noticing that if $\Pi_{k-1}$ is  given by \eqref{Pidet} then  the entry $Y_{22}(z)$ of the matrix $Y(z)$ in 
\eqref{Ymatrix} satisfies 
\begin{equation}
- \int_{\Sigma} \frac{\Pi_{k-1}(z')}{z' - z} w_{k}(z') dz' = z^{-k} \left(1 + \order{\frac{1}{z}}\right)\ ,\quad z \to \infty.
\end{equation}

\subsection{Initial undressing step}
In order to simplify the subsequent analysis, we define the following modified matrix:
\be
\label{Ytilde}
\tilde Y(z) := Y(z)\left(1-\frac{1}{z}\right)^{-\frac{\gamma}{2}\sigma_3} \qquad z \in \C\setminus (\Sigma\cup [0,1])\,.
\ee
%
This matrix-valued function  $\tilde Y(z) $ satisfies the following Riemann-Hilbert problem:
\begin{enumerate}
\item Piecewise analyticity:
\be
\tilde Y(z)\text{ is analytic in }\C\setminus(\Sigma\cup [0,1])
\ee
\item Jumps on $\Sigma$ and $[0,1]$:
\[
\tilde Y_{+}(z) = \tilde Y_{-}(z)\left\{
\begin{array}{cl}
\displaystyle \begin{pmatrix}1&  e^{-kV(z)} \\ 0 & 1 \end{pmatrix} &\quad  z \in \Sigma,\\
&\\
\displaystyle e^{-\gamma\pi i\sigma_3} & \quad  z \in (0,1),
\end{array}\right .
\]
with $V(z)$ as \eqref{def_V}.
\item Behaviour at $z=\infty$:
\be
\tilde Y(z) = \left(I+{\mathcal O}\left(\frac{1}{z}\right)\right)z^{k\sigma_3}\ ,\quad z \to \infty\ .
\ee
\item Endpoint behaviour at $z=0$ and at $z=1$:
\be
\tilde Y(z)z^{-\frac{\gamma}{2}\sigma_3} = \order{1} \quad z \to 0\ , \qquad \tilde Y(z)(z-1)^{\frac{\gamma}{2}\sigma_3} = \order{1} \quad z \to 1\ .
\ee
\end{enumerate}
The polynomials $\pi_k(z)$ is recovered from $\tilde Y(z)$ as
\[
\pi_k(z)=\tilde{Y}_{11}(z)\left(1-\frac{1}{z}\right)^{\frac{\gamma}{2}} .
\]

\section{Asymptotic analysis in the pre-critical case}
\label{sec:pre-crit}
In order to analyse the large $k$ behaviour of $\tilde{Y}(z)$ we use the Deift--Zhou nonlinear steepest descent method \cite{DeiftZhou}.
The first step to study the large $k$ behaviour of the matrix function $\tilde{Y}(z)$ is to make a transformation  $\tilde{Y}(z)\to U(z)$ so that the Riemann--Hilbert problem for $U(z)$ is  normalised to the identity as $|z|\to \infty$.
For this purpose we introduce a  contour $\mathcal{C}$ homotopically equivalent to $\Sigma$  in   $\C\setminus [0,1]$  and a function $g(z)$ analytic off $\mathcal{C}$.
Both the contour $\mathcal{C}$ and the function $g(z)$ will be determined later.
We assume that the function $g(z)$ is of the form
\begin{equation}
\label{g}
g(z)=\int_{\mathcal{C}}\log(z-s)d\nu(s),
\end{equation}
where $d\nu(s)$ is a positive  measure  with support on $\mathcal{C}$ such that
\[
\int_{\mathcal{C}}d\nu(s)=1.
\]
With this assumption clearly one has
\begin{equation}
\label{gasymp}
g(z)=\log z+\order{z^{-1}}\quad \mbox{az $|z|\to\infty$},
\end{equation}
where the logarithm is branched on the positive real axis.
\subsection{First transformation $\tilde{Y}\to U$}
Since $\mathcal{C}$ is homotopically equivalent to $\Sigma$ in $\C\setminus [0,1]$  we can deform the contour $\Sigma$ appearing in the Riemann--Hilbert problem for $\tilde{Y}$ to  $\mathcal{C}$. Define the modified matrix
\be
U(z) = e^{-k(\ell/2)\sigma_3}\tilde Y(z)e^{-kg(z)\sigma_3}e^{k(\ell/2)\sigma_3} \qquad z \in \C\setminus (\mathcal{C}\cup [0,1]),
\ee
where $\ell$ is a real number, to be determined below.
Then $U(z)$ solves the following RHP problem 
\begin{enumerate}
\item Piecewise analyticity:
\be
U(z)\text{ is analytic in }\C\setminus(\mathcal{C}\cup [0,1]).
\ee
\item Jump discontinuity on $\mathcal{C}$:
\be
\label{jU1}
U_{+}(z) = U_{-}(z)\begin{pmatrix}e^{-k(g_+-g_-)}&  e^{ k(g_++g_--\ell-V)}\\ 0 & e^{k(g_+-g_-)} \end{pmatrix}\ ,
\ee
\item Jump discontinuity on $(0,1)$:
\be
\label{jU2}
U_{+}(z) = U_{-}(z)e^{-\gamma\pi i\sigma_3}\ ,
\ee
\item Endpoint behaviour at $z=0$ and $z=1$:
\be
\label{jU3}
U(z)z^{-\frac{\gamma}{2}\sigma_3} = \order{1} \quad z \to 0\ , \qquad U(z)(z-1)^{\frac{\gamma}{2}\sigma_3} = \order{1} \quad z \to 1\ .
\ee
\item Large $z$ boundary behaviour:
\be
\label{jU4}
U(z) = I+{\mathcal O}\left(\frac{1}{z}\right)\ ,\quad z \to \infty\ .
\ee
\end{enumerate}
The polynomial $\pi_k(z)$ is determined  from $U(z)$ by
\be
\label{pi_U}
\pi_k(z)=U_{11}(z)e^{kg(z)}\left(1-\frac{1}{z}\right)^{\frac{\gamma}{2}} .
\ee
\subsubsection{The choice of $g$-function}
In order to  determine the function $g$ and the contour $\mathcal{C}$ we   impose that the  jump matrix \eqref{jU1} becomes purely oscillatory for large $k$.  This is accomplished if  the following conditions are satisfied:
\begin{equation}
\label{eq_g}
\begin{split}
g_+(z)+g_-(z)-\ell-V(z)&=0\\
\Re(g_+(z)-g_-(z))&=0
\end{split}\qquad z \in \mathcal{C}.
\end{equation}
Next we show that we can find  a function $g$ and a contour $\mathcal{C}$ that satisfy the conditions \eqref{eq_g}.
For the purpose we use the following elementary result in the theory of boundary value problems.
\begin{lemma}[\cite{Gakhov}, p.~78]
\label{lemmaGakhov}
Let $L$ be a simple closed contour dividing the complex plane in two regions $D_+$ and $D_-$ where $D_+=\inte(L)$ and $D_-=\exte(L)$. Suppose that a function $\psi(\zeta)$ defined on $L$  can be represented  in the form
\[
\psi(\zeta)=\psi_+(\zeta)+\psi_-(\zeta),\quad \zeta\in L
\]
where $\psi_+(\zeta)$ is  the boundary value  of a function $\psi_+(z)$ analytic for $z\in D_+$ and $\psi_-(\zeta)$ is the boundary value of a function $\psi_-(z)$  analytic for $z\in D_-$ and such that $\psi_-(\infty)=0$.
Then the Cauchy integral
\[
\Phi(z)=\dfrac{1}{2\pi i}\int_{L}\dfrac{\psi(\zeta)}{\zeta-z}d\zeta, 
\]
can be represented in the form
\begin{equation}
\label{eqPhi}
\begin{split}
&\Phi_+(z)=\psi_+(z),\quad \text{ for $z\in D_+$},\\
&\Phi_-(z)=-\psi_-(z), \quad \text{for $z\in D_-$}.
\end{split}
\end{equation}
The boundary values of the function $\Phi$ on the two sides of the contour $L$  then satisfy
\begin{equation}
\begin{split}
\label{Phi2}
&\Phi_+(\zeta)+\Phi_-(\zeta)=\psi_+(\zeta)-\psi_-(\zeta)\\
&\Phi_+(\zeta)-\Phi_-(\zeta)=\psi_+(\zeta)+\psi_-(\zeta)
\end{split}
\qquad \zeta \in L.
\end{equation}
\end{lemma}
Now we apply Lemma~\ref{lemmaGakhov} to  the function $g'(z)$ that satisfies the differentiated boundary condition
\[
g_+'(\zeta)+g'_-(\zeta)=V'(\zeta)=\dfrac{1}{z_0}+\dfrac{1}{\zeta},\quad \zeta\in\mathcal{C}.
\]
The functions $\psi_\pm(z)$  of Lemma \ref{lemmaGakhov} are $\psi_+(z)=\dfrac{1}{z_0}$ and $\psi_-(z)=-\dfrac{1}{z}$ and therefore the function $g'(z)$ takes the form
\be
\label{gprime}
g'(z)=\dfrac{1}{ 2\pi i }\int_{\mathcal{C}}\dfrac{  \frac{1}{\tau}-\frac{1}{z_0}}{z-\tau}d\tau=
\left\{
\begin{array}{ll}
\displaystyle \frac{1}{z_0} & z \in \inte(\mathcal{C})\\
&\\
\displaystyle \frac{1}{z} & z \in \exte(\mathcal{C})\,,
\end{array}
\right.
\ee
so that the measure $d\nu$ in \eqref{g} is given by
\be
\label{dnu}
d\nu(z)=\dfrac{1}{2\pi i}\left(\dfrac{1}{z}-\dfrac{1}{z_0}\right)dz,\quad z\in\mathcal{C}.
\ee
Integrating the  relation  \eqref{gprime} and using \eqref{gasymp} one has
\be
\label{gint}
g(z)=
\left\{
\begin{array}{ll}
\displaystyle \frac{z}{z_0}+\ell & z \in \inte(\mathcal{C})\\
&\\
\displaystyle \log z & z \in \exte(\mathcal{C})\ 
\end{array},
\right.
\ee
where $\ell $ is an integration constant and $\log z$ is analytic in $\mathbb{C}\backslash \mathbb{R}^+$. Performing the integral in \eqref{g} for a specific value of $z\in\mathcal{C}$, say $z = 0$  and deforming $\mathcal{C}$ to a circle of radius $r$ one can determine the value of $\ell$:
\begin{equation}
\label{def_l}
\ell=\log r-\dfrac{r}{z_0},\quad r>0.
\end{equation}
The total integral of  $d\nu$ in \eqref{dnu} is normalized to one on any closed contour containing the point $z=0$. However we have to define the contour $\mathcal{C}$ so that $d\nu(z)$ is a real and positive measure along $\mathcal{C}$.
To this end we introduce  the function
\be
\label{def_phir}
\phi_r(z)= 
\left\{
\begin{array}{ll}
-2g(z)+V(z)+\log r-\dfrac{r}{z_0}=\log z-\dfrac{z}{z_0}-\log r+\dfrac{r}{z_0}& z \in \inte(\mathcal{C})\backslash[0,r]\\
&\\
2g(z)-V(z)-\log r+\dfrac{r}{z_0}=\log z-\dfrac{z}{z_0}-\log r+\dfrac{r}{z_0}& z \in \exte(\mathcal{C})\backslash(r,\infty)\ .
\end{array}
\right.
\ee
Observe that $\phi_r(z)$ is analytic  across $\mathcal{C}$, namely $\phi_{r+}(z)=\phi_{r-}(z)$, $z\in\mathcal{C}$ and that 
\be
\label{def_phir1}
\phi_r(z)=\dfrac{\phi_{r+}(z)+\phi_{r-}(z)}{2}=-g_+(z)+g_-(z),\quad z\in\mathcal{C}.
\ee
The next identity follows in a trivial way
\[
d\nu(z)=\dfrac{1}{2\pi i} d\phi_r(z).
\]

Imposing the relation \eqref{eq_g} on  \eqref{def_phir1} one has 
\be
\Re(\phi_r(z))=\log |z|-\frac{\Re(z)}{z_0}-\log r+\dfrac{ r}{z_0}=0.
\ee
This equation defines a family of contours which are closed for
$|z|\leq r$ (easy to verify). Since the function $-\log r+\dfrac{ r}{z_0}$ for $r>0$  has a single minimum at $r=z_0$ and diverge to $+\infty$ for $r\to 0$ and $r\to\infty$,  it is sufficient to consider only  the values $0<r\leq z_0$.

We define the contour $\mathcal{C}_r$  associated to $r$ as  (see Figure~\ref{fig:tikz_jump_Tam})
\begin{equation}
\label{def_Gammar}
\mathcal{C}_r=\{z\in\mathbb{C}\ \colon\ \Re(\phi_r(z))=0,\;\;|z|\leq z_0\},\quad  0< r\leq z_0.
\end{equation}
\begin{wrapfigure}{lhtb!}{0.4\textwidth}
\includegraphics[width=0.4\textwidth]{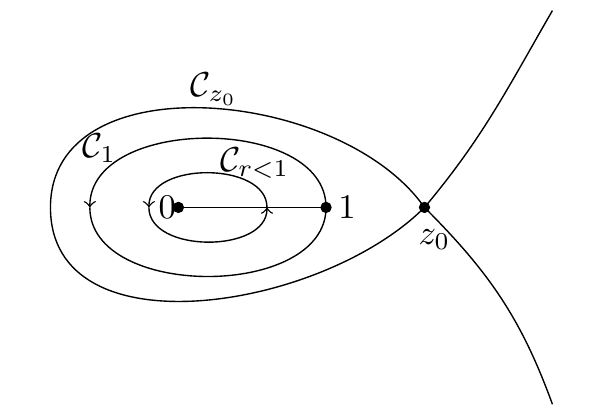}
\caption{The family of contours $\mathcal{C}_r$ for three values of $r$:  $r=z_0$, $r=1$ and $r<1$. For $r=z_0$, the region $|z|>z_0$ is plotted as well.}
\label{fig:tikz_jump_Tam}
\end{wrapfigure}
Since $ \partial_x\Re\phi_r|_{y=0}=\dfrac{|x|}{x^2}-\dfrac{1}{z_0}>0$, for $|x|<z_0$,  $\Re\phi_r(z)$ is negative inside $\mathcal{C}_r$.
Note that the point $z=0$ lies inside $\mathcal{C}_r$ since $\Re\,\phi(z)\to-\infty$ for $|z|\to 0$. Furthermore $\mathcal{C}_r$ intersects the real line in the points $z=r$ and $z\in(-r,0)$. Indeed $\Re\,\phi(-r)>0$ which shows that the point $z=-r$ lies outside $\mathcal{C}_r$.

\begin{lemma}
\label{lemma1b}
The a-priori complex measure $\nu$ in \eqref{dnu}  is a probability measure on the contour $\mathcal{C}_r$ defined in \eqref{def_Gammar} for $0<r\leq z_0$.
\end{lemma}
\begin{proof}
By the Residue Theorem the measure  $d\nu$ has total mass $1$ on any contour $\mathcal{C}_r$. Since $\Re\phi_r=0$ on the contour $\mathcal{C}_r$ it follows that the measure $d\nu=\dfrac{1}{2\pi i}d\phi_r$ is real on the contour $\mathcal{C}_r$.
In order to show that the measure is positive on the contour $\mathcal{C}_r$ we introduce 
 the variable 
\[
\psi_r=e^{\phi_r},
\] 
 On the contour $\mathcal{C}_r$ we have that $ |\psi_r|=1$, so that in the $\psi_r$--plane the contour $\mathcal{C}_r$ is  mapped to a circle of radius one  and the map $\psi_r=e^{\phi_r}$ is a univalent conformal map from the interior of $\mathcal{C}_r$ to the interior of a circle \cite{Szego}.
 We have 
 \[
d\nu=\dfrac{1}{2\pi i }\dfrac{d \psi_r}{\psi_r},
 \] 
 which shows that the measure $d\nu$  in the variable $\psi_r$ is a uniform  measure on the circle and therefore $d\nu$ is a positive measure on each contour $\mathcal{C}_r$.
\end{proof}\\
 We are now ready to prove Lemma~\ref{lemma1} and Lemma~\ref{lemma_balayage} announced  in the Introduction.\\
 
 \noindent
 {\bf Proof of Lemma~\ref{lemma1}}.
By using the Residue Theorem, it is straightforward to check that  $d\hat{\nu}$ is a probability measure on any contour $\hat{\mathcal{C}}_r$ defined in \eqref{Gamma_r}.
In the post-critical case such curve $\hat{\mathcal{C}}_r $ has  $s$ connected components $\hat{\mathcal{C}}^j_r$,  $j=0,\dots,s-1$, each  encircling exactly one root of the equation  $\lb^s=t$. In the pre-critical case we denote with the same symbol $\hat{\mathcal{C}}^j_r$ the connected components of $\hat{\mathcal{C}}_r\setminus\{0\}$ together with $\{0\}$.

With the appropriate choice of the branch of the $s^{\text{th}}$ root, any $\hat{\mathcal{C}}_{r}^{j}$ can be parametrised by the map $\lambda(z):=(-t(z-1))^{\frac{1}{s}}$, which respects the orientations of $\mathcal{C}$ and $\hat{\mathcal{C}}^{j}_{r}$ . 
Since $\phi(z)=\hat{\phi}(\lambda(z))$, equation \eqref{nur} implies
\[ d\hat{\nu}(\lambda(z))=\frac{1}{s}d\nu(z)\ .\]
This relation, together with Lemma~\ref{lemma1b}, where we proved that $d\nu(z)$ is real and positive on $\mathcal{C}_{r}$, implies that also $d\hat{\nu}$ is real and positive on any $\hat{\mathcal{C}}^{j}_{r}$ and hence on the whole $\hat{\mathcal{C}}_{r}$.
\hfill{\bf Q.E.D.}\\

\noindent {\bf Proof of Lemma~\ref{lemma_balayage}}.
By using the Residue Theorem we first calculate the r.h.s. of \eqref{balayage00}  obtaining
\be
\label{Right}
\int_{\hat{\mathcal{C}}}\dfrac{d\hat{\nu}(\eta)}{\lambda-\eta}=\dfrac{1}{2\pi i }\int_{\hat{\mathcal{C}}}\dfrac{\eta^{s-1}}{\lambda-\eta}\left(\dfrac{1}{\eta^s-t}+\dfrac{1}{z_0t}\right)d\eta=\dfrac{\lb^{s-1}}{\lb^s-t}.
\ee
Therefore  \eqref{balayage00} is equivalent to
\[
\dfrac{s\lb^{s-1}(\bar{\lb}^s-t)}{T}=\dfrac{\lb^{s-1}}{\lb^s-t},\quad\mbox{or}\quad (\bar{\lb}^s-t)(\lb^s-t)=\dfrac{T}{s}=t_c^2,
\]
which is equivalent to the equation for the  boundary of $D$ defined in \eqref{domainmu}.
Next we show that the l.h.s. of \eqref{balayage1} is equal to its r.h.s. By using Stokes' Theorem, the relation \eqref{Schwartz} and the Residue Theorem one obtains
\begin{align}
&\int_{D}\dfrac{d\mu^*(\eta)}{\lambda-\eta}=\dfrac{1}{2\pi i  t_c^2}\int_{\partial D}\dfrac{\eta^{s-1}\bar{\eta}^sd\eta}{\lb-\eta}=\dfrac{1}{2\pi i  t_c^2}\int_{\partial D}\dfrac{\eta^{s-1}S(\eta)d\eta}{\lb-\eta}
=\dfrac{\lb^{s-1}}{\lb^s-t}
\end{align}
which coincides with \eqref{Right}.
\hfill{\bf Q.E.D.}

\subsubsection{Choice of the contour}

\begin{wrapfigure}[17]{lhtb!}{0.4\textwidth}
\includegraphics[width=0.4\textwidth]{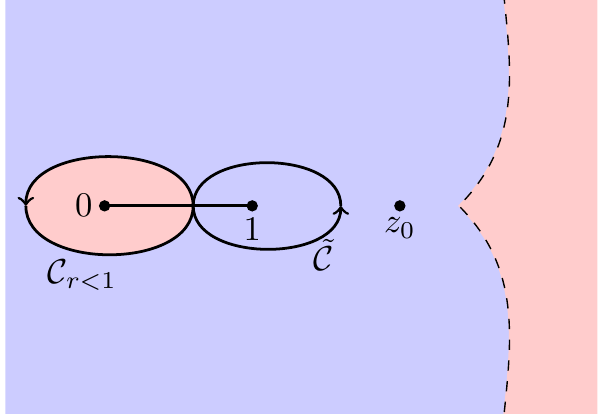}
\caption{The   contour  $\mathcal{C}_{r<1}\cup\tilde{\mathcal{C}}$ that is homotopic to $\Sigma$. The region where $\Re \phi(z)<0$ is coloured in red.}
\label{fig:tikz_jump_Cont}
\end{wrapfigure}
   We now argue that the relevant contour on which the zeros of the orthogonal polynomials $\pi_k(z)$  accumulate in the pre-critical case is given by the level $r=1$.
    The family of contours $0<r<1$ can be immediately ignored because in this case $\Sigma$ has to be deformed to two contours 
    $\Sigma\simeq \mathcal{C}_{r<1}\cup\tilde{\mathcal{C}}$ as shown in Figure~\ref{fig:tikz_jump_Cont}.
On the contour $\tilde{\mathcal{C}}$ we have  $\Re(\phi)>0$  and it is not possible to perform any contour deformation to get exponentially small terms in the jump matrix
\eqref{jU1} as $k\to\infty$.

For $1< r\leq  z_0$ the original contour $\Sigma$ can be deformed homotopically to  the contours $\mathcal{C}_r$. In this case the asymptotic expansion  of  the matrix entry $Y_{11}(z)=\pi_k(z)$ as $k\to \infty$, does not give any sub-leading contribution.
For this reason we discard this case and we omit the corresponding  asymptotic analysis.
The only possible case remains  $r=1$.
In this case  the  analysis as $k\to\infty$ to the matrix entry  $Y_{11}(z)=\pi_k(z)$  which will performed in subsections~\ref{Sec3.2} to \ref{Sec3.5},  
 gives leading and sub-leading terms. The comparison of these two terms enable us to locate the zeros of the polynomial $\pi_k(z)$ on a contour that lies within a distance $\order{\log k/k}$ from the contour $\mathcal{C}_{r=1}$.

So for the reasons explained above, we are going to perform the asymptotic analysis of the RH-problem \eqref{jU1}-\eqref{jU4} by deforming the contour $\Sigma$ to the contour $\mathcal{C}_{r=1}$. For simplicity we denote this contour by $\mathcal{C}$:
\begin{equation}
\label{def_Gamma1}
\mathcal{C}=\left\{z\in\mathbb{C}\ \colon\ \log |z|-\frac{\Re(z)}{z_0}+\dfrac{ 1}{z_0}=0,\;\;|z|\leq 1\right\} 
\end{equation}
\begin{wrapfigure}{lhtb!}{0.3\textwidth}
\includegraphics[width=0.3\textwidth]{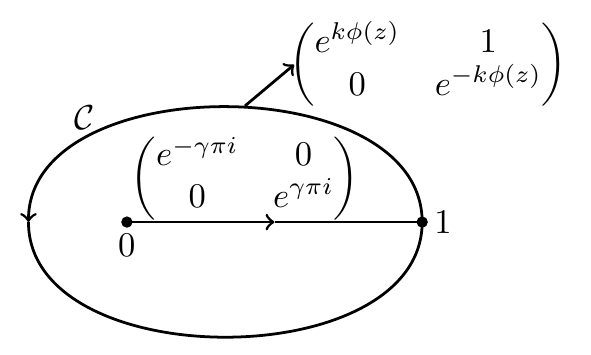}
\caption{The jump matrices for $U(z)$}
\label{fig:jump_U}
\end{wrapfigure}
Also, for simplicity, the function $\phi_{r=1}$ will be referred to as $\phi$ below:
\be
\label{def_phi}
\phi(z)=\log z-\dfrac{z}{z_0}+\dfrac{1}{z_0},
\ee
where $\log z$ is analytic in $\mathbb{C}\backslash (-\infty,0]$.
With this choice of contour, the  jump matrices of the RH problem \eqref{jU1}-\eqref{jU4}  for the matrix $U(z)$  are summarized in  Figure \ref{fig:jump_U}.

\subsection{The second transformation $U \mapsto T$ \label{Sec3.2} }
Consider two extra loops $\mathcal{C}_{i}$ and $\mathcal{C}_{e}$ as shown in Figure \ref{fig:jump_T}.
These define new domains $\Omega_{0}$, $\Omega_{1}$, $\Omega_{2}$ and $\Omega_{\infty}$.
Define the new matrix-valued function
\be
\label{def_T}
T(z)
=\left\{
\begin{array}{ll}
U(z) & z \in \Omega_{\infty}\cup \Omega_0\\
U(z)\begin{pmatrix}1&  0 \\ -e^{k\phi(z)} & 1 \end{pmatrix} & z \in \Omega_{1}\\
U(z)\begin{pmatrix}1&  0 \\ e^{-k\phi(z)} & 1 \end{pmatrix} & z \in \Omega_{2}\ .
\end{array}
\right.
\ee
Then $T(z)$ satisfies the following Riemann--Hilbert problem:\\

\begin{enumerate}
\item Piecewise analyticity:
\be
T(z) \text{ is analytic in }\C \setminus\left(\mathcal{C}_{e}\cup \mathcal{C} \cup \mathcal{C}_{i} \cup [0,1]\right) 
\ee
\end{enumerate}

\begin{enumerate}
\setcounter{enumi}{1}
\item Jump discontinuities on $\Sigma_T= \mathcal{C}_{e}\cup \mathcal{C} \cup \mathcal{C}_{i} \cup [0,1]$:
\be
T_{+}(z)=T_{-}(z)v_{T}(z),\quad z\in\Sigma_T,
\ee
where
\be
v_{T}(z)=\left\{
\begin{array}{cl}
\begin{pmatrix}0&1\\-1&0\end{pmatrix}  & z \in \mathcal{C}\\
\begin{pmatrix}1& 0 \\ e^{-k\phi(z)} & 1 \end{pmatrix}  & z \in \mathcal{C}_{e}\\
\begin{pmatrix}1& 0 \\ e^{k\phi(z)} & 1 \end{pmatrix}  & z \in \mathcal{C}_{i}\\
e^{-\gamma\pi i\sigma_3}  & z \in (0,1)\,.
\end{array}
\right.
\ee
\item Endpoint behaviour at $z=0$ and $z=1$:
\be
T(z)z^{-\frac{\gamma}{2}\sigma_3} = \order{1} \quad z \to 0\ , \qquad T(z)(z-1)^{\frac{\gamma}{2}\sigma_3} = \order{1} \quad z \to 1\ .
\ee
\item Large $z$ boundary behaviour:
\be
T(z) = I+{\mathcal O}\left(\frac{1}{z}\right)\ ,\quad z \to \infty\ .
\ee
\end{enumerate}
Note that $e^{k\phi(z)}$ is analytic in $\C\setminus \{0\}$\ .
The important feature of this Riemann--Hilbert problem  is that the jumps are either constant or they tend to the identity matrix as $k \to\infty$ at an exponential rate.
\begin{proposition} There exists a constant $c_0>0$ so that 
\[
v_T(z)=I+\order{e^{-c_0k}}\quad \mbox{as $k\to\infty$}
\]
uniformly for $z\in\mathcal{C}_e\cup\mathcal{C}_i\backslash {\cal U}_{1}$, where ${\cal U}_{1}$ is a small neighbourhood of $z=1$.
\end{proposition}
\begin{wrapfigure}[12]{rhtb!}{0.45\textwidth}
\includegraphics[width=0.45\textwidth]{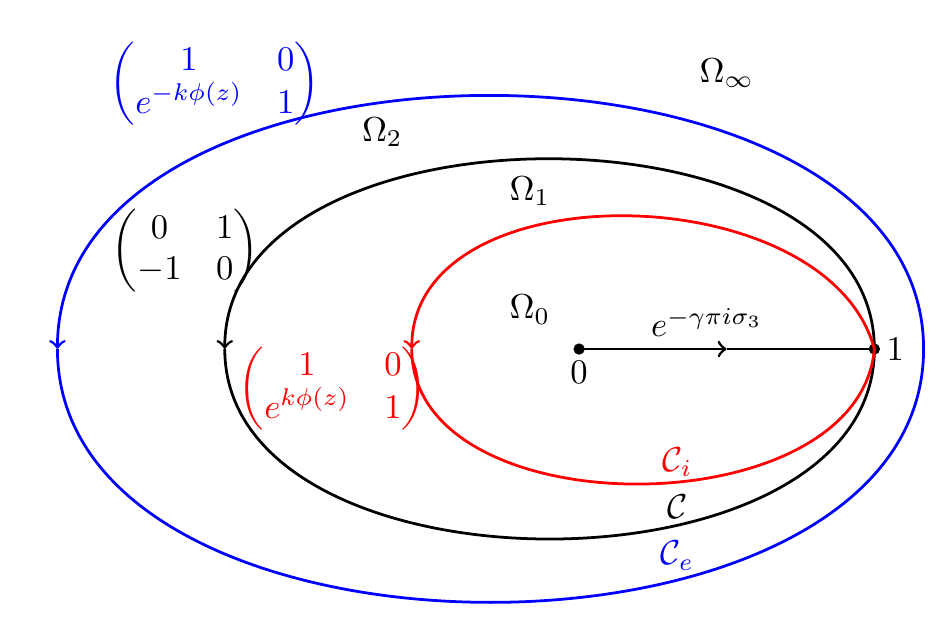}
\caption{The jump matrices for $T(z)$ with contour $\Sigma_T=\mathcal{C}_e\cup\mathcal{C}\cup\mathcal{C}_i\cup(0,1)$.}
\label{fig:jump_T}
\end{wrapfigure}
The proof of this proposition follows immediately from the fact that $\Re\phi(z)<0$   for $z\in\inte(\mathcal{C})\backslash {\cal U}_1$ and
$\Re\phi(z)>0$ for  $z\in\exte(\mathcal{C})\backslash {\cal U}_1$. Here $\inte(\mathcal{C})$ is the interior region bounded by 
$\mathcal{C}$ and $\exte(\mathcal{C})$ is the exterior region bounded by $\mathcal{C}$.

From the above proposition it follows that 
\[
v_T(z)\to v^{\infty}(z)\quad \mbox{as $k\to\infty$},
\]
exponentially fast, where
\begin{equation}
\label{vinfinity_pre}
v^{\infty}(z)=\left\{
\begin{array}{cl}
\begin{pmatrix}0&1\\-1&0\end{pmatrix}&\mbox{ as $z\in\mathcal{C}$},\\
\\
e^{-\gamma\pi i \sigma_3}& \mbox{ as $z\in(0,1)$}\\
\\
I& \mbox{ as $z\in\mathcal{C}_e\cup\mathcal{C}_i$}.
\end{array}\right.
\end{equation}

\subsection{The outer parametrix for large $z$}

We need to find a matrix-valued function $P^{\infty}(z)$ analytic in $\C\setminus(\mathcal{C} \cup[0,1])$ such that it has jump discontinuities given by the matrix $v^{\infty}(z)$ in \eqref{vinfinity_pre}, namely
\be
P^{\infty}_{+}(z) = P^{\infty}_{-}(z)
\left\{\begin{split}
\begin{pmatrix}0& 1 \\ -1 & 0 \end{pmatrix}& \qquad z \in \mathcal{C}\\
e^{-\gamma\pi i\sigma_3} & \qquad z \in (0,1)\ ,
\end{split}\right.
\ee
with endpoint behaviour
\be
P^{\infty}(z)z^{-\frac{\gamma}{2}\sigma_3} = \order{1} \quad z \to 0\ , \qquad P^{\infty}(z)(z-1)^{\frac{\gamma}{2}\sigma_3} = \order{1} \quad z \to 1\ ,
\ee
and large $z$ boundary behaviour
\be
P^{\infty}(z) = I+{\mathcal O}\left(\frac{1}{z}\right)\ ,\quad z \to \infty.
\ee
Define
\be
\tilde P^{\infty}(z) := P^{\infty}(z)\chi^{-1}(z)\,
\ee
where 
\be
\label{chi}
\chi(z):= \left\{\begin{array}{cl}
\begin{pmatrix}0& 1 \\ -1 & 0 \end{pmatrix}& z \in \inte(\mathcal{C})\\
I& z \in \exte(\mathcal{C})\ .
\end{array}\right.
\ee

The $k$-independent matrix $\tilde P^{\infty}(z)$ has no jump on $\mathcal{C}$ and it satisfies the following Riemann--Hilbert problem: 
\begin{enumerate}
\item Piecewise analyticity:
\be
\tilde P^{\infty}\text{ is holomorphic in }\C\setminus (0,1)
\ee
\item Jump discontinuity on $(0,1)$:
\be
\tilde P_{+}^{\infty}(z)=\tilde P_{-}^{\infty}(z)e^{\gamma\pi i\sigma_3}
\ee
\item 
Boundary behaviour at $z=\infty$:
\be
\tilde P^{\infty}(z) = I+{\mathcal O}\left(\frac{1}{z}\right)\ ,\quad z \to \infty\ .
\ee
\end{enumerate}
The  solution to this Riemann--Hilbert problem is given by
\be
\tilde P^{\infty}(z) = \left(1-\frac{1}{z}\right)^{\frac{\gamma}{2}\sigma_3}
\ee
which leads to the particular solution
\be
\label{eq:global_parametrix_pre}
P^{\infty}(z) = \left(1-\frac{1}{z}\right)^{\frac{\gamma}{2}\sigma_3}\chi(z)= 
\left\{\begin{array}{ll}
\displaystyle \left(1-\frac{1}{z}\right)^{\frac{\gamma}{2}\sigma_3}\begin{pmatrix}0& 1 \\ -1 & 0 \end{pmatrix}& z \in \inte(\mathcal{C})\\
\displaystyle \left(1-\frac{1}{z}\right)^{\frac{\gamma}{2}\sigma_3} & z\in \exte(\mathcal{C})\ .
\end{array}\right.
\ee

\subsection{The local parametrix at $z=1$}
The aim of this section is to construct a local parametrix $P^0(z)$  in a small  neighbourhood
${\mathcal U}_{1}$ of $z=1$ having  the same jump property as $T$  for $z$ near $1$  and matching  the outer parametric $P^{\infty}(z)$ 
in the limit $k\to \infty$ and $z\in\partial {\cal U}_{1}$. Then the Riemann-Hilbert problem for $P^0(z)$  is given by
\be
P^0_+(z)=P_-^0(z)v_{T}(z),\quad z\in {\mathcal U}_{1}\cap\Sigma_T,
\ee 
and 
\be
\label{match_pre}
P^0(z)=P^{\infty}(z)(I+o(1))\qquad \mbox{as $k\to\infty$ and $z\in \partial {\mathcal U}_{1}$. }
\ee


In order to build such  local parametrix near the point $z=1$ we first construct a new matrix function $B(z)$ from $P^0(z)$:
\be
\label{V_pre}
B(z)=P^0(z)\chi(z)^{-1}\mathcal{Q}(z)
\ee
where
\be
\label{def_Q}
\mathcal{Q}(z):=\left\{\begin{array}{lc}\begin{pmatrix}1& e^{k\phi(z)} \\0 & 1 \end{pmatrix} & z\in \Omega_0\cap {\mathcal U}_{1} \\ I & z\in {\mathcal U}_{1}\setminus\Omega_0\end{array}  \right.
\ee
The matrix $B(z)$ satisfies the following jump relations in a neighbourhood of $1$:
\be
B_{+}(z)=B_{-}(z)\begin{pmatrix}e^{\gamma\pi i}& (e^{\gamma\pi i}-e^{-\gamma\pi i})e^{k\phi(z)}\\ 0 & e^{-\gamma\pi i} \end{pmatrix} \qquad z\in(-\infty,1)\,.
\ee

Next we construct a solution to the so called model problem, namely a Riemann-Hilbert problem that has the same jumps as $B$.

\subsubsection{Model problem}

Consider the model problem for the $2\times2$ matrix function $\Psi(\xi)$  analytic in $\mathbb{C}\backslash \mathbb{R}^{-}$
with boundary behaviour
\begin{align}
\label{model_pre1}
\Psi_{+}(\xi)&=\Psi_{-}(\xi)\begin{pmatrix}e^{\gamma\pi i}& (e^{\gamma\pi i}-e^{-\gamma\pi i})e^{\xi}\\ 0 & e^{-\gamma\pi i} \end{pmatrix}& \xi\in(-\infty,0)\\
\label{model_pre2}
\Psi(\xi)&=\left(I+\order{\frac{1}{\xi}}\right)\xi^{\frac{\gamma}{2}\sigma_{3}}& \xi\to\infty\,.
\end{align}
Defining 
\be
\tilde \Psi(\xi) := \Psi(\xi)\xi^{-\frac{\gamma}{2}\sigma_{3}}
\ee
we obtain the following Riemann--Hilbert problem for $\tilde \Psi(\xi)$:
\begin{align}
\label{model_pre11}
\tilde\Psi_{+}(\xi)&=\tilde\Psi_{-}(\xi)\begin{pmatrix}1& (1-e^{-2\gamma\pi i})(\xi^{\gamma})_{+}e^{\xi}\\ 0 & 1 \end{pmatrix}\qquad \xi\in(-\infty,0)\\
\label{model_pre21}
\tilde\Psi(\xi)&=I+\order{\frac{1}{\xi}} \qquad \xi\to\infty\,,
\end{align}
where the function $\xi^{\gamma}$ is analytic in $\mathbb{C}\backslash (-\infty, 0]$ and $(\xi^{\gamma})_{+}$ denotes the boundary from the right with respect to the contour $(-\infty,0)$ oriented in the positive  direction.
This is an abelian Riemann--Hilbert problem, and therefore easily solvable by the Sokhotski--Plemelj formula:
\be
\label{tildePsi}
\tilde \Psi(\xi)= \begin{pmatrix}1&   \displaystyle -\frac{1-e^{-2\gamma\pi i}}{2\pi i}\int_{\mathbb{R}^{-}}\frac{(\zeta^{\gamma})_+e^{\zeta}d\zeta}{\zeta-\xi}\\ 0 & 1 \end{pmatrix},
\ee
so that 
\be
\label{def_Lambda}
 \Psi(\xi)=\left[I+\sum_{j=1}^{\infty}\frac{\Psi_j}{\xi^j}\Lambda\right]\xi^{\frac{\gamma}{2}\sigma_3},\quad  \Lambda = \begin{pmatrix}0&   1\\ 0 & 0 \end{pmatrix}\,, \quad\mbox{as $\xi\to\infty$}.
\ee
In particular
\begin{equation}
\label{Psiexpansion_pre}
\Psi_1=\frac{e^{-2\gamma\pi i}-1}{2\pi i}\int_{\mathbb{R}^{-}}(\zeta^{\gamma})_+e^{\zeta}d\zeta=\frac{e^{-\gamma\pi i}-e^{\gamma\pi i}}{2\pi i}\int_{0}^{\infty}r^{\gamma}e^{-r}dr=-\frac{\sin(\gamma\pi)}{\pi}\Gamma(\gamma+1)=\dfrac{1}{\Gamma(-\gamma)}\ .
\end{equation}

\subsubsection{Construction of the parametrix}
We are now ready to specify the parametric $P^0(z)$ which follows from combining \eqref{def_Lambda} with \eqref{V_pre}, that is 
\be
P^0(z)=E(z)\Psi(kw(z))\mathcal{Q}(z)^{-1}\chi(z)\,,
\ee
where  ${\cal Q}(z)$ and $\chi(z)$ are defined in \eqref{def_Q}) and \eqref{chi} respectively and $E(z)$ is an  analytic  matrix  in a neighbourhood of $\mathcal{U}_{1}$ and $w(z)$ is a conformal mapping from a neighbourhood of $1$ to a neighbourhood of $0$.

The conformal map $w(z)$ is specified
by 
\be
w(z):=\left\{
\begin{array}{ll}
\phi(z) + 2\pi i& z \in \mathcal{U}_{1}\cap \C_{-}\\
\phi(z) & z \in \mathcal{U}_{1}\cap \C_{+}\ .
\end{array}
\right.
\ee
We observe that 
\be
w(z) = \left(1-\frac{1}{z_0}\right)(z-1)-\dfrac{1}{2}(z-1)^2+\order{(z-1)^3} \qquad z \to 1\ .
\ee
The matrix $E(z)$ is obtained from condition \eqref{match_pre} which, when combined with \eqref{Psiexpansion_pre} gives
 \begin{equation}
 \label{Pmatching_pre}
 \begin{split}
 &P^{\infty}(z)(P^0(z))^{-1}=P^{\infty}(z)\chi(z)^{-1}\mathcal{Q}(z)\Psi(kw(z))^{-1}E(z)^{-1}=\\
 &=P^{\infty}(z)\chi(z)^{-1}(kw(z))^{-\frac{\gamma}{2}\sigma_3}\left(I-\frac{1}{\Gamma(-\gamma)}\frac{\Lambda}{kw(z)}+
\order{k^{-2}}\right)E^{-1}(z),\quad k\to\infty,\quad z\in\partial{\cal U}_{1}
 \end{split}
 \end{equation}
 where we use the fact that ${\cal Q}(z)\to I$ exponentially fast as  $ k\to\infty,$  and $z\in\partial{\cal U}_{1}$.
 From the above expression it turns out that the matrix $E(z)$ takes the form
\be
\label{E_pre}
E(z)=P^{\infty}(z)\chi(z)^{-1}(kw(z))^{-\frac{\gamma}{2}\sigma_3}= \left(1-\frac{1}{z}\right)^{\frac{\gamma}{2}\sigma_3}(kw(z))^{-\frac{\gamma}{2}\sigma_3}.
\ee
We observe that the function $E(z)$ is single valued in a neighbourhood of ${\cal U}_{1}$. Indeed the jumps of $(z-1)^{\frac{\gamma}{2}\sigma_3}$ and $w(z)^{-\frac{\gamma}{2}\sigma_3}$ cancel each other. 

From the expression  \eqref{Pmatching_pre} and \eqref{E_pre} the matching between  $P^0(z)$ and $P^{\infty}(z)$ takes the form
\begin{equation}
\begin{split}
\label{P_match}
P^{\infty}(z)(P^0(z))^{-1}&=E(z)\left(I-\left(\sum_{j=1}^{M-1}\frac{\Psi_j}{(kw(z))^j}+\order{k^{-M}}\right)\begin{pmatrix}0&   1\\ 0 & 0 \end{pmatrix}\right)E^{-1}(z)\\
&=I-\left(\sum_{j=1}^{M-1}\frac{\Psi_j}{(kw(z))^j}+\order{k^{-M}}\right)E(z) \begin{pmatrix}0&   1\\ 0 & 0 \end{pmatrix} E^{-1}(z) \\
&=I-\left(1-\frac{1}{z}\right)^{\gamma}(k w(z))^{-\gamma}\left(\sum_{j=1}^{M-1}\frac{\Psi_j}{(kw(z))^j}+\order{k^{-M}}\right)\begin{pmatrix}0&   1\\ 0 & 0 \end{pmatrix}\qquad \text{as }k\to\infty
\end{split}
\end{equation}
for $z\in\partial{\cal U}_{1}$, where $\gamma\in(0,1)$. 

\subsubsection{Riemann-Hilbert problem for the error matrix $R$}
\label{Rmatrix}
\begin{wrapfigure}[14]{rhtb!}{0.45\textwidth}
\includegraphics[width=0.45\textwidth]{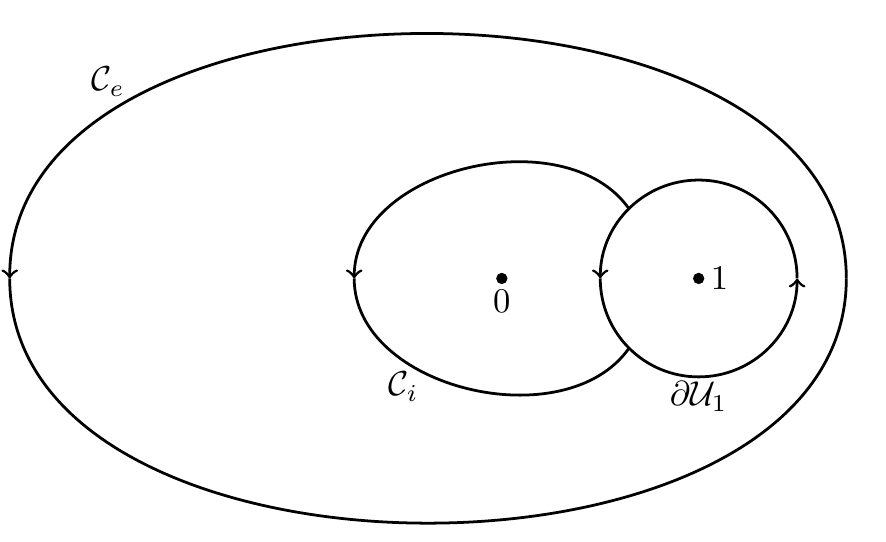}
\vspace*{-.4cm}
\caption{The jump contour structure $\mathcal{C}_R$}
\label{fig:jump_R}
\end{wrapfigure}
We now define the error matrix  $R$ in two  regions of the plane, using our approximations to the matrix  $T$.  Set
\begin{equation}\label{def R_pre}
R(z) =\begin{cases} T(z) \left( P^{(0)}(z) \right)^{-1} ,& z \in \mathcal{U}_{1}\ ,\\
 T(z) \left( P^{\infty}(z) \right)^{-1},& \mbox{ everywhere else.}
\end{cases}
\end{equation}
The matrix $R$ is piecewise analytic in $\mathbb{C}$ with a jump across  the contour $\mathcal{C}_R=\mathcal{C}_i\cup\mathcal{C}_e\cup \partial {\cal U}_1$ given in  Fig.~\ref{fig:jump_R}.
\subsubsection*{RH problem for $R$}
\begin{itemize}
\item[1.] $R$ is analytic in $\mathbb C\setminus\mathcal{C}_R$,
\end{itemize}
\clearpage
\begin{itemize}
\item[2.] For $z\in \mathcal{C}_R$, we have 
\be
\label{RHR_pre}
R_+(z)=R_-(z)v_R(z),
\ee
with
\be
\label{v_RU}
v_R(z) =\left\{
\begin{split}
P^{\infty}_-(z)v_T(z)\left(P_+^{\infty}(z)\right)^{-1} &\qquad z\in\mathcal{C}_R\backslash  \partial \mathcal{U}_{1}\\
P^{\infty}(z)\left(P^{(0)}(z)\right)^{-1} &\qquad z\in\partial  \mathcal{U}_{1}.
\end{split}\right.
\ee
\item[3.] As $z\to\infty$, we have
\begin{equation}
R(z)=I+\order{z^{-1}}.
\end{equation}
\end{itemize}

The jump matrices across the contour  $\mathcal{C}_R\setminus  \partial\mathcal{U}_{1}$ are all exponentially close to $I$ for large $k$ because $v_{T}$  converges  exponentially  fast to $v^{\infty}$ defined in \eqref{vinfinity_pre} and the product  $P^{\infty}_-(z)v^{\infty}(z) \left(P_+^{\infty}(z)\right)^{-1}=I$  with $P^{\infty}(z)$ defined in \eqref{eq:global_parametrix_pre}.
The only jump that is not exponentially small is the one on $\partial  \mathcal{U}_{1}$ since from \eqref{P_match} for any integer $M>2$ we have
\begin{align}
v_R(z)&=P^{\infty}(z)(P_{k}^{0}(z))^{-1}\\
&=I-\left(1-\frac{1}{z}\right)^{\gamma}(k w(z))^{-\gamma}\left(\sum_{j=1}^{M-1}\frac{\Psi_j}{(kw(z))^j}+\order{k^{-M}}\right)\Lambda,
\end{align}
with the shift matrix $\Lambda$ defined in \eqref{def_Lambda}.
By employing the notation
\[
v_R^{(j)}(z) =-\left(1-\frac{1}{z}\right)^{\gamma}\frac{\Psi_j}{(w(z))^{j+\gamma}},\qquad j=1,2,\dots,
\]
we can rewrite the matrix $v_R(z)$ in the form
\begin{equation}
\label{exp_vR}
v_R=I+\dfrac{1}{k^{\gamma}}
\left(
\sum_{j=1}^{M-1}\dfrac{v_R^{(j)}(z) }{ k^j }+\order{k^{-M}}\right)\Lambda
\end{equation}
where, in particular, \eqref{Psiexpansion_pre} implies
\begin{equation}
\label{exp_vR1}
\quad v_R^{(1)}(z)=-\dfrac{1}{\Gamma(-\gamma)}\left(1-\frac{1}{z}\right)^{\gamma}( w(z))^{-\gamma-1}.
\end{equation}
By a standard perturbation theory argument one has the expansion
\be
\label{exp_R}
R(z) = I +\dfrac{1}{k^{\gamma}}\left(\sum_{j=1}^{M-1}\frac{R^{(j)}(z)}{k^j}+\order{k^{-M}}\right) \quad k \to \infty\ ,
\ee
which gives, using   \eqref{RHR_pre}, \eqref{exp_vR} and \eqref{exp_R}
\be
R_{+}^{(1)}(z)=R_{-}^{(1)}(z)+v^{(1)}_R(z)\Lambda.
\ee

Therefore
\be
R^{(1)}(z) = \left[\frac{1}{2\pi i}\oint_{\partial \mathcal{U}_{1}}\frac{v_R^{(1)}(\zeta)d\zeta}{\zeta-z}\right] \Lambda\ 
\ee
where we observe that the function $v_R^{(1)}(z)$ has a simple pole in $z=1$ with expansion
\[
v_R^{(1)}(z)=-\dfrac{1}{\Gamma(-\gamma)} \frac{1}{z-1}\left(1-\frac{1}{z_0}\right)^{-\gamma-1}\left(1+\order{z-1}\right)\ .
\]
By a simple residue calculation we obtain
\be
\label{R1_pre}
R^{(1)}(z) = 
\left\{\begin{array}{ll}
\displaystyle  \frac{1}{(z-1)\Gamma(-\gamma)}\left(1-\frac{1}{z_0}\right)^{-\gamma-1}\begin{pmatrix}0&   1\\ 0 & 0 \end{pmatrix} & z \in \C\setminus \mathcal{U}_{1}\\
&\\
\displaystyle\frac{1}{(z-1)\Gamma(-\gamma)}\left(1-\frac{1}{z_0}\right)^{-\gamma-1}\begin{pmatrix}0&   1\\ 0 & 0 \end{pmatrix}+v_R^{(1)}(z)\begin{pmatrix}0&   1\\ 0 & 0 \end{pmatrix} & z \in \mathcal{U}_{1}\ .
\end{array}
\right.
\ee
In general, given the structure of the jump matrix \eqref{exp_vR} the error matrix $R^{(j)}(z)$  in \eqref{exp_R} is of the form
\be
\label{Rj}
R^{(j)}(z)=\begin{pmatrix}
0&*\\
0&0\end{pmatrix},
\ee
namely, only the $(1,2)$-entry of the matrices $ R^{(j)}(z)$ is  non-zero.
From \eqref{def R_pre} and \eqref{exp_R}  one has
\begin{equation}\label{T_exp}
T(z)=\left\{
\begin{array}{ll}
P^{(0)}(z) + \dfrac{1}{k^{\gamma}}\left(\displaystyle\sum\limits_{j=1}^{M-1}\dfrac{R^{(j)}(z)}{k^j}+\order{k^{-M}}\right) P^{(0)}(z)&\quad  z \in { \mathcal U}_1\\
P^{\infty}(z) +\dfrac{1}{k^{\gamma}}\left(\displaystyle\sum\limits_{j=1}^{M-1}\dfrac{R^{(j)}(z)}{k^j}+\order{k^{-M}}\right)  P^{\infty}(z)&\quad \mbox{ everywhere else,}
\end{array}\right.
\end{equation}
where, in particular, the first term $R^{(1)}(z)$ is given in  \eqref{R1_pre}.

\subsection{ Proof of Theorem~\ref{theorem2}: asymptotics for $p_n(\lb)$ for $z_0 >1$ \label{Sec3.5}}
In order to obtain the asymptotic  expansion of the polynomials $p_n(\lb)$ for $n\to\infty$, and $NT=n-l$ we first derive the asymptotic expansion of the reduced polynomials  
$\pi_k(z)$ for $k\to \infty$ with $n=sk+l$. For the purpose  we use \eqref{pi_U}, \eqref{def_T}, \eqref{R1_pre} and \eqref{T_exp} to obtain\begin{align}
\poly_k(z) &=e^{kg(z)}\left(1-\frac{1}{z}\right)^{\frac{\gamma}{2}}\left[U_k(z)\right]_{11}\\
&=
e^{kg(z)}\left(1-\frac{1}{z}\right)^{\frac{\gamma}{2}}\left\{
\begin{array}{ll}
\displaystyle \left[T_k(z)\right]_{11} & z \in \Omega_{\infty}\cup \Omega_{0}\\
&\\
\displaystyle \left[T_k(z)\begin{pmatrix}1& 0 \\ e^{k\phi(z)} & 1 \end{pmatrix}\right]_{11} & z \in \Omega_{1}\\
&\\
\displaystyle \left[T_k(z)\begin{pmatrix}1& 0 \\ -e^{-k\phi(z)} & 1 \end{pmatrix}\right]_{11} & z \in \Omega_{2}
\end{array}
\right.\\
&=
e^{kg(z)}\left(1-\frac{1}{z}\right)^{\frac{\gamma}{2}}\left\{
\begin{array}{ll}
\displaystyle \left[R(z)P^{\infty}(z)\right]_{11} & z \in \Omega_{\infty}\cup (\Omega_{0}\setminus \mathcal{U}_{1})\\
&\\
\displaystyle \left[R(z)P^{\infty}(z)\begin{pmatrix}1& 0 \\ e^{k\phi(z)} & 1 \end{pmatrix}\right]_{11} & z \in \Omega_{1}\setminus \mathcal{U}_{1}\\
&\\
\displaystyle \left[R(z)P^{\infty}(z)\begin{pmatrix}1& 0 \\ -e^{-k\phi(z)} & 1 \end{pmatrix}\right]_{11} & z \in \Omega_{2}\setminus \mathcal{U}_{1}\\
&\\
\displaystyle \left[R(z)P^{0}(z)\right]_{11} & z \in \Omega_{0} \cap \mathcal{U}_{1}\\
&\\
\displaystyle \left[R(z)P^{0}(z)\begin{pmatrix}1& 0 \\ e^{k\phi(z)} & 1 \end{pmatrix}\right]_{11} & z \in \Omega_{1}\cap \mathcal{U}_{1}\\
&\\
\displaystyle \left[R(z)P^{0}(z)\begin{pmatrix}1& 0 \\ -e^{-k\phi(z)} & 1 \end{pmatrix}\right]_{11} & z \in \Omega_{2}\cap \mathcal{U}_{1}\ .
\end{array}
\right.
\end{align}
From the above relation we obtain the expansions in the following regions.\\

\vskip 1pt
\noindent
{\bf The exterior region $\Omega_{\infty}$.}\\
In this region, from \eqref{Rj} and \eqref{T_exp},  for any  integer $M\geq 2$ we have 
\begin{align}
\nonumber
\poly_k(z) &= e^{kg(z)}\left(1-\frac{1}{z}\right)^{\gamma}\left(1+\order{\frac{1}{k^{M+\gamma}}}\right)\\
\label{exp_Omega_infty}
&=z^k\left(1-\frac{1}{z}\right)^{\gamma}\left(1+\order{\frac{1}{k^{M+\gamma}}}\right)\,,
\end{align}
 on any compact subset of $\Omega_{\infty}$. Therefore there are no zeros accumulating in this region.

\vskip 2pt

\noindent
{\bf The interior region $\Omega_{0}\setminus \mathcal{U}_{1}$.}\\
\begin{align}
\label{Omega0}
\poly_k(z) &= e^{kg(z)}\left(\frac{1}{k^{1+\gamma}}\frac{1}{(z-1)\Gamma(-\gamma)}\left(1-\frac{1}{z_0}\right)^{-\gamma-1}+\order{\frac{1}{k^{2+\gamma}}}\right)\\
\end{align}
The leading term of the above expansion is of order  $\order{k^{-1-\gamma}}$, so there are no zeros in this region.

\vskip 2pt
\noindent
{\bf The interesting region $\Omega_{1}\setminus \mathcal{U}_{1}$.}
\begin{align}
\label{Omega1}
\poly_k(z) &= e^{kg(z)}\left(1-\frac{1}{z}\right)^{\gamma}\left[e^{k\phi(z)}-\frac{1}{k^{1+\gamma}}\frac{\left(1-\frac{1}{z}\right)^{-\gamma}}{(z-1)\Gamma(-\gamma)}\left(1-\frac{1}{z_0}\right)^{-\gamma-1}+\order{\frac{1}{k^{2+\gamma}}}\right]\,,
\end{align}
where $e^{k\phi(z)}$ is uniformly bounded on $\Omega_1 \subset \{\Re(\phi) \leq 0\}$.\\
%
\noindent
\vskip 2pt
\noindent
{\bf The other interesting region $z \in \Omega_{2}\setminus \mathcal{U}_{1}$:}\\
\begin{align}
\label{Omega2}
\poly_k(z) &= e^{kg(z)}\left(1-\frac{1}{z}\right)^{\gamma}\left[1-\dfrac{e^{-k\phi(z)}}{k^{\gamma+1}}\left[
\frac{\left(1-\frac{1}{z}\right)^{-\gamma}}{(z-1)\Gamma(-\gamma)}\left(1-\frac{1}{z_0}\right)^{-\gamma-1}+\order{\frac{1}{k}}\right]\right]\,,
\end{align}
 where $e^{-k\phi(z)}$ is uniformly bounded on $\Omega_2 \subset \{\Re(\phi) \geq 0\}$.\\
 
 \vskip 2pt
\noindent
{\bf The region $ \mathcal{U}_{1}$}: for $z\in  \mathcal{U}_{1}\cap \exte(\mathcal{C})$
\begin{align}
\label{pol_pre_U1a}
\poly_k(z)& = e^{kg(z)}\dfrac{(z-1)^\gamma}{(zw(z))^{\gamma}}\left[w(z)^{\gamma}-\dfrac{e^{-k\phi(z)}}{k^{\gamma}}
\left( \tilde{\Psi}_{12}(kw(z))-
\frac{ w(z)^\gamma}{k}\left(\frac{\left(1-\frac{1}{z_0}\right)^{-\gamma-1}}{(z-1)\Gamma(-\gamma)}+v_R^{(1)}(z)+\order{\frac{1}{k}}\right)\right)\right]
\end{align}
for $z\in  \mathcal{U}_{1}\cap \inte(\mathcal{C})$
\begin{align}
\label{pol_pre_U1b}
\poly_k(z)& = e^{kg(z)}\dfrac{(z-1)^\gamma}{(zw(z))^{\gamma}}\left[e^{k\phi(z)}w(z)^\gamma-\dfrac{\tilde{\Psi}_{12}(kw(z))}{k^{\gamma}}-
\frac{w(z)^{\gamma}}{k^{1+\gamma}}\left(\frac{\left(1-\frac{1}{z_0}\right)^{-\gamma-1}}{(z-1)\Gamma(-\gamma)}+v_R^{(1)}(z)+\order{\frac{1}{k}}\right)\right]
\end{align}
where  $\tilde{\Psi}_{12}$ is the $(1,2)$-entry of the matrix $\tilde{\Psi}$  defined in \eqref{tildePsi} and $v_R^{(1)}(z)$ has been defined in \eqref{exp_vR1}.
Here $\inte(\mathcal{C})$ and $\exte(\mathcal{C})$ is the interior and exterior of $\mathcal{C}$ respectively.\\

In order to obtain the asymptotic behaviour of the polynomials $p_n(\lb)$ we use the substitution
\[
p_n(\lb)=(-t)^k\lb^l \pi_k\left(1-\frac{\lb^s}{t}\right),\quad n=ks+l,\;\; \gamma=\frac{s-1-l}{s}.
\]
which gives the relations \eqref{pn1_theo2}-\eqref{pn4_theo2} in Theorem~\ref{theorem2}.
\hfill{\bf Q.E.D.}

\begin{proposition}
\label{propo_zeros_pre}
The support of the counting measure of the zeros of the polynomials $\pi_k(z)$  outside an arbitrary small disk  ${\cal U}_1$ surrounding the point $z=1$  tends uniformly to the curve $\mathcal{C}$ defined in \eqref{Gamma0}.
The zeros are within a distance $o(1/k)$ from the curve defined by
\be
\label{phi_deform}
\Re\phi(z)=-(1+\gamma)\frac{\log(k)}{k}+\dfrac{1}{k}\log\left(\dfrac{1}{|\Gamma(-\gamma)|}\dfrac{|z|^\gamma}{|z-1|^{\gamma+1}}\left|1-\frac{1}{z_0}\right|^{-\gamma-1}\right)
\ee
where the function $\phi(z)$ has been defined in \eqref{def_phi}.
The curves in \eqref{phi_deform} approach $\mathcal{C}$ at the rate  $\order{\log k /k}$ and lies in $\inte(\mathcal{C})$. The normalised counting measure of the zeros of $\pi_k(z)$ converges  to the probability measure $\nu$ defined in \eqref{dnu}.
\end{proposition}
\begin{proof}
Observing the asymptotic expansion \eqref{exp_Omega_infty} of $\pi_k(z)$ in $\Omega_\infty\setminus {\cal U}_1$ it is clear that $\pi_k(z)$ does not have any zeros in that region, since $z=0$ and $z=1$ do not belong to  $\Omega_\infty\setminus {\cal U}_1$.  The same reasoning applies to the region $\Omega_0\setminus {\cal U}_1$ where there are no zeros of $\pi_k(z)$ for $k$ sufficiently large. 

From the relations \eqref{Omega1} and \eqref{Omega2} one has that in $\Omega_1\cup\Omega_2$  using the explicit expression of $g(z)$ defined in \eqref{gint}
\begin{align}
\label{Omega22}
\poly_k(z) &= z^{k}\left(1-\frac{1}{z}\right)^{\gamma}\left[1-\dfrac{e^{-k\phi(z)}}{k^{\gamma+1}}\left[
\frac{\left(1-\frac{1}{z}\right)^{-\gamma}}{(z-1)\Gamma(-\gamma)}\left(1-\frac{1}{z_0}\right)^{-\gamma-1}+\order{\frac{1}{k}}\right]\right]\,.
\end{align}
The zeros of $\pi_k(z)$ may only lie asymptotically where the expression
\[
1-\dfrac{e^{-k\phi(z)}}{k^{\gamma+1}}
\frac{\left(1-\frac{1}{z}\right)^{-\gamma}}{(z-1)\Gamma(-\gamma)}\left(1-\frac{1}{z_0}\right)^{-\gamma-1},\quad z\in \Omega_1\cup\Omega_2,
\]
is equal to zero.
Since $\Omega_2 \subset \{\Re(\phi) \geq 0\}$ and $\Omega_1 \subset \{\Re(\phi) \leq 0\}$, it follows that the zeros of $\pi_k(z)$ may lie only in the region $\Omega_1$  and such that 
  $\Re\phi(z)=\order{\log k / k}$. Namely the zeros of the polynomials $\pi_k(z)$ lie on the   curve given by \eqref{phi_deform}  with an error of order  $\order{1/k^2}$. Such curves  converge to the curve $\mathcal{C}$ defined \eqref{Gamma0} at a rate 
$\order{\log k /k}$ (see Fig.~\ref{fzero_pre}).
\begin{figure}[H]
\centering
\includegraphics[scale=0.25]{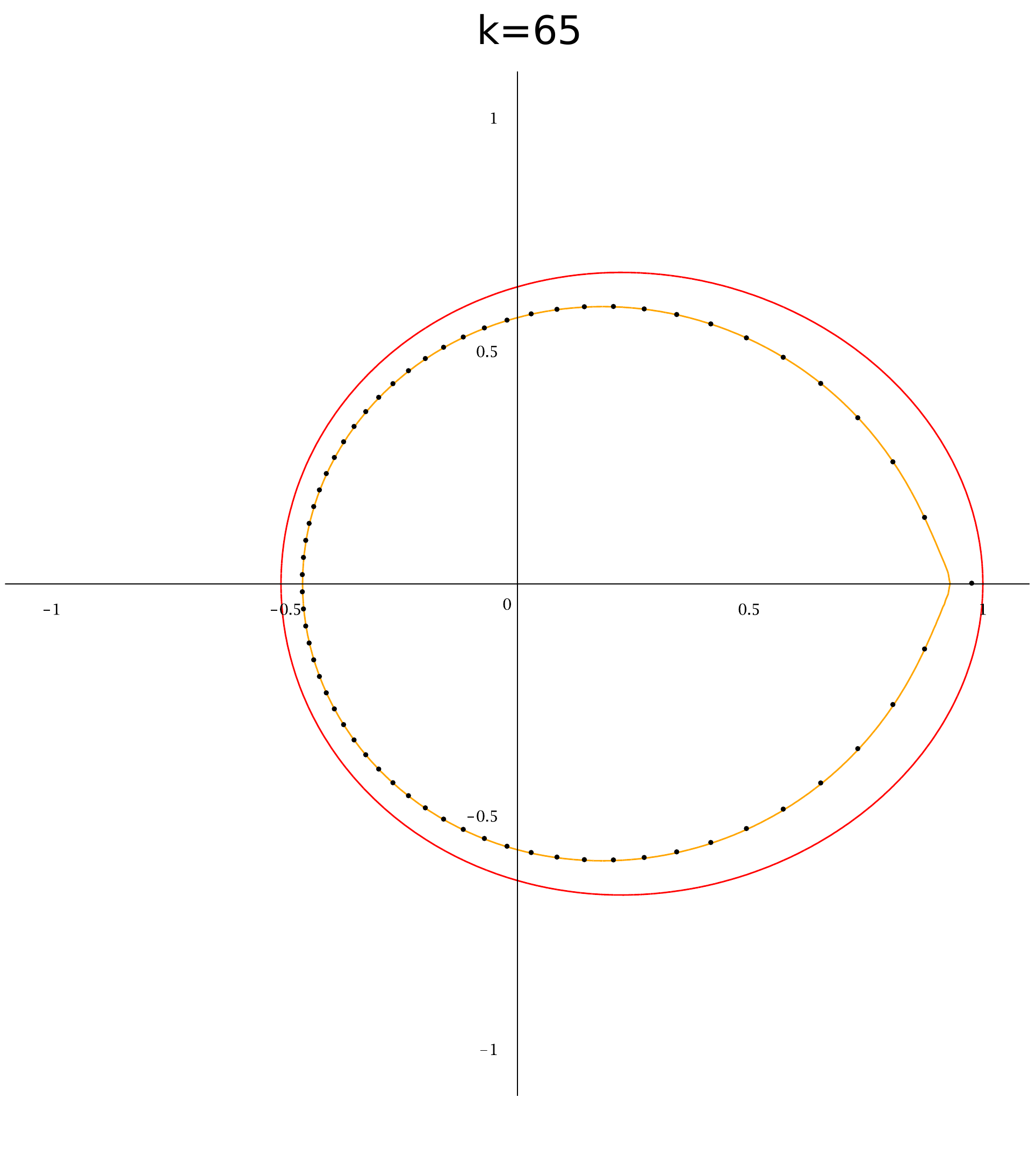}\includegraphics[scale=0.25]{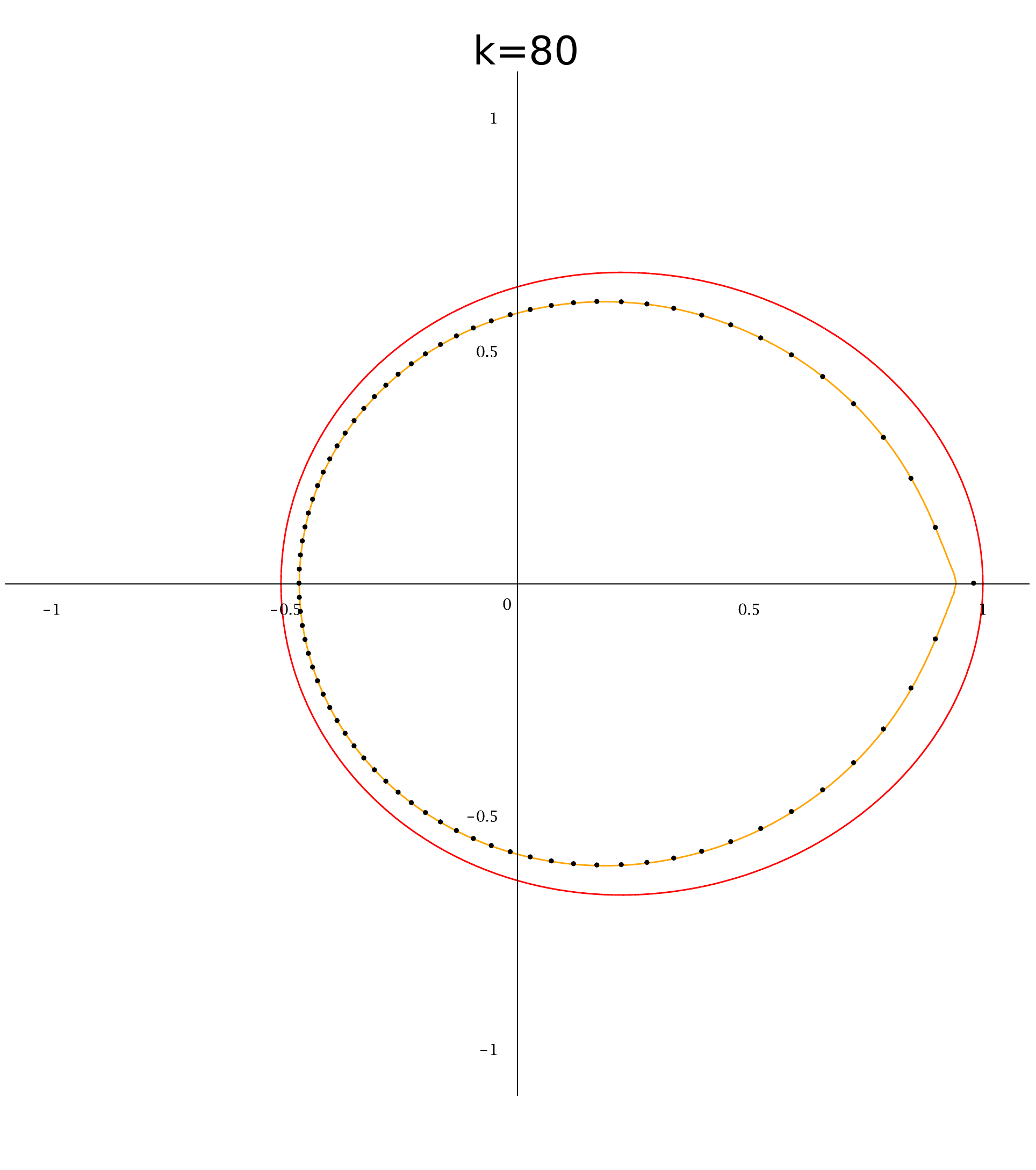}\includegraphics[scale=0.25]{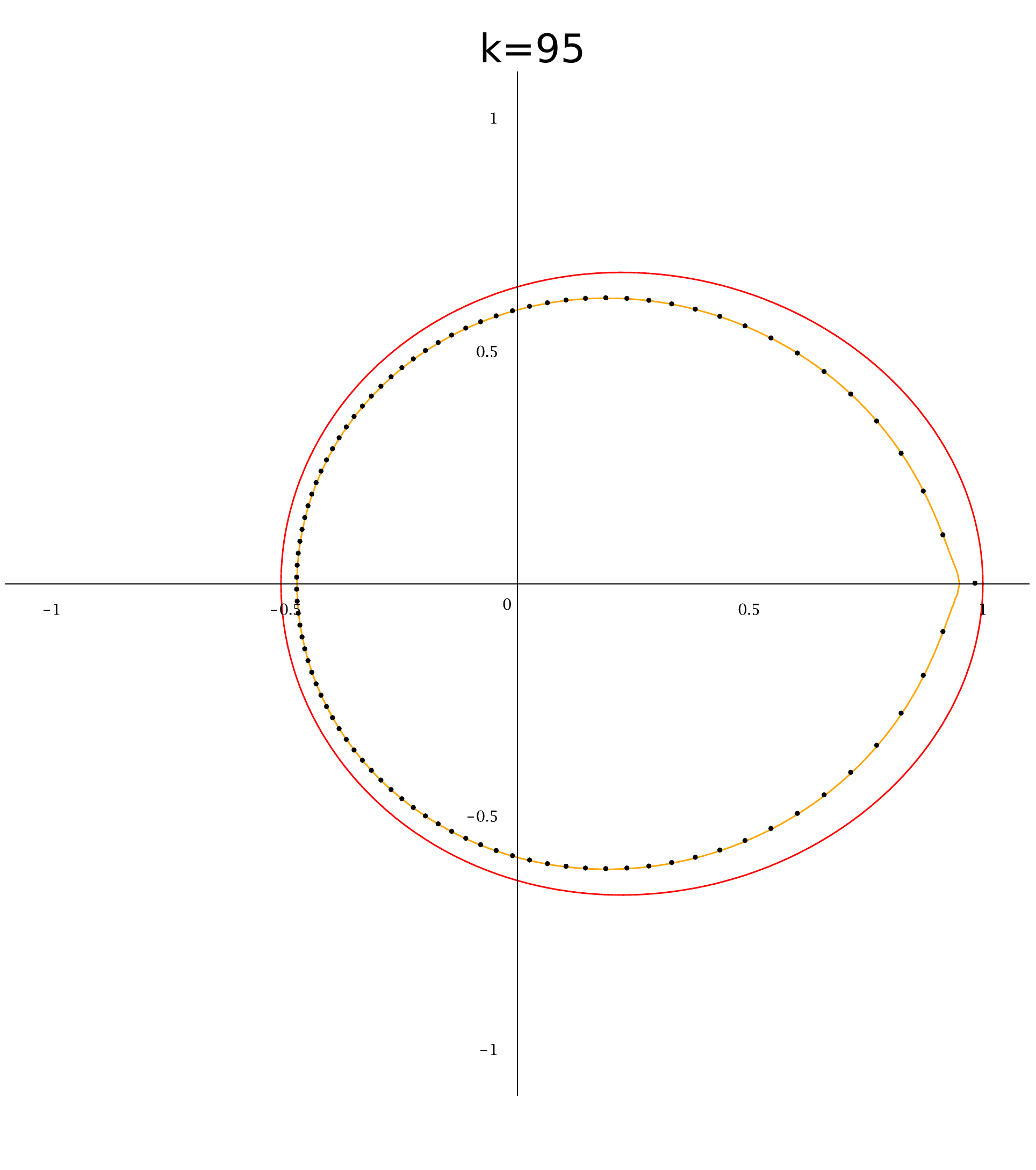}\\
\caption{The zeros of $\pi_k(z)$ for $s=3$, $l=0$,  $t<t_c$, and $k=65, 80, 95$. The  red contour is $\mathcal{C}$ while the   yellow contour is the curve  \protect\eqref{phi_deform}.}
\label{fzero_pre}
\end{figure}
We now know that the zeros of $\pi_k(z)$ accumulate on the curve $\mathcal{C}$. We still need to determine the asymptotic zero  distribution.
  Due to the strong asymptotic of $\pi_k(z)$ we have from \eqref{exp_Omega_infty}
\[
\lim_{k\to\infty}\frac{1}{k}\log(\pi_k(z))=\int_{\mathcal{C}}\log(z-\xi)d\nu(\xi),
\]
uniformly on compact subsets of the exterior of $\mathcal{C}$.  Furthermore $\mathcal{C}$ is the boundary of its polynomial convex hull. Then it follows  that the measure  $d\nu$ in \eqref{dnu} is the weak-star limit  of the zeros distribution  of  the polynomials $\pi_k(z)$ (see \cite{MhaskarSaff}, Theorem 2.3 and \cite{SaffTotik} Chapter 3).

\end{proof}

\section{Post-critical case}
In this section we assume that
\be
0<z_0< 1\,.
\ee
To perform the analysis similar to the pre-critical case, one has to choose the appropriate contour from the  homotopy class of the contour $\Sigma$ that encircles the branch cut $[0,1]$.
 Since the point $z_0$ is lying in $(0,1)$, the  family of curves $\mathcal{C}_r$ ,  $0<r\leq z_0$ defined in \eqref{def_Gammar} are loops encircling $z=0$ and crossing the real line in $z=r$.
 The regions where $\Re \phi_r(z)<0$  with $\phi_r(z)$  defined in \eqref{def_phir}, are depicted in the figures below in red for two values of the parameter $r$.
\begin{figure}[H]
\centering
\includegraphics{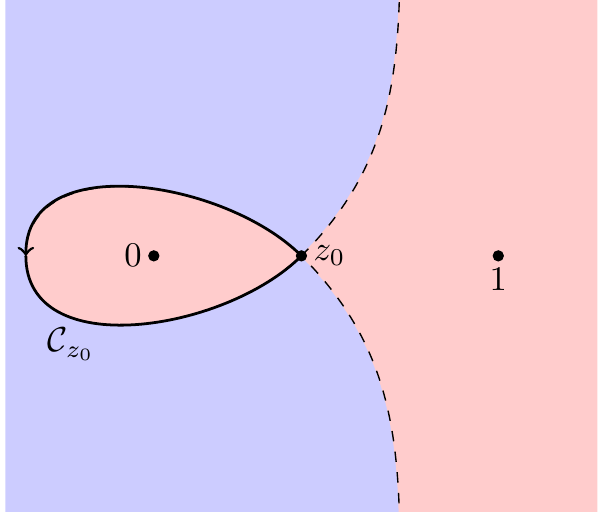}\includegraphics{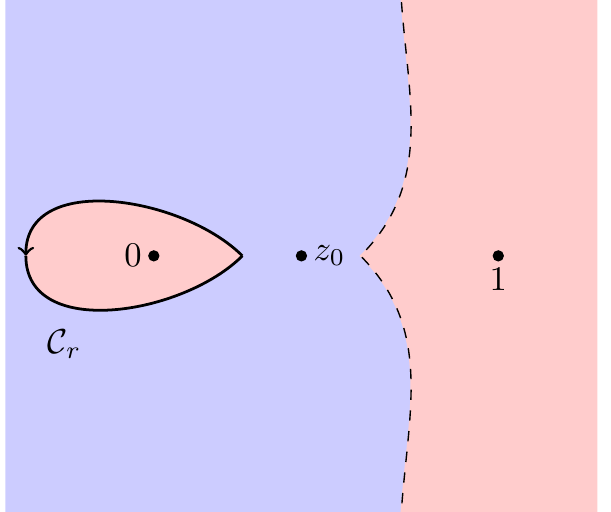}\\
\caption{The contour $\mathcal{C}_r$ for $r=z_0$ on the left and for $0<r<z_0$ on the right. The region where  $\Re \phi_r(z)<0$, with $\phi_r(z)$  defined in \protect\eqref{def_phir},  is red.}
\label{tikz_gamma_R_post1}
\end{figure}

\begin{wrapfigure}[15]{rhtb!}{0.35\textwidth}
\includegraphics[width=0.35\textwidth]{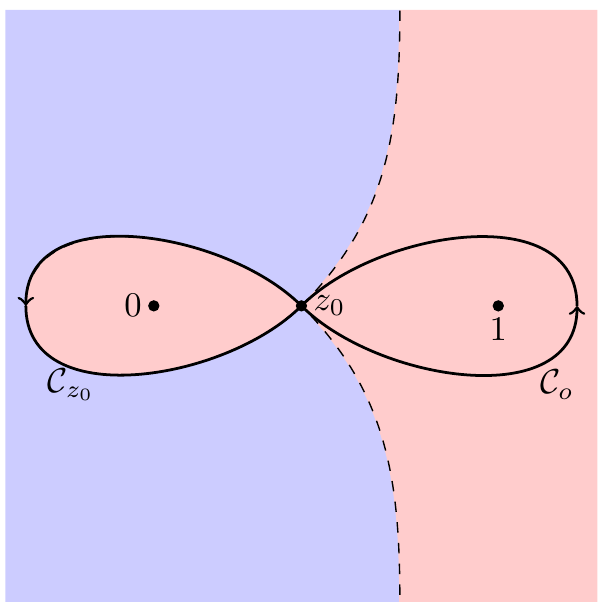}
\caption{The contour $\mathcal{C}_{z_0}\cup \mathcal{C}_o$. The region $\Re \phi_{z_0}(z)<0$  is plotted in red.}
\label{tikz_gammapost3}
\end{wrapfigure}
In order to perform the asymptotic analysis of the Riemann-Hilbert problem \eqref{RH_Y0},  \eqref{RH_Y11} and  \eqref{RH_Y2}, we need to deform  $\Sigma $ to a homotopic  contour $\mathcal{C}_r\cup \mathcal{C}_o$ in such a way that the $\Re\phi_r(z)$ is  negative on $\mathcal{C}_o$.

 It is clear from Figure~\ref{tikz_gamma_R_post1} that   the only possibility is to deform the contour $\Sigma$ to the contour $\mathcal{C}_{z_0}\cup \mathcal{C}_o$ with $\mathcal{C}_o$ as depicted in 
 Fig.~\ref{tikz_gammapost3}.
For simplifying the notation we define the following.
  \begin{definition} Let
\be
\label{def_phi_post}
\phi(z) :=\phi_{z_0}(z)=\log\left(\frac{z}{z_0}\right)-\frac{z}{z_0}+ 1\qquad z \in \C\setminus(0,+\infty)\ ,
\ee
with
\be
\phi_{+}(z) - \phi_{-}(z) = -2\pi i \quad z \in (0,+\infty)\ .
\ee
and 
\be
\mathcal{C}:=\mathcal{C}_{z_0}
\ee
with $\phi_{z_0}(z)$ defined in  \eqref{def_phir} and $\mathcal{C}_{z_0}$ defined in \eqref{def_Gammar}. 
\end{definition}

 According to \eqref{def_phir1} on the contour $\mathcal{C}$ we have
\begin{align}
g_{+}(z)+g_{-}(z)&=V(z)+\ell\\
g_{+}(z)-g_{-}(z)&=-\phi(z)\ ,
\end{align}
where
\[
\ell = \log(z_0) - 1\ .
\]

\subsection{First transformation $\tilde Y \mapsto S$}
Since $\mathcal{C}\cup\mathcal{C}_{o}$ is homotopically equivalent to $\Sigma$ in $\C\setminus [0,1]$ when $z_0<1$, we can deform the contour $\Sigma$ appearing in the Riemann--Hilbert problem \eqref{RH_Y0}-\eqref{RH_Y2}  for the matrix $\tilde Y(z)$ 
to $\mathcal{C}\cup\mathcal{C}_{o}$. Define the modified matrix
\be
S(z) = e^{-k(\ell/2)\sigma_3}\tilde Y(z)e^{-kg(z)\sigma_3}e^{k(\ell/2)\sigma_3} \qquad z \in \C\setminus (\mathcal{C}\cup\mathcal{C}_o\cup [0,1]),
\ee
where $\tilde{Y}(z)$ has been defined in \eqref{Ytilde}.
Then the matrix $S(z)$  is the unique solution of the following Riemann-Hilbert problem with standard large $z$ behaviour at $z=\infty$:\\
\begin{wrapfigure}[5]{rhtb!}{0.5\textwidth}
\includegraphics[width=0.5\textwidth]{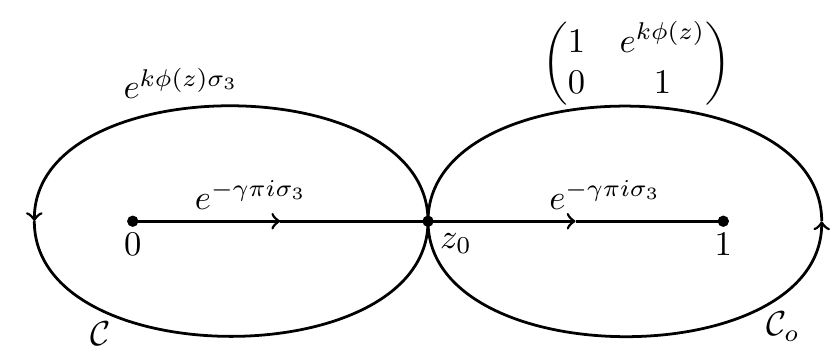}
\caption{The jump matrices for $S(z)$. The function $\phi(z)$ is defined in \protect\eqref{def_phi_post}.}
\label{fig:jump_U_post}
\end{wrapfigure}
\begin{itemize}
\item[1.] $S(z)$ is analytic for  $z\in \C\setminus(\mathcal{C}\cup\mathcal{C}_o\cup [0,1]),$.
\item[2.] Jump discontinuities (with $\phi(z)$ as in \eqref{def_phi_post}):
\be
S_{+}(z) = S_{-}(z)
\left\{
\begin{array}{cl}
e^{k\phi(z)\sigma_3} & z \in \mathcal{C}\\
\begin{pmatrix}1&  e^{k\phi(z)} \\ 0 & 1 \end{pmatrix} & z \in \mathcal{C}_{o}\\
e^{-\gamma\pi i\sigma_3} & z\in (0,1)
\end{array}
\right.
\ee

\item[3.] Endpoint behaviour at $z=0$ and $z=1$:
\be
\begin{split}
S(z)z^{-\frac{\gamma}{2}\sigma_3} &= \order{1} \quad z \to 0\\
S(z)(z-1)^{\frac{\gamma}{2}\sigma_3} &= \order{1} \quad z \to 1\ .
\end{split}
\ee
\end{itemize}

\begin{itemize}
\item[4.] Large $z$ boundary behaviour:
\be
S(z) = I+{\mathcal O}\left(\frac{1}{z}\right)\ ,\quad z \to \infty\ .
\ee
\end{itemize}

The orthogonal polynomials $\pi_k(z)$ are recovered from the matrix $S(z)$ using the relation
\be
\pi_k(z)=e^{kg(z)}\left(1-\dfrac{1}{z}\right)^{\frac{\gamma}{2}}S_{11}(z).
\ee
\subsection{The second transformation $S \mapsto T$: opening of the lenses}
Consider two extra loops $\mathcal{C}_{i}$ and $\mathcal{C}_{e}$ as shown in Figure \ref{fig:jump_T_postT}.
These define new domains $\Omega_{0}$, $\Omega_{1}$, $\Omega_{2}$ and $\Omega_{\infty}$.
Define
\be
\label{Tk1}
T(z)
=\left\{
\begin{array}{ll}
S(z) & z \in \Omega_{\infty}\cup \Omega_0\cup{\Omega_{3}}\\
S(z)\begin{pmatrix}1&  0 \\ -e^{k\phi(z)} & 1 \end{pmatrix} & z \in \Omega_{1}\\
S(z)\begin{pmatrix}1&  0 \\ e^{-k\phi(z)} & 1 \end{pmatrix} & z \in \Omega_{2}\ .
\end{array}
\right.
\ee
Then this matrix-valued function has the following jump discontinuities:
\be
\begin{split}
\label{Tk2}
T_{+}(z)
&=T_{-}(z)v_T(z),\quad z\in\Sigma_T\\
v_T(z)&=\left\{
\begin{array}{cl}
\begin{pmatrix}0&1\\-1&0\end{pmatrix}  & z \in \mathcal{C}\\
\begin{pmatrix}1& 0 \\ e^{-k\phi(z)} & 1 \end{pmatrix}  & z \in \mathcal{C}_{e}\\
\begin{pmatrix}1& 0 \\ e^{k\phi(z)} & 1 \end{pmatrix}  & z \in \mathcal{C}_{i}\\
\begin{pmatrix}1 & e^{k\phi(z)}\\0 & 1 \end{pmatrix}  & z \in \mathcal{C}_{o}\\
e^{-\gamma\pi i\sigma_3}  & z \in (0,1)\,,
\end{array}
\right.
\end{split}
\ee
where $\Sigma_T$ is the contour defined in the Figure~\ref{fig:jump_T_postT}.
\tfigure{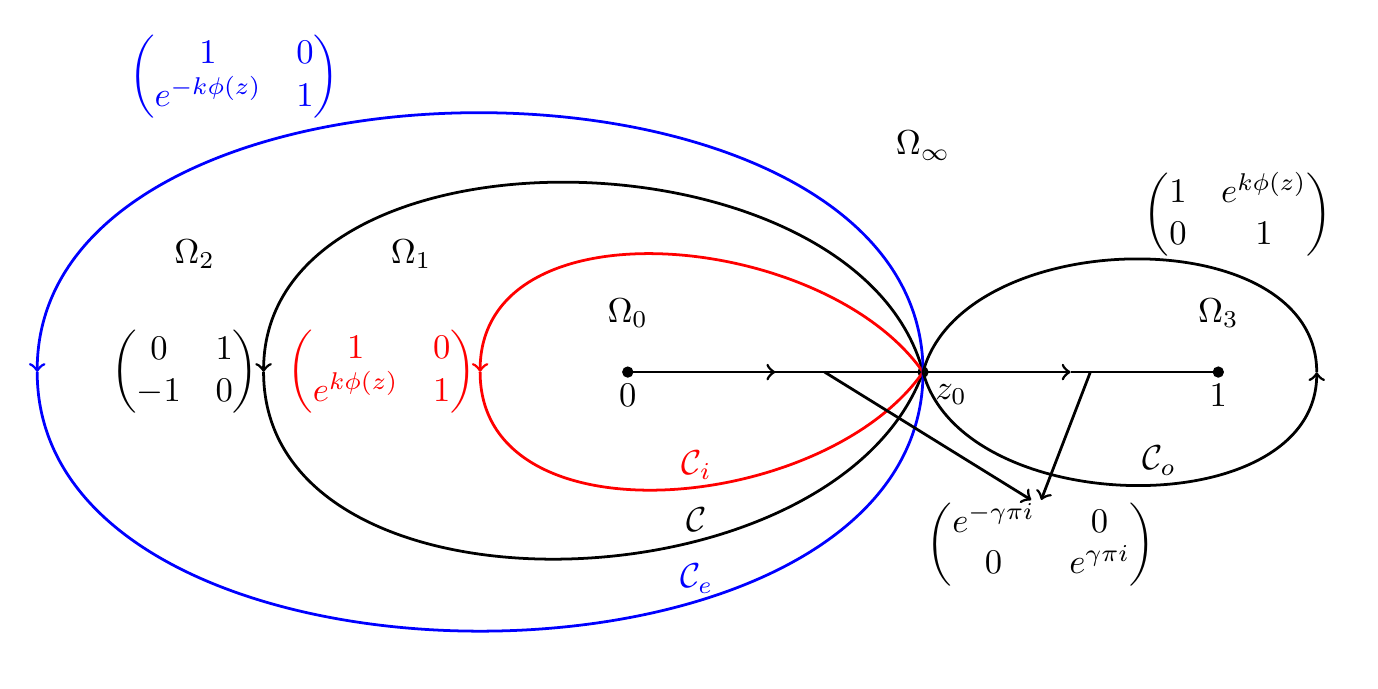}{The jump matrices for $T(z)$ and the contour $\Sigma_T=\mathcal{C}\cup\mathcal{C}_i\cup\mathcal{C}_e\cup\mathcal{C}_o\cup(0,1)$.}{fig:jump_T_postT}{scale=1}

Endpoint behaviour:
\be
T(z)z^{-\frac{\gamma}{2}\sigma_3} = \order{1} \quad z \to 0\ , \qquad T(z)(z-1)^{\frac{\gamma}{2}\sigma_3} = \order{1} \quad z \to 1\ .
\ee
Large $z$ boundary behaviour:
\be
T(z) = I+{\mathcal O}\left(\frac{1}{z}\right)\ ,\quad z \to \infty\ .
\ee

\begin{proposition}There exists a constant $c_0>0$ so that 
\[
v_T(z)=I+\order{e^{-c_0k}}\quad \mbox{as $k\to\infty$}
\]
uniformly for $z\in\mathcal{C}_o\cup\mathcal{C}_i\cup\mathcal{C}_e\backslash {\cal U}_{z_0}$, where ${\cal U}_{z_0}$ is a small neighbourhood of $z_0$.
\end{proposition}
The proof of the proposition follows in a straightforward way by observing that by contruction (see Figure~\ref{tikz_gammapost3}) $\Re\phi(z)>0$ for $z\in\mathcal{C}_e\backslash {\cal U}_{z_0}$ and $\Re\phi(z)<0$ for $z\in\mathcal{C}_i\cup\mathcal{C}_o\backslash {\cal U}_{z_0}$.
Therefore we have 
\[
v_T(z)\to v^{\infty}(z)\quad \mbox{as $k\to\infty$}
\]
exponentially fast,  where 
\begin{equation}
\label{vinfinity}
v^{\infty}(z)=\left\{
\begin{array}{cl}
\begin{pmatrix}0&1\\-1&0\end{pmatrix}  &\mbox{ as $z\in\mathcal{C}_k$},\\
&\\
e^{-\pi i \gamma\sigma_3} &\mbox{ as $z\in(0,1)$},\\
&\\
I&\mbox{ as $z\in\mathcal{C}_i\cup\mathcal{C}_e\cup\mathcal{C}_o$}.
\end{array}\right.
\end{equation}

The polynomials $\poly_{k}(z)$ can be expressed in terms of $T(z)$ in the following way:
\begin{equation}
\begin{split}
\label{pik_T}
\poly_k(z) &=e^{kg(z)}\left(1-\frac{1}{z}\right)^{\frac{\gamma}{2}}S_{11}(z)=\\
&=e^{kg(z)}\left(1-\frac{1}{z}\right)^{\frac{\gamma}{2}}\left\{
\begin{array}{ll}
T_{11}(z) & z \in \Omega_{\infty}\cup \Omega_{0}\cup \Omega_{3}\\
\displaystyle \left[T(z)\begin{pmatrix}1& 0 \\ e^{k\phi(z)} & 1 \end{pmatrix}\right]_{11} & z \in \Omega_{1}\\
\displaystyle \left[T(z)\begin{pmatrix}1& 0 \\ -e^{-k\phi(z)} & 1 \end{pmatrix}\right]_{11} & z \in \Omega_{2}
\end{array}
\right.\\
&=e^{kg(z)}\left(1-\frac{1}{z}\right)^{\frac{\gamma}{2}}\left\{
\begin{array}{ll}
T_{11}(z) & z \in \Omega_{\infty}\cup \Omega_{0}\cup \Omega_{3}\\
T_{11}(z) + e^{k\phi(z)}T_{12}(z) & z \in \Omega_{1}\\
T_{11}(z) - e^{-k\phi(z)}T_{12}(z) & z \in \Omega_{2}\,.
\end{array}
\right.
\end{split}
\end{equation}

\subsection{The outer parametrix for large $z$}

Ignoring the exponentially small jumps and a small neighbourhood
$\mathcal U_{z_0}$ of $z_0$ where the uniform exponential decay does not
remain valid, we are led to the following RH problem for $P^{(\infty)}$:
\begin{enumerate}
\item $P^{(\infty)}$ is holomorphic in $\C\setminus (\mathcal{C} \cup [0,1])$,

\item$P^{(\infty)}$
satisfies the following jump conditions  on $\mathcal{C}$ and $(0,1)$:
\be
\label{RHP Pinfty b}
P_+^{(\infty)}(z)= P_-^{(\infty)}(z)
\left\{
\begin{array}{ll}
\begin{pmatrix}0&1\\-1&0\end{pmatrix} & z\in\mathcal{C}\\
&\\
e^{-\pi i \gamma\sigma_{3}} &z\in(0,1)
\end{array}
\right. .
\ee
\item$P^{(\infty)}(z)$ has the following behavior as $z\to\infty$,
\begin{equation}
P^{\infty}(z) = I+{\mathcal O}\left(\frac{1}{z}\right)\ ,\quad z \to \infty\ 
\label{RHP Pinfty c}.
\end{equation}
\end{enumerate}

The above Riemann-Hilbert problem can be solved explicitly in terms of the piecewise defined matrix function
\be
\chi(z):= \left\{\begin{array}{cl}
\begin{pmatrix}0& 1 \\ -1 & 0 \end{pmatrix}& z \in \inte(\mathcal{C})\\
I& z \in \exte(\mathcal{C})\ .
\end{array}\right.
\ee
Define
\be
\tilde P^{\infty}(z) := P^{\infty}(z)\chi^{-1}(z)\ .
\ee
The  matrix $\tilde P^{\infty}(z)$ has no jump on $\mathcal{C}$ and it satisfies the following RHP: 
\begin{enumerate}
\item $\tilde P^{\infty}$ is holomorphic in $\C\setminus [0,1]$,
\item Jump across $(0,1)$:
\be
\tilde P_{+}^{\infty}(z)=\tilde P_{-}^{\infty}(z)\left\{\begin{array}{lc}
e^{\gamma\pi i\sigma_3} & z \in (0,z_0)\\
&\\
e^{-\gamma\pi i\sigma_3} & z \in (z_0,1)
\end{array}\right.
\ee
\item Large $z$ boundary behaviour:
\be
\tilde P^{\infty}(z) = I+{\mathcal O}\left(\frac{1}{z}\right)\ ,\quad z \to \infty\ .
\ee
\end{enumerate}
A particular solution is given by
\be
\tilde P^{\infty}(z) = \left[\frac{(z-z_0)^2}{z(z-1)}\right]^{\frac{\gamma}{2}\sigma_3}
\ee
which leads to the particular solution 
\be
\label{eq:global_parametrix_post}
P^{\infty}(z) = \left[\frac{(z-z_0)^2}{z(z-1)}\right]^{\frac{\gamma}{2}\sigma_3}\chi(z)= 
\left\{\begin{array}{ll}
\displaystyle \left[\frac{(z-z_0)^2}{z(z-1)}\right]^{\frac{\gamma}{2}\sigma_3}\begin{pmatrix}0& 1 \\ -1 & 0 \end{pmatrix}& z \in \inte(\mathcal{C})\\
\displaystyle \left[\frac{(z-z_0)^2}{z(z-1)}\right]^{\frac{\gamma}{2}\sigma_3} & z\in \exte(\mathcal{C})\ .
\end{array}\right.
\ee

\subsection{The local parametrix at $z=z_0$}
The aim of this section is to construct a local parametrix $P^0(z)$  in a small  neighbourhood
${\mathcal U}_{z_0}$ of $z_0$   having  the same jump properties as $T(z)$  for $z$ near $z_0$  and matching  the outer parametric $P^{\infty}(z)$ 
in the limit $k\to \infty$ for $z\in\partial {\cal U}_{z_0}$.
\subsubsection{RH problem for $P^{0}(z)$} 
\begin{enumerate}
\item $P^{0}(z)$ is analytic for  $z\in \overline{ {\cal U}}_{z_0}\setminus \Sigma_T$,
\item $P_+^{0}(z)=P_-^{0}(z)v_T(z)$ for $z\in {\cal U}_{z_0}\cap\Sigma_T$,
\item for $z\in\partial {\cal U}_{z_0}$, we have
\be
\label{match}
P^0(z)=P^{\infty}(z)(I+o(1))\qquad \mbox{as $k\to\infty$ and $z\in \partial {\mathcal U}_{z_0}$. }
\ee

\end{enumerate}
In order to build such  local parametrix near the point $z=z_0$ we first construct  a new matrix function $B(z)$ from   $P^0(z)$. Let us first  define
\be
\label{def_Delta}
\Delta(z)
=\left\{
\begin{array}{ll}
I  & \Im(z)<0\\
e^{-\gamma\pi i\sigma_{3}} & \Im(z)>0,\ 
\end{array}
\right.
\ee
and  the matrix ${\cal Q}$ as follows
\be
\label{def_Q1}
{\cal Q}(z)=
\left\{
\begin{array}{cl}
\begin{pmatrix}
1&-e^{k\phi(z)}\\
0&1\end{pmatrix}&\mbox{ for $z\in\Omega_3\cap  {\mathcal U}_{z_0}$}\\
\begin{pmatrix}
1&0\\
-e^{k\phi(z)}&1\end{pmatrix}\begin{pmatrix}0&-1\\1&0\end{pmatrix}&\mbox{ for $z\in\Omega_0\cap {\mathcal U}_{z_0}$}\\
\begin{pmatrix}
0&-1\\
1&0\end{pmatrix}&\mbox{ for $z\in(\Omega_1\backslash\Omega_0)\cap  {\mathcal U}_{z_0}$}\\
I&\mbox{ elsewhere}.
\end{array}\right.
\ee
Then the matrix $B(z)$ is  defined  from $P^0(z)$  by the relation
\begin{equation}
\label{Uk}
B(z)=P^0(z){\cal Q}(z)e^{k\phi(z)\sigma_3/2}\Delta(z)^{-1}.
\end{equation}
The matrix $B(z)$ satisfies the jump relations specified in  Figure \ref{fig:Tamaramodel} in a neighbourhood of $z_0$.
In the next section we construct the solution of the so called model problem, namely a $2\times2$  matrix $\Psi$ that has the same jumps as the matrix $B$.

\subsubsection{Model problem}
\begin{wrapfigure}[10]{rhtb!}{0.4\textwidth}
\includegraphics[width=0.4\textwidth]{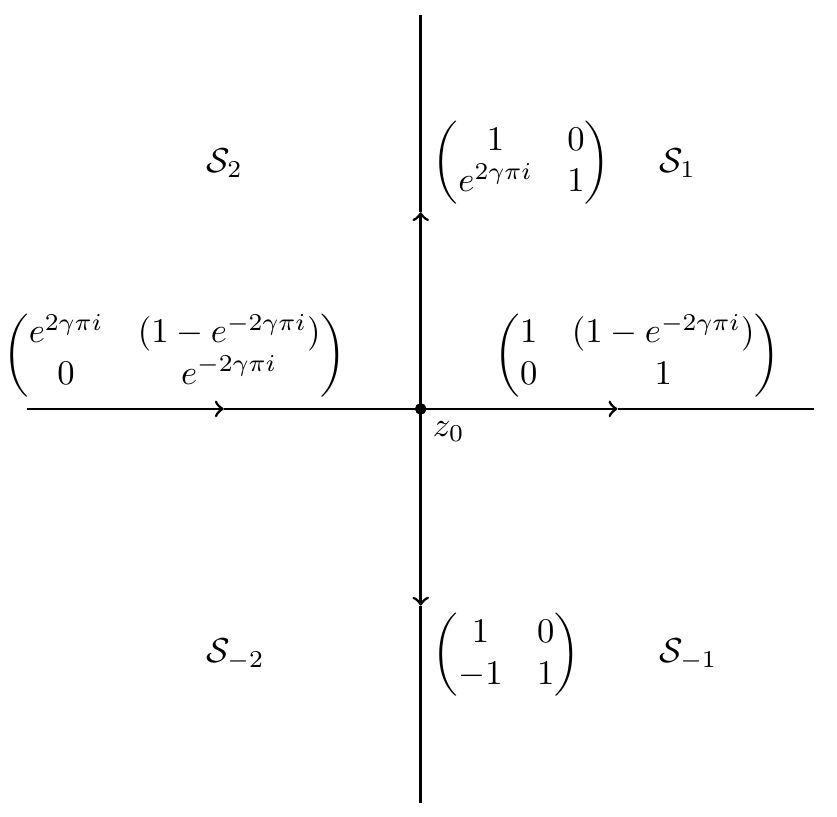}
\caption{The jumps for the matrix $B$}
\label{fig:Tamaramodel}
\end{wrapfigure}
Consider the model problem for the $2\times2$ matrix function $\Psi(\xi)$  analytic in $\mathbb{C}\backslash \{\mathbb{R}\cup i\mathbb{R}\}$ with jumps and boundary behaviour
\begin{align}
\label{model1}
\Psi_{+}(\xi)&=\Psi_{-}(\xi)v_{\Psi}(\xi)\\
\Psi(\xi)&=\left(I+\frac{\Psi_1}{\xi}+\frac{\Psi_2}{\xi^2}+\frac{\Psi_3}{\xi^3}+\order{\frac{1}{\xi^{4}}}\right)\left[e^{-\frac{\xi^2}{2}}\xi^{\gamma}\right]^{\sigma_{3}}\\
&\qquad \text{as }\xi\to\infty\ ,
\label{model2}
\end{align}
with $\Psi_1$  $\Psi_2$ and $\Psi_3$ constant matrices (independent from $\xi$) and 
where the matrix $v_{\Psi}(\xi)$ is defined as 
\be
v_{\Psi}(\xi)=
\left\{
\begin{array}{cl}
\begin{pmatrix}1& 1-e^{-2\gamma\pi i}\\ 0 & 1 \end{pmatrix}  &\xi\in\mathbb{R}^+\\
\begin{pmatrix}1&0  \\ e^{2\gamma\pi i} & 1 \end{pmatrix}  & \xi \in i\mathbb{R}^+\\
\begin{pmatrix}e^{2\gamma\pi i}& 1-e^{-2\gamma\pi i}\\0 & e^{-2\gamma\pi i} \end{pmatrix}  & \xi\in\mathbb{R}^-\\
\begin{pmatrix}1& 0\\-1 & 1 \end{pmatrix}  & \xi\in -i\mathbb{R}^+ \\
\end{array}
\right. .
\ee
%
The solution of the Riemann--Hilbert problem \eqref{model1} and \eqref{model2} is obtained in the following way, \cite{DeiftZhou,JM,M}.
Recall the parabolic cylinder equation
\begin{equation}
\label{parabolic}
\dfrac{d^2}{d\xi^2}f-\left(\dfrac{1}{4}\xi^2+a\right)f=0,
\end{equation}
(see, e.g., Chapter 19, \cite{AbramowitzStegun}).
This equation has a non-trivial solution $\U(a,\xi)$ associated to any $a\in\mathbb{C}$ specified by the asymptotic behaviour
\be
\label{asympU}
\U(a,\xi) = \xi^{-a-\frac{1}{2}}e^{-\frac{\xi^2}{4}}\left(1-{\frac{\frac{3}{4}+a^2+2a}{2\xi^2}}+\order{\frac{1}{\xi^4}}\right), \quad \xi \to \infty, |\arg(\xi)|<\frac{\pi}{2},
\ee
which is an entire analytic function of $\xi$. Three other solutions can be obtained by symmetry:
\begin{equation}
\label{parabolicsol}
\U(a,-\xi),\quad \U(-a,i\xi),\quad \U(-a,-i\xi).
\end{equation}
The relations among the above four solutions are 
\begin{align}
\label{rel1}
\U(-a,\pm i\xi)&=\dfrac{\Gamma(\frac{1}{2}+a) }{\sqrt{2\pi}}\left(e^{-i\pi(a-\frac{1}{2})/2}\U(a,\pm\xi)+
e^{i\pi(a-\frac{1}{2})/2}\U(a,\mp\xi)\right)\\
\label{rel2}
\U(a,\pm \xi)&=\dfrac{\Gamma(\frac{1}{2}-a) }{\sqrt{2\pi}}\left(e^{-i\pi(a+\frac{1}{2})/2}\U(-a,\pm i\xi)+
e^{i\pi(a+\frac{1}{2})/2}\U(-a,\mp i\xi)\right)
\end{align}
Moreover, solutions for $a$ and $a+1$ are connected by
\be
\label{rel3}
\dfrac{d}{d \xi}\U(a,\xi)+\dfrac{\xi}{2}\U(a,\xi)+(a+\frac{1}{2})\U(a+1,\xi)=0.
\ee

By using \eqref{rel1}-\eqref{rel3}, the solution of the Riemann--Hilbert problem \eqref{model1} and \eqref{model2} takes the form
\begin{equation}
\label{model_problem}
\Psi(\xi)=\left\{\begin{array}{ll}
\begin{pmatrix}
\U(-\gamma-\frac{1}{2},\sqrt{2}\xi)&  \mp\frac{i \gamma e^{\mp\frac{i\pi\gamma}{2}}}{\sqrt{2}\beta_{21}}\U(\gamma+\frac{1}{2},\mp i\sqrt{2}\xi)\\
&\\
 \frac{\gamma}{\sqrt{2}\beta_{12}}\U(-\gamma+\frac{1}{2},\sqrt{2}\xi)&e^{\mp\frac{i\pi\gamma}{2}}\U(\gamma-\frac{1}{2},\mp i\sqrt{2}\xi)\end{pmatrix}2^{-\frac{\gamma}{2}\sigma_3},& \mbox{for $\xi\in{\cal S}_{\pm 1}$}\\
 &\\
 \begin{pmatrix}
e^{\pm i\pi\gamma}\U(-\gamma-\frac{1}{2},-\sqrt{2}\xi)& \mp\frac{i\gamma e^{\mp\frac{i\pi\gamma}{2}}}{\sqrt{2}\beta_{21}}\U(\gamma+\frac{1}{2},\mp i\sqrt{2}\xi)\\
&\\
 -\frac{\gamma e^{\pm i\pi\gamma}}{\sqrt{2}\beta_{12}}\U(-\gamma+\frac{1}{2},-\sqrt{2}\xi).
&e^{\mp \frac{i\pi\gamma}{2}}\U(\gamma-\frac{1}{2},\mp i\sqrt{2}\xi)\\
\end{pmatrix}2^{-\frac{\gamma}{2}\sigma_3}& \mbox{for $\xi\in{\cal S}_{\pm 2}$ }\\
 \end{array}\right.
\end{equation}
where
\begin{equation}
\label{beta}
\beta_{12}=-e^{-i\pi\gamma}\frac{\sqrt{\pi}\gamma}{\Gamma(1-\gamma)2^{\gamma}}=\frac{\gamma}{2\beta_{21}}.
\end{equation}
Furthermore, from \eqref{asympU} one obtains the extra terms of  the asymptotic expansion of $\Psi(\xi)$ as $\xi\to \infty$ 
\begin{multline}
\label{Psiexpansion}
\Psi(\xi) = \left(I+
\frac{\gamma}{2}\begin{bmatrix}
0 & \frac{1}{\beta_{21}}\\
\frac{1}{\beta_{12}} & 0 
\end{bmatrix}\frac{1}{\xi}
+
\begin{bmatrix}
-\gamma(\gamma-1) & 0\\
0& \gamma(\gamma+1) 
\end{bmatrix}\frac{1}{4\xi^2}\right.\\
\left.+\frac{\gamma}{8}\begin{bmatrix}
0 & \frac{(\gamma+1)(\gamma+2)}{\beta_{21}}\\
-\frac{(\gamma-1)(\gamma-2)}{\beta_{12}} & 0 
\end{bmatrix}\frac{1}{\xi^3}
\right)\xi^{\gamma\sigma_3}e^{-\frac{\xi^{2}}{2}\sigma_3},
\end{multline}

namely
\begin{equation}
\label{Psi123}
\begin{split}
&\Psi_1=\frac{\gamma}{2}\begin{bmatrix}
0 & \frac{1}{\beta_{21}}\\
\frac{1}{\beta_{12}} & 0 
\end{bmatrix},\quad \Psi_2=\dfrac{1}{4}\begin{bmatrix}
-\gamma(\gamma-1) & 0\\
0& \gamma(\gamma+1) 
\end{bmatrix},\\
&\Psi_3=\frac{\gamma}{8}\begin{bmatrix}
0 & \frac{(\gamma+1)(\gamma+2)}{\beta_{21}}\\
-\frac{(\gamma-1)(\gamma-2)}{\beta_{12}} & 0 
\end{bmatrix}.
\end{split}
\end{equation}
\subsubsection{Construction of the local parametrix}
From \eqref{Uk} and the previous section we are now ready to specify the form of the local parametrix at $z=z_0$, by
\be
P^0(z)=E(z)\Psi(\sqrt{k}w(z))\Delta(z)e^{-\frac{k}{2}\phi(z)\sigma_3}{\cal Q}(z)^{-1}
\ee
where $E(z)$ is an  analytic  matrix  in a neighbourhood of ${\cal U}_{z_0}$, the matrix $\Psi$ has been defined in \eqref{model_problem} and $\Delta(z)$ and ${\cal Q}(z)$ have been defined in \eqref{def_Delta} and \eqref{def_Q1} respectively.  The function  $w(z)$ is a conformal mapping from a neighbourhood of $z_0$ to a neighbourhood of $0$ and it is specified by 
\be
w^2(z):=\left\{
\begin{array}{cc}
-\phi(z) - 2\pi i& z \in {\mathcal U}_{z_0}\cap \C_{-}\\
-\phi(z) & z \in {\mathcal U}_{z_0}\cap \C_{+}
\end{array}
\right.\ .
\ee
We observe that 
\be
\label{exp_w}
w(z) = \frac{1}{\sqrt{2}z_0}(z-z_0)-\dfrac{1}{3\sqrt{2}z_0^2}(z-z_0)^2+\order{(z-z_0)^3} \qquad z \to z_0\ .
\ee
The matrix $E(z)$ is obtained from condition \eqref{match} which, when combined with \eqref{Psiexpansion}, gives
 \begin{equation}
 \label{Pmatching}
 \begin{split}
 &P^{\infty}(z)(P^0(z))^{-1}=P^{\infty}(z){\cal Q}(z)e^{\frac{k}{2}\phi(z)\sigma_3}\Delta(z)^{-1}\Psi(\sqrt{k}w(z))^{-1}E(z)^{-1}\\
 &=P^{\infty}(z)\chi(z)^{-1}\Delta(z)^{-1}(\sqrt{k}w(z))^{-\gamma\sigma_3}\left(I-\frac{\gamma}{2}\begin{bmatrix}
0 & \frac{1}{\beta_{21}}\\
\frac{1}{\beta_{12}} & 0 
\end{bmatrix}\frac{1}{\sqrt{k}w(z)}+
\order{k^{-1}}\right)E^{-1}(z),\quad k\to\infty,\quad z\in\partial{\cal U}_{z_0}
 \end{split}
 \end{equation}
 where we used the fact that ${\cal Q}(z)\to \chi(z)^{-1}$ exponentially fast as  $ k\to\infty,$  and $z\in\partial{\cal U}_{z_0}$.
 From the above expression and \eqref{match},  it turns out that the matrix $E(z)$ takes the form
\be
\label{E}
E(z)=P^{\infty}(z)\chi(z)^{-1}\Delta(z)^{-1}(\sqrt{k}w(z))^{-\gamma\sigma_3}= \left[\frac{(z-z_0)^2}{z(z-1)}\right]^{\frac{\gamma}{2}\sigma_3} 
\Delta(z)^{-1}(\sqrt{k}w(z))^{-\gamma\sigma_3}
\ee
We observe that the function $E(z)$ is single valued in a neighbourhood of ${\cal U}_{z_0}$. Indeed the boundary values of $(z-z_0)^{\gamma\sigma_3}$ and $w^{-\gamma\sigma_3}$ cancel each other. The boundary value of $(z(z-1))^{\gamma}_+=(z(z-1))^{\gamma}_-e^{2\pi i\gamma}$, so that  $\Delta(z)(z(z-1))^{\frac{\gamma}{2}\sigma_3}$ remains single valued in a neighbourhood of $z_0$.

From the expression  \eqref{Pmatching} and \eqref{E} the matching between  $P^0(z)$ and $P^{\infty}(z)$ takes the form
\[
P^{\infty}(z)(P^0(z))^{-1}=E(z)\left(I-\frac{\gamma}{2}\begin{bmatrix}
0 & \frac{1}{\beta_{21}}\\
\frac{1}{\beta_{12}} & 0 
\end{bmatrix}\frac{1}{\sqrt{k}w(z)}+
\order{k^{-1}}\right)E^{-1}(z),\quad k\to\infty,\quad z\in\partial{\cal U}_{z_0},
\]
where $\gamma\in(0,1)$. It is clear from the above expression that the $(2,1)$-entry of the  matrix above is not small as $k\to\infty$. For this reason we need to introduce an improvement of the  parametrix.
\subsection{Improvement of the local parametrix}
In order to have  a uniformly small error for $k\to\infty$ we have to modify the parametrices as in \cite{Bertola_Lee}, \cite{Claeys}, \cite{Claeys_Grava}:
\be
\hat{P}^{\infty}(z):=\left(I+\frac{C}{z-z_0}\right)P^{\infty}(z) \ ,
\ee
where  $C$ is a nilpotent matrix to be determined and 
\be
\hat{P}^{0}(z):=\hat{E}(z)\begin{pmatrix}1 & 0\\ -\frac{\Psi_{1,21}}{\sqrt{k}w(z)}&1\end{pmatrix}\Psi(\sqrt{k}w(z))\Delta(z)e^{-\frac{k}{2}\phi(z)\sigma_3}{\cal Q}(z)^{-1}
\ee
where the matrix $\Psi_1$ has been defined in \eqref{model2} and 
\be
\hat{E}(z)=\left(I+\frac{C}{z-z_0}\right)E(z).
\ee
 With those improved definitions of the parametrices, $\hat{P}^{\infty}(z)$ and  $\hat{P}^{0}(z)$ 
have the same Riemann--Hilbert jump discontinuities as before but they might  have poles at $z=z_0$. Note that $\hat{E}(z)$ also has a pole at 
$z=z_0$. However we can choose $C$ in such a way that $\hat{P}^{0}(z)$ is bounded in $z_0$.
This is accomplished  by
\[
C=-E(z_0)\begin{pmatrix}0& 0\\ -\frac{\Psi_{1,21}}{\sqrt{k}w'(z_0)}&0\end{pmatrix}\left(
E'(z_0)\begin{pmatrix}0 & 0\\ -\frac{\Psi_{1,21}}{\sqrt{k}w'(z_0)}&0\end{pmatrix}+E(z_0)\begin{pmatrix}1 & 0\\ \frac{\Psi_{1,21}w''(z_0)}{2\sqrt{k}(w'k(z_0))^2}&1\end{pmatrix}\right)^{-1}
\]
where 
\be
E(z_0)=\left(\dfrac{2z_0}{k(1-z_0)}\right)^{\frac{\gamma}{2}\sigma_3}e^{\pi i\frac{ \gamma}{2}\sigma_3}\quad \text{ and }\quad E'(z_0)=\gamma\sigma_3\dfrac{4z_0-1}{6z_0(1-z_0)}E(z_0).
\ee
From the above relation the matrix $C$ takes the form
\begin{equation}
\label{C}
C=\begin{pmatrix}0&0\\c k^{\gamma-\frac{1}{2}}&0\end{pmatrix}
\end{equation}
and 
\be
\label{def_c}
c=\left(\dfrac{1-z_0}{2z_0}\right)^{\gamma}e^{-\pi i \gamma}z_0(\Psi_{1})_{21} 2^{\frac{1}{2}}=-\dfrac{\Gamma(1-\gamma)}{\sqrt{2\pi  }}\left(\dfrac{1-z_0}{z_0}\right)^{\gamma}z_0.
\ee

The improved parametrix gives the following matching between $\hat{P}^{\infty}(z)$ and $\hat{P}^0(z)$ as $k\to\infty$ and $z_0\in\partial{\cal U}_{z_0}$ :
\be
\label{hatvR0}
\begin{split}
&\hat{P}^{\infty}(z)(\hat{P}^{0}(z))^{-1}\\&\quad=\hat{E}(z)\left(I-\dfrac{\Psi_1}{\sqrt{k} w(z)}+\dfrac{\Psi_1^2-\Psi_2}{kw(z)^2}+\dfrac{\Psi_2\Psi_1+\Psi_1\Psi_2-\Psi_1^3-\Psi_3}{k^{\frac{3}{2}}w(z)^3}+\order{k^{-2}}\right)\begin{pmatrix}1 & 0\\ -\frac{(\Psi_{1})_{21}}{\sqrt{k}w(z)}&1\end{pmatrix}^{-1}\hat{E}(z)^{-1}\\
&\quad =\hat{E}(z)\left[I-\dfrac{\begin{pmatrix}0&(\Psi_{1})_{12}\\0&0\end{pmatrix}}{\sqrt{k}w(z)}-\dfrac{\begin{pmatrix}(\Psi_{2})_{11}&0\\0&(\Psi_{2})_{22}-\frac{\gamma}{2}\end{pmatrix}}{kw(z)^2}\right.\\
&\quad\quad\quad\quad\quad\quad\quad\quad \left. +\dfrac{\begin{pmatrix}0&-(\Psi_{3})_{12}\\(\Psi_{1})_{21}(\Psi_{2})_{11}-(\Psi_{3})_{21}&0 \end{pmatrix}}{k^{\frac{3}{2}}w(z)^3}+\order{k^{-2}}\right]\hat{E}(z)^{-1}
\end{split},
\ee
which shows that  $\hat{P}^{\infty}(z)(\hat{P}^{0}(z))^{-1}= I+\order{1/k^{\alpha}}$, $\alpha>0$, as $k\to\infty$ when $z_0\in\partial{\cal U}_{z_0}$ 
for $\gamma\in[0,1)$.

\subsubsection{Riemann-Hilbert problem for the error matrix $R$}
\begin{wrapfigure}[7]{rhtb!}{0.5\textwidth}
\includegraphics[width=0.5\textwidth]{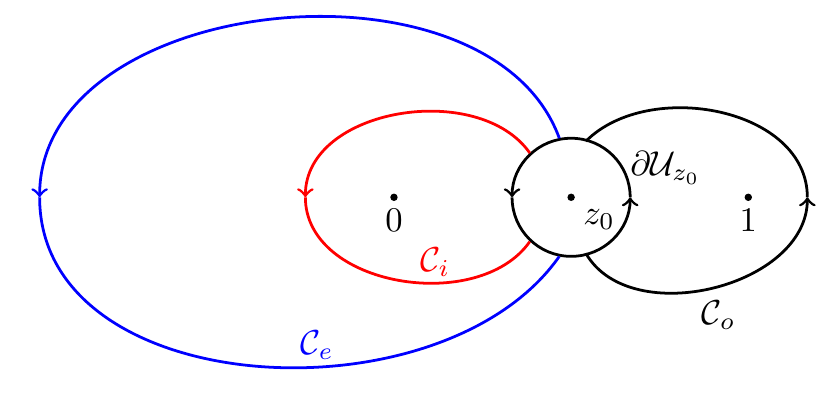}
\caption{The contour $\Sigma_R=\mathcal{C}_i\cup\mathcal{C}_e\cup\mathcal{C}_o\cup\partial{\cal U}_{z_0}$ where $\mathcal{C}_i,\mathcal{C}_e$ and $\mathcal{C}_o$ are defined only in $\mathbb{C}\backslash {\cal U}_{z_0}$.}
\label{figure: R}
\end{wrapfigure}
We now define the error matrix $R$ in two regions of the plane, using our approximations to the matrix  $T$.  Set
\begin{equation}\label{def_R_post}
R(z) =\begin{cases} T(z) \left( \hat{P}^{(0)}(z) \right)^{-1} ,& z \in \mathcal{U}_{z_0}\ ,\\
 T(z) \left( \hat{P}^{\infty}(z) \right)^{-1},& z\in \C\setminus\mathcal{U}_{z_0}.
\end{cases}
\end{equation}
The matrix $R$ is piecewise analytic in $\mathbb{C}$ with a jump across  $\Sigma_R$ (see Fig.~\ref{figure: R}).\\
{\bf Riemann--Hilbert problem for $R$}
\begin{enumerate}
\item $R$ is analytic in $\mathbb C\setminus \Sigma_R$,
\item For $z\in \Sigma_R$, we have 
\be
\label{RHR}
R_+(z)=R_-(z)v_R(z),
\ee
with
\be
v_R(z)=
\left\{\begin{split}
\hat{P}^{\infty}_-(z)v_T(z)\left(\hat{P}_+^{\infty}(z)\right)^{-1} &\quad z\in\Sigma_R\backslash \partial \mathcal{U}_{z_0}\\
\hat{P}^{\infty}(z)\left(\hat{P}^{(0)}(z)\right)^{-1} &\quad z\in\partial  \mathcal{U}_{z_0}
\end{split}\right.
\ee
\item As $z\to\infty$, we have
\begin{equation}
R(z)=I+\order{z^{-1}}.
\end{equation}
\end{enumerate}

The jump matrices across the contour $\Sigma_R \backslash  \partial\mathcal{U}_{z_0}$ are all exponentially close to $I$ for large $k$ because $v_{T}(z)$  converges  exponentially  fast to $v^{\infty}(z)$ defined in \eqref{vinfinity}  for $z\in  \mathbb{C}\backslash {\cal U}_{z_0}$ and the product  
\[\hat{P}^{\infty}_-(z)v^{\infty}(z) \left(\hat{P}_+^{\infty}(z)\right)^{-1}=\left(I+\frac{C}{z-z_0}\right)P^{\infty}_-(z)v^{\infty}(z) \left(P_+^{\infty}(z)\right)^{-1}\left(I+\frac{C}{z-z_0}\right)^{-1}=I
\]  
with $P^{\infty}(z)$ defined in \eqref{eq:global_parametrix_post} and $C$ in \eqref{C}.
The only jump that is not exponentially small is the one on $\partial  \mathcal{U}_{z_0}$.
Indeed one has from \eqref{hatvR0}
\be
\begin{split}
v_R(z)=\hat{P}^{\infty}(z)(\hat{P}^{0}(z))^{-1}&=I+\dfrac{\gamma e(z)^2}{2k^{\frac{1}{2}+\gamma}w(z)\beta_{21}}\begin{pmatrix}
\dfrac{c k^{\gamma-\frac{1}{2}}}{z-z_0}&-1\\\dfrac{c^2k^{2\gamma-1}}{(z-z_0)^2}&-\dfrac{c k^{\gamma-\frac{1}{2}}}{z-z_0}
\end{pmatrix}+\dfrac{\gamma(1-\gamma)}{4kw(z)^2}\begin{pmatrix}1&0\\ \dfrac{2c k^{\gamma-\frac{1}{2}}}{(z-z_0)}&-1\end{pmatrix}\\
&+\dfrac{1}{k^{\frac{3}{2}-\gamma}w(z)^3}
\begin{pmatrix}\order{k^{-\frac{1}{2}-\gamma}}&\order{k^{-2\gamma}}\\
\dfrac{\gamma(1-\gamma)}{4\beta_{12}e(z)^2}+\order{k^{-1}}&
\order{k^{-\frac{1}{2}-\gamma}}\end{pmatrix}+\order{k^{-2},k^{\gamma-\frac{5}{2}}}
\end{split}
\ee
where 
\be
\label{def_e}
e(z)=E_{11}(z)k^{\gamma/2}
\ee
 and we have  substituted the explicit expressions of the matrix $\Psi_1$, $\Psi_2$ and $\Psi_3$ as given by \eqref{Psi123}.
 
We have the following two cases depending on the value of $\gamma\in(0,1)$.
\begin{itemize}
\item[(a)] $0< \gamma< \frac{1}{2}$ 
\begin{equation}
\label{case1a}
v_R(z)=\hat{P}^{\infty}(z)(\hat{P}^{0}(z))^{-1}=I+\dfrac{v_R^1(z)}{ k^{\frac{1}{2}+\gamma}}+\order{k^{-1}},
\end{equation}
where
\be
\label{case1b}
v_R^1(z)=-\dfrac{\gamma e(z)^2}{2w(z)\beta_{21}}\begin{pmatrix}0&1\\0&0\end{pmatrix}
\ee
 \item[(b)] $ \frac{1}{2}\leq \gamma<1$

\be
\label{case3a}
\begin{split}
v_R(z)=\hat{P}^{\infty}(z)(\hat{P}^{0}(z))^{-1}&=I+\dfrac{v_R^{(1)}(z)}{ k^{\frac{3}{2}-\gamma}}+\dfrac{v_R^{(2)}(z)}{ k}+\dfrac{v_R^{(3)}(z)}{ k^{\frac{1}{2}+\gamma}}+\order{k^{\gamma-\frac{5}{2}}}.
\end{split}
\ee
where
\be
\label{case3b}
\begin{split}
v_R^{(1)}(z)&=\left(\dfrac{ c^2 }{\beta_{21}}\dfrac{w(z) e(z)^2 }{(z-z_0)^2 }+\dfrac{c(1-\gamma)}{2w(z)^2 (z-z_0)}+\dfrac{(1-\gamma)}{2\beta_{12}w(z)^3e(z)^2}\right)\begin{pmatrix}0&0\\1&0\end{pmatrix}\\
v_R^{(2)}(z)&=\left(\dfrac{\gamma c}{2\beta_{21}}\dfrac{e(z)^2}{w(z)(z-z_0)}+\dfrac{\gamma(1-\gamma)}{4 w(z)^2}\right)\sigma_3\\
v_R^{(3)}(z)&=-\dfrac{\gamma e(z)^2}{2w(z)\beta_{21}}\begin{pmatrix}0&1\\0&0\end{pmatrix},
\end{split}
\ee
\end{itemize}
By the standard theory of small norm Riemann--Hilbert problems one has a similar  expansion for $R(z)$ in the large $k$ limit.
\subsubsection*{Case (a): $0< \gamma<\frac{1}{2}$}
\begin{equation}
\label{Rexpansion1}
R(z)=I+\dfrac{R^{(1)}}{k^{\frac{1}{2}+\gamma}}+\order{k^{-1}}
\end{equation}
Compatibility of \eqref{RHR}, \eqref{case1a} and \eqref{Rexpansion1}  and the jump condition $R_+=R_-\hat{v}_R$ on $\partial  \mathcal{U}_{z_0}$ gives the following relations.
\be
\label{R1a}
R^{(1)}_+(z)=R^{(1)}_-(z)+v^{(1)}_R(z),\quad z\in \partial  \mathcal{U}_{z_0}
\ee
with $v^{(1)}_R(z)$ defined in \eqref{case1b}.
In addition  $R^{(1)}(z)$ is analytic in $\mathbb{C}\backslash \partial\mathcal{U}_{z_0}$ and $R^{(1)}(\infty)=0$. The unique function that satisfies those conditions is given by
\be
R^{(1)}(z)=\dfrac{1}{2\pi i }\oint\dfrac{v^{(1)}_R(\xi)}{\xi-z}d\xi=\left\{
\begin{array}{ll}
\dfrac{1}{z-z_0}
\dfrac{\gamma e(z_0)^2}{2w'(z_0)\beta_{21}}\begin{pmatrix}0&1\\0&0\end{pmatrix}
& z\in\mathbb{C}\,\backslash\, \mathcal{U}_{z_0}\\
&\\
v^{(1)}_R(z)+\dfrac{1}{z-z_0}
\dfrac{\gamma e(z_0)^2}{2w'(z_0)\beta_{21}}\begin{pmatrix}0&1\\0&0\end{pmatrix}
& z\in \mathcal{U}_{z_0},
\end{array}\right.
\ee
where the integral is taken along $\partial {\cal U}_{z_0}$.
Using the definition of   $\beta_{21}$ , $w'(z)$ and $e(z)$  given in \eqref{beta}, \eqref{exp_w} and \eqref{def_e} respectively one has

\be
R^{(1)}(z)=\left\{
\begin{array}{ll}
\dfrac{1}{z-z_0}\dfrac{\gamma z_0^2}{c}
\begin{pmatrix}0&1\\0&0\end{pmatrix}
& z\in\mathbb{C}\,\backslash\, \mathcal{U}_{z_0}\\
&\\
v^{(1)}_R(z)+\dfrac{1}{z-z_0}\dfrac{\gamma z_0^2}{c}\begin{pmatrix}0&1\\0&0\end{pmatrix}
& z\in \mathcal{U}_{z_0},
\end{array}\right.
\ee
with the constant $c$ defined in \eqref{def_c}.

\subsubsection*{Case (b): $\frac{1}{2}\leq \gamma<1$}
\begin{equation}
\label{Rexpansion3}
R(z)=I+\dfrac{R^{(1)}}{ k^{\frac{3}{2}-\gamma}}+\dfrac{R^{(2)}}{ k}+\dfrac{R^{(3)}}{ k^{\frac{1}{2}+\gamma}}+\order{k^{\gamma-\frac{5}{2}}}
\end{equation}
Compatibility of \eqref{RHR}, \eqref{case3a} and \eqref{Rexpansion3}  and the jump condition $R_+=R_-\hat{v}_R$ on $\partial  \mathcal{U}_{z_0}$ gives the following relations.
\be
\label{R1b}
R^{(i)}_+(z)=R^{(i)}_-(z)+v^{(i)}_R(z),\quad z\in \partial  \mathcal{U}_{z_0}
\ee
with $v^{(i)}_R(z)$ defined in \eqref{case3b}.
In addition  $R^{(i)}(z)$ is analytic in $\mathbb{C}\backslash\partial  \mathcal{U}_{z_0}$ and $R^{(i)}(\infty)=0$.  

Since $v^{(1)}_R(z)$
has a third order pole at $z=z_0$,  the unique function that satisfies those conditions is given by
\be
R^{(1)}(z)=-\dfrac{\resi[\lambda=z_0] v^{(1)}_R(\lambda)}{z-z_0}-\dfrac{\resi[\lambda=z_0](\lambda-z_0) v^{(1)}_R(\lambda)}{(z-z_0)^2}-\dfrac{\resi[\lambda=z_0](\lambda-z_0)^2 v^{(1)}_R(\lambda)}{(z-z_0)^3}=\begin{pmatrix}0&0\\ \ast&0\end{pmatrix}\quad z\in\mathbb{C}\,\backslash\, \mathcal{U}_{z_0}.
\ee
Given the structure of the matrix $\hat{P}^0(z)$, the matrix  $R^{(1)}(z)$ does not give any relevant contribution to the orthogonal polynomials $\pi_k(z)$.

Regarding $R^{(2)}(z)$ one has
\be
\begin{split}
R^{(2)}(z)&=-\dfrac{\resi[\lambda=z_0] v^{(2)}_R(\lambda)}{z-z_0}-\dfrac{\resi[\lambda=z_0](\lambda-z_0) v^{(2)}_R(\lambda)}{(z-z_0)^2}\\
&=-\dfrac{1}{z-z_0} \dfrac{c\gamma e(z_0)^2}{2\beta_{21}w'(z_0)}\left(\dfrac{2e'(z_0)}{e(z_0)}-\dfrac{w''(z_0)}{2w'(z_0)}+\dfrac{1}{z-z_0}\right)\sigma_3\\
&+\dfrac{1}{z-z_0} \frac{\gamma(1-\gamma)}{4w'(z_0)^2}\left(\dfrac{w''(z_0)}{w'(z_0)}-\dfrac{1}{z-z_0}\right)\sigma_3\quad z\in\mathbb{C}\,\backslash\, \mathcal{U}_{z_0},
\end{split}
\ee
so that, using  \eqref{beta}, \eqref{exp_w}, \eqref{def_c} and \eqref{def_e}  one obtains
\be
R^{(2)}(z)=\left\{
\begin{array}{ll}
-\dfrac{\gamma z_0^2}{z-z_0} \left(\dfrac{\gamma (4z_0-1)}{3z_0(1-z_0)}+\dfrac{2-\gamma}{3z_0}+\dfrac{3-\gamma}{2(z-z_0)}\right)\sigma_3,& z\in\mathbb{C}\,\backslash\, \mathcal{U}_{z_0}\\
&\\
v^{(2)}_R(z)-\dfrac{\gamma z_0^2}{z-z_0} \left(\dfrac{\gamma (4z_0-1)}{3z_0(1-z_0)}+\dfrac{2-\gamma}{3z_0}+\dfrac{3-\gamma}{2(z-z_0)}\right)\sigma_3,& z\in \mathcal{U}_{z_0}.
\end{array}
\right.
\ee
In a similar way for $R^{(3)}(z)$ we obtain
\be
R^{(3)}(z)=-\dfrac{\resi[\lambda=z_0] v^{(3)}_R(\lambda)}{z-z_0}=\dfrac{1}{z-z_0}
\dfrac{\gamma e(z_0)^2}{2w'(z_0)\beta_{21}}\begin{pmatrix}0&1\\0&0\end{pmatrix}\quad z\in\mathbb{C}\,\backslash\, \mathcal{U}_{z_0},
\ee
so that 
\be
R^{(3)}(z)=\left\{
\begin{array}{ll}
\dfrac{1}{z-z_0}\dfrac{\gamma z_0^2}{c}
\begin{pmatrix}0&1\\0&0\end{pmatrix}
& z\in\mathbb{C}\,\backslash\, \mathcal{U}_{z_0}\\
&\\
v^{(3)}_R(z)+\dfrac{1}{z-z_0}\dfrac{\gamma z_0^2}{c}\begin{pmatrix}0&1\\0&0\end{pmatrix}
& z\in \mathcal{U}_{z_0},
\end{array}\right.
\ee
with the constant $c$ defined in \eqref{def_c}.

\subsubsection{Proof of Theorem~\ref{theorem3}: asymptotics for $p_n(\lb)$ for $0<z_0<1$}
In order to obtain the asymptotic  expansion of the polynomials $p_n(\lb)$ for $n\to\infty$,  $NT=n-l$ and $0<z_0<1$,  we first derive the asymptotic expansion of the reduced polynomials  $\pi_k(z)$ as $k\to\infty$.

Using the relation \eqref{pik_T} and the relation \eqref{def_R_post} one obtains
\begin{equation}
\label{exp_post_poly}
\poly_k(z) =
e^{kg(z)}\left(1-\frac{1}{z}\right)^{\frac{\gamma}{2}}\left\{
\begin{array}{ll}
\displaystyle \left[R(z)\hat{P}^{\infty}(z)\right]_{11} & z \in (\Omega_{\infty}\cup \Omega_{0}\cup \Omega_{3})\setminus \mathcal{U}_{z_0}\\
&\\
\displaystyle \left[R(z)\hat{P}^{\infty}(z)\begin{pmatrix}1& 0 \\ e^{k\phi(z)} & 1 \end{pmatrix}\right]_{11} & z \in \Omega_{1}\setminus \mathcal{U}_{z_0}\\
&\\
\displaystyle \left[R(z)\hat{P}^{\infty}(z)\begin{pmatrix}1& 0 \\ -e^{-k\phi(z)} & 1 \end{pmatrix}\right]_{11} & z \in \Omega_{2}\setminus \mathcal{U}_{z_0}\\
&\\
\displaystyle \left[R(z)\hat{P}^{0}(z)\right]_{11} & z \in (\Omega_{\infty}\cup \Omega_{0}\cup \Omega_{3})\cap \mathcal{U}_{z_0}\\
&\\
\displaystyle \left[R(z)\hat{P}^{0}(z)\begin{pmatrix}1& 0 \\ e^{k\phi(z)} & 1 \end{pmatrix}\right]_{11} & z \in \Omega_{1}\cap \mathcal{U}_{z_0}\\
&\\
\displaystyle \left[R(z)\hat{P}^{0}(z)\begin{pmatrix}1& 0 \\ -e^{-k\phi(z)} & 1 \end{pmatrix}\right]_{11} & z \in \Omega_{2}\cap \mathcal{U}_{z_0}\ .
\end{array}
\right.
\end{equation}

\subsubsection*{The region $(\Omega_{\infty}\cup \Omega_{3})\setminus \mathcal{U}_{z_0}$}
\be
\poly_{k}(z)=e^{kg(z)}\left(1-\frac{1}{z}\right)^{\frac{\gamma}{2}}\left[\frac{(z-z_0)^2}{z(z-1)}\right]^{\frac{\gamma}{2}}\left(1+\order{\frac{1}{k}}\right)=e^{kg(z)}\left(\frac{z-z_0}{z}\right)^{\gamma}\left(1+\order{\frac{1}{k}}\right).
\ee

\subsubsection*{The region $\Omega_{0}\setminus \mathcal{U}_{z_0}$}
\be
\poly_{k}(z)=\frac{1}{k^{\frac{1}{2}+\gamma}}e^{kg(z)}
\left(
\dfrac{\gamma z_0^2}{c}\frac{(z-1)^{\gamma}}{(z-z_0)^{\gamma+1}}+\order{\frac{1}{k}}\right).
\ee
with $c$ defined in \eqref{def_c}.

\subsubsection*{The region $\Omega_{1}\setminus \mathcal{U}_{z_0}$}
\be
\begin{split}
\poly_{k}(z)&=e^{kg(z)}\left(\frac{z-z_0}{z}\right)^{\gamma}
\left(e^{k\phi(z)}-\dfrac{1}{z-z_0}
\dfrac{\gamma z_0^2}{c}
\left[\frac{(z-z_0)^2}{z(z-1)}\right]^{-\gamma}\frac{1}{k^{\frac{1}{2}+\gamma}}\right.\\
&+\left\{
\begin{array}{ll}
\left.\order{\dfrac{1}{k}}\right),& 0<\gamma<\frac{1}{2}\\
&\\
\left.-\dfrac{e^{k\phi(z)}\gamma z_0^2}{z-z_0} \dfrac{1}{k}\left(\dfrac{\gamma (4z_0-1)}{3z_0(1-z_0)}+\dfrac{2-\gamma}{3z_0}+\dfrac{3-\gamma}{2(z-z_0)}\right)
+\order{\dfrac{1}{k^{\frac{5}{2}-\gamma}}}\right),&\frac{1}{2}\leq \gamma<1
\end{array}
\right.
\end{split}
\ee
with $c$ defined in \eqref{def_c} and where we observe that $\Re\phi(z)\leq 0 $ in $\Omega_1$.  In a similar way we can obtain the expansion in the region $\Omega_{2}\setminus \mathcal{U}_{z_0}$.
%
%

\subsubsection*{The region $\Omega_{2}\setminus \mathcal{U}_{z_0}$}

\be
\begin{split}
\poly_{k}(z)&=e^{kg(z)}\left(\frac{z-z_0}{z}\right)^{\gamma}
\left(1-\dfrac{e^{-k\phi(z)}}{k^{\frac{1}{2}+\gamma}}
\dfrac{\gamma z_0^2}{c}
\left[\frac{(z-z_0)^2}{z(z-1)}\right]^{-\gamma}\frac{1}{z-z_0}\right.\\
&+\left\{
\begin{array}{ll}
\left.\order{\dfrac{1}{k}}\right),& 0<\gamma<\frac{1}{2}\\
&\\
\left.-\dfrac{\gamma z_0^2}{z-z_0}\dfrac{1}{k} \left(\dfrac{\gamma (4z_0-1)}{3z_0(1-z_0)}+\dfrac{2-\gamma}{3z_0}+\dfrac{3-\gamma}{2(z-z_0)}\right)
+\order{\dfrac{1}{k^{\frac{5}{2}-\gamma}}}\right),&\frac{1}{2}\leq \gamma<1
\end{array}
\right.
\end{split}
\ee
where now we observe that $\Re\phi(z)\geq 0$ in $\Omega_2$.
\subsubsection*{The region $ \mathcal{U}_{z_0}$}
Using the relations \eqref{exp_post_poly}, \eqref{model_problem}, \eqref{rel1}  and \eqref{rel2} one obtains
\[
\pi_k(z)=e^{kg(z)}\left(\frac{z-z_0}{z}\right)^{\gamma}\dfrac{e^{-k\phi(z)/2}}{(\sqrt{k}w(z))^{\gamma}}\left( 
{\cal U}(-\gamma-\frac{1}{2}; \sqrt{2k}w(z))+\order{\dfrac{1}{k^{\frac{1}{2}} }}\right),\quad z\in \mathcal{U}_{z_0}\cap \exte(\mathcal{C})
\]
where $\U$ is the parabolic cylinder function that solves equation \eqref{parabolic}. For  $z\in \mathcal{U}_{z_0}\cap \inte(\mathcal{C})$
we have 
\[
\pi_k(z)=e^{kg(z)}\left(\frac{z-z_0}{z}\right)^{\gamma}\dfrac{e^{k\phi(z)/2}}{(\sqrt{k}w(z))^{\gamma}}\left( 
{\cal U}(-\gamma-\frac{1}{2}; \sqrt{2k}w(z))+\order{\dfrac{1}{k^{\frac{1}{2}} }}\right),\quad z\in \mathcal{U}_{z_0}\cap \inte(\mathcal{C})
\]

In order to prove Theorem~\ref{theorem2} it is sufficient to use the above expansions,   make the change of coordinates $\lb^s=-t(z-1)$ and use  for $n=ks+l$ the relation
\[
p_n(\lb)=(-t)^k\lb^l\pi_k(1-\lb^s/t).
\]
\hfill {\bf Q.E.D.}
%
%
%
%
%
%
%
%

\begin{proposition}
\label{propo_zeros_post}
The support of the counting measure of the zeros of the polynomials $\pi_k(z)$ for $0<t<t_c$  outside an arbitrary small disk ${\cal U}_{z_0}$ surrounding the point $z=z_0$  tends uniformly to the curve $\mathcal{C}$ defined in \eqref{Gamma0}.
The zeros are within a distance  $o(1/k)$ from the curve defined by
\be
\label{phi_deform_post}
\Re\phi(z)=-\left(\frac{1}{2}+\gamma \right)\frac{\log(k)}{k}+\dfrac{1}{k}\log\left(\dfrac{\gamma z_0^2}{c}\left|\frac{(z-z_0)^2}{z(z-1)}\right|^{-\gamma}\right)
\ee
where the function $\phi(z)$ has been defined in \eqref{def_phi}.
Such curves tends to $\mathcal{C}$ at a rate  $\order{\log k /k}$.  The normalised counting measure of the zeros of $\pi_k(z)$ converges  to the probability measure $\nu$ defined in \eqref{dnu}.
\end{proposition}
\begin{proof}
 As in the proof of Proposition~\ref{propo_zeros_pre}, it is clear from the above expansions of the polynomials $\pi_k(z)$ that there are no zeros in the region $(\Omega_{\infty}\cup \Omega_{3}\cup\Omega_0)\setminus \mathcal{U}_{z_0}$ for sufficiently large $k$. Then we observe that  the asymptotic expansion  of the polynomials $\pi_k(z)$ in the regions $\Omega_1\cup\Omega_{2}\setminus \mathcal{U}_{z_0}$ takes the form
\be
\begin{split}
\poly_{k}(z)&=z^{k}\left(\frac{z-z_0}{z}\right)^{\gamma}
\left(1-\dfrac{e^{-k\phi(z)}}{k^{\frac{1}{2}+\gamma}}
\dfrac{\gamma z_0^2}{c}
\left[\frac{(z-z_0)^2}{z(z-1)}\right]^{-\gamma}\frac{1}{z-z_0}\right.\\
&+\left\{
\begin{array}{ll}
\left.\order{\dfrac{1}{k}}\right),& 0<\gamma<\frac{1}{2}\\
&\\
\left.-\dfrac{\gamma z_0^2}{z-z_0}\dfrac{1}{k} \left(\dfrac{\gamma (4z_0-1)}{3z_0(1-z_0)}+\dfrac{2-\gamma}{3z_0}+\dfrac{3-\gamma}{2(z-z_0)}\right)
+\order{\dfrac{1}{k^{\frac{5}{2}-\gamma}}}\right),&\frac{1}{2}\leq \gamma<1
\end{array}
\right.
\end{split}
\ee
so that we conclude that the zeros of $\pi_k(z)$ occur in the region where 
\[
1-\dfrac{e^{-k\phi(z)}}{k^{\frac{1}{2}+\gamma}}
\dfrac{\gamma z_0^2}{c}
\left[\frac{(z-z_0)^2}{z(z-1)}\right]^{-\gamma}\frac{1}{z-z_0}
\]
is equal to zero.

Since $\Omega_2 \subset \{\Re(\phi) \geq 0\}$ and $\Omega_1 \subset \{\Re(\phi) \leq 0\}$, it follows that the zeros of $\pi_k(z)$ may lie only in the region $\Omega_1$  and such that 
  $\Re\phi(z)=\order{\log k/ k}$. Namely the zeros of the polynomials $\pi_k(z)$ lie on the   curve given by \eqref{phi_deform_post}  with an error of order  $\order{1/k^2}$. Such curves  converges to the curve $\mathcal{C}$ defined \eqref{Gamma0} at a rate 
$\order{\log k/k}$ (see Figure~\ref{zero_pre}).
\begin{figure}[H]
\centering
\includegraphics[scale=0.28]{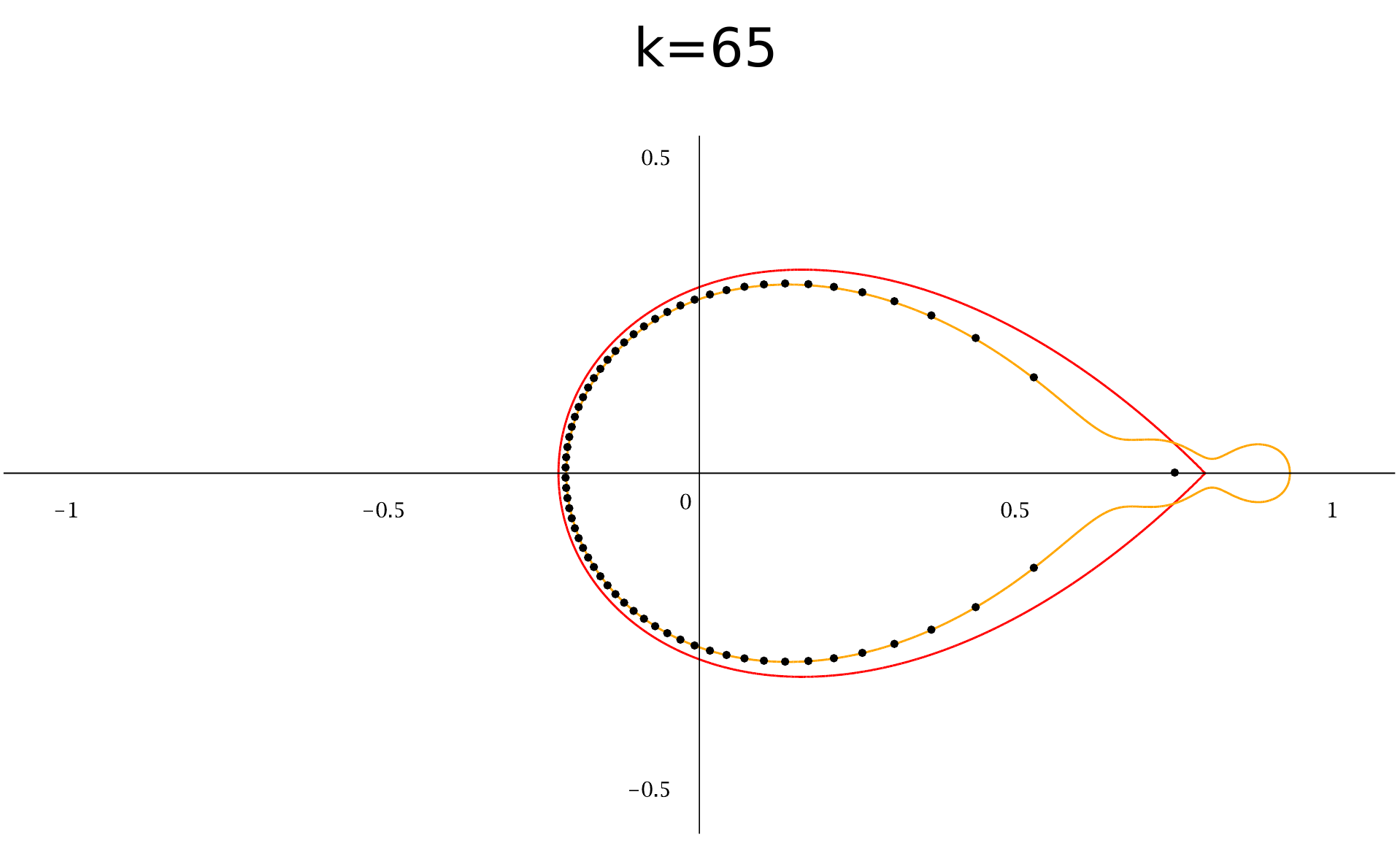}\includegraphics[scale=0.28]{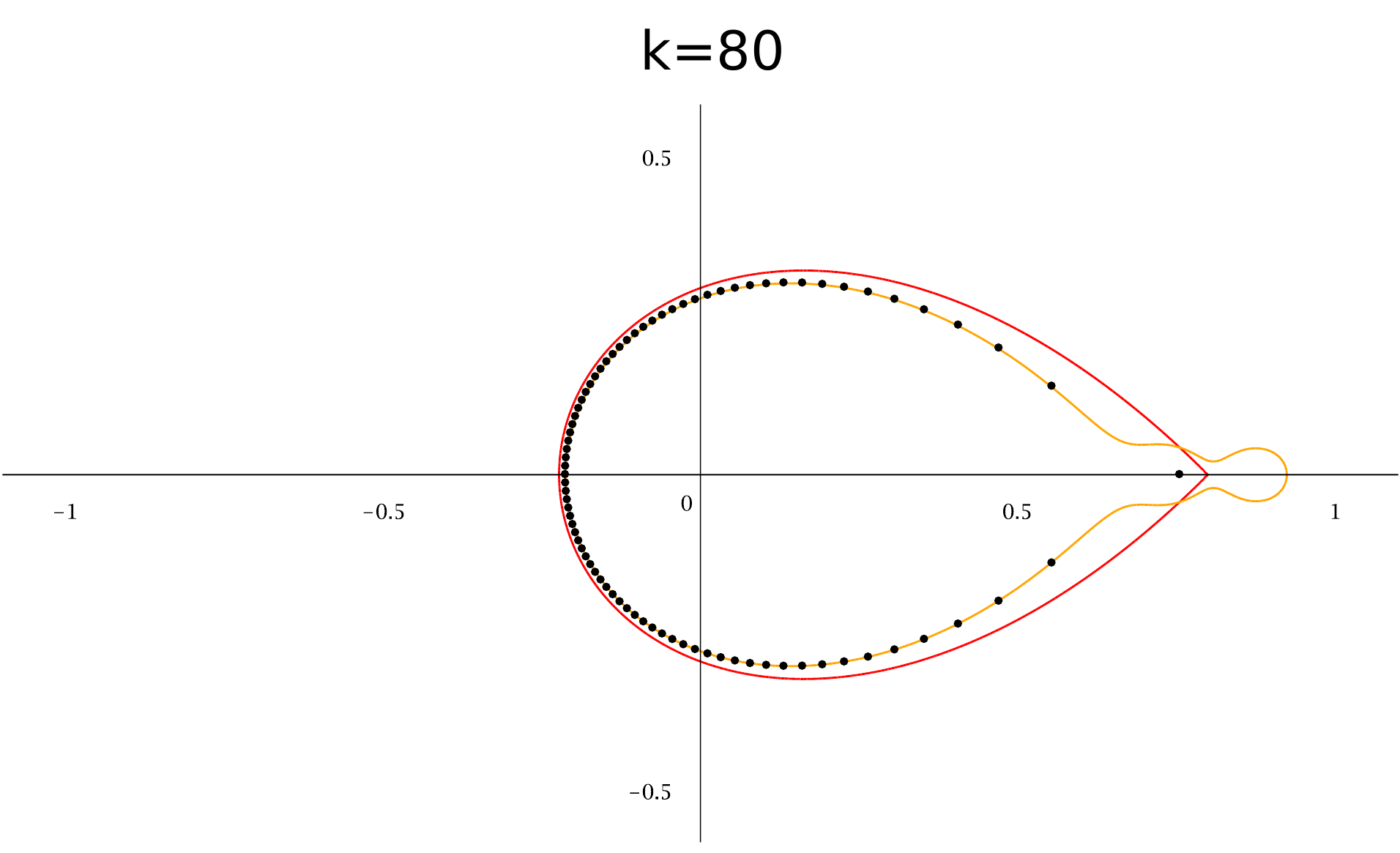}\includegraphics[scale=0.28, trim=0cm -.29cm 0cm 0cm]{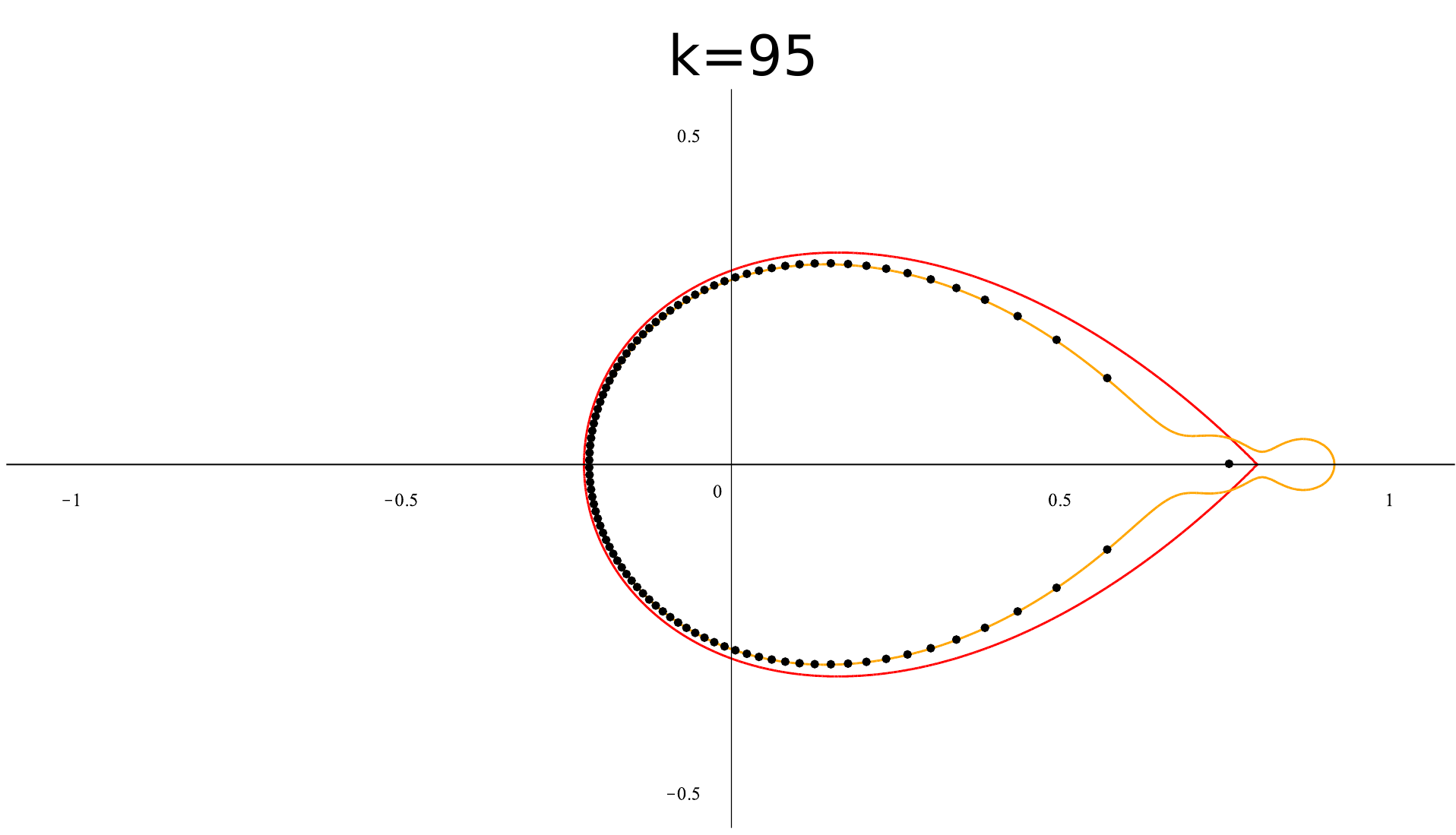}\\
\caption{The zeros of $\pi_k(z)$ for $s=3$, $l=0$, $t>t_c$  and $k=65, 80, 95$. The  red contour is $\mathcal{C}$ while the   yellow contour is the curve  \protect\eqref{phi_deform_post}.}
\label{zero_pre}
\end{figure}
The proof of the remaining points of Proposition~\ref{propo_zeros_post} follows the lines of the proof of Proposition~\ref{propo_zeros_pre}.
\end{proof}

\section{Proof of Theorem~\ref{theorem1}}
From the comments after the statement of Theorem~\ref{theorem2} and Theorem~\ref{theorem3} it is clear  from \eqref{phi_modified} and \eqref{phi_modified_post} 
that the zeros of the polynomials $p_n(\lb)$ accumulate in the limit $n=ks+l\to\infty$, with $s$ and $l$ given and $l=0,\dots,s-2$,  along the curve $\hat{\mathcal{C}}$ defined in \eqref{Gamma}.
In order to show that the measure $\hat{\nu}$ is the weak-star limit of the zero density $\nu_n$ we will show that 
\[
\lim_{n\to\infty}\dfrac{1}{n}\log p_n(\lb)=\int_{\hat{\mathcal{C}}}\log(\lb-\xi)d\nu(\xi),
\]
for $\lb$ in compact subsets of  the exterior of $\hat{\mathcal{C}}$ (namely the unbounded component of  $\mathbb{C}\setminus \hat{\mathcal{C}}$).
Indeed using the relation 
\[p_{n}(\lambda)=(-t)^{k}\lambda^{l}\pi_{k}(z(\lambda))\] 
one has
\be
\begin{aligned}
\lim_{n\to\infty}\frac{1}{n}\log p_{n}(\lambda)&= \lim_{k\to\infty}\frac{1}{ks+l}\log(\lambda^{l} (-t)^{k}\pi_{k}(z(\lambda)))=\frac{1}{s}\log(-t)+\frac{1}{s}\lim_{k\to\infty}\frac{1}{k}\log\pi_{k}(z(\lambda))\ .
\end{aligned}
\ee
Using Proposition~\ref{propo_zeros_pre} and the relation $d\nu(z(\lambda))=s\, d\hat\nu(\lambda)$ we get
\be
\begin{aligned}
\frac{1}{s}\log(-t)+\frac{1}{s}\lim_{k\to\infty}\frac{1}{k}\log \pi_{k}(z(\lambda))&=\frac{1}{s}\log(-t)+\frac{1}{s}\int_{\mathcal{C}}\log(z-\xi)d\nu(\xi) \\
&= \frac{1}{s}\log(-t)+\int_{\hat{\mathcal{C}}^{0}}\log\frac{\lambda^{s}-\sigma^{s}}{-t}d\hat{\nu}(\sigma)\ , \\
\end{aligned}
\ee
where $\hat{\mathcal{C}}^{j}$ ($j=0,\dots,s-1$) are the components of $\hat{\mathcal{C}}$ (as defined in the proof of Lemma~\ref{lemma1}). Since on each $\hat{\mathcal{C}}^{j}$ we have that $d\hat\nu$ is normalized to $\frac{1}{s}$ we obtain
\be
\begin{aligned}
\lim_{n\to\infty}\frac{1}{n}\log p_{n}(\lambda)&=\int_{\hat{\mathcal{C}}^{0}}\log(\lambda^{s}-\sigma^{s})d\hat{\nu}(\sigma)
=\sum_{j=0}^{s-1}\int_{\hat{\mathcal{C}}^{0}}\log(\lambda-\sigma\omega^{j})d\hat{\nu}(\sigma)\\
&=\sum_{j=0}^{s-1}\int_{\hat{\mathcal{C}}^{j}}\log(\lambda-\sigma)d\hat{\nu}(\sigma)=\int_{\hat{\mathcal{C}}}\log(\lambda-\sigma)d\hat{\nu}(\sigma)\ ,
\end{aligned}
\ee
with $\omega=e^{\frac{2\pi i}{s}}$ and where in the last steps we used the symmetry of $d\hat\nu$. Hence we have obtained the relation
\be
\lim_{n\to\infty}\frac{1}{n}\log p_{n}(\lambda) =\int_{\hat{\mathcal{C}}}\log(\lambda-\sigma)d\hat{\nu}(\sigma)\,,
\ee
uniformly for $\lb$ in compact subsets of unbounded component of  $\mathbb{C}\backslash \hat{\mathcal{C}}$. 
  Furthermore $\hat{\mathcal{C}}$ is the boundary of its polynomial convex hull. Then it follows  that the measure  $\hat{\nu}$  is the weak-star limit of the zero counting  measure $ \nu_n$ of the   polynomials $p_n(\lb)$ for $n=kd+l$, $l=0,\dots,s-2$ (see \cite{MhaskarSaff} and \cite{SaffTotik} Chapter 3).
The proof of Theorem~\ref{theorem1} is then completed.

\end{document}